\PassOptionsToPackage{table,svgnames}{xcolor}
\PassOptionsToPackage{greek,english}{babel}
\PassOptionsToPackage{scaled=.8}{beramono}

\documentclass[acmsmall,screen]{acmart}

\setcopyright{acmlicensed}
\copyrightyear{2026}
\acmYear{2026}
\acmDOI{XXXXXXX.XXXXXXX}
\acmJournal{TOPLAS}
\acmISBN{978-1-4503-XXXX-X/18/06} 

\usepackage{mathtools}

\usepackage{float}
\usepackage[noteubner, frozen]{showcode}
\usepackage{amsthm}
\usepackage{amsmath}
\usepackage{stmaryrd}
\usepackage{xspace}
\usepackage{paralist}
\usepackage{algorithm}
\usepackage{algpseudocode}
\usepackage{setspace}
\usepackage{wrapfig}
\usepackage{cleveref}

\usepackage{tcolorbox}
\colorlet{baseboxcolor}{green}
\colorlet{shadeboxcolor}{yellow}
\newtcolorbox{attributebox}[2][]{
    enhanced,
    breakable,
    colbacktitle=baseboxcolor!25!white,
    coltitle=baseboxcolor!10!black,
    colframe=baseboxcolor!35!black,
    boxrule=0.4pt,
    interior style={top color=shadeboxcolor!10!white,bottom color=baseboxcolor!10!white},
    title={\footnotesize #2},
    attach boxed title to top*,
    fontupper=\footnotesize,
    left=2pt,
    right=2pt,
    top=2pt,
    bottom=2pt,
    #1
}

\colorlet{baseviolationcolor}{red}
\colorlet{shadeviolationcolor}{yellow}
\newtcolorbox{violationbox}[2][]{
    enhanced,
    breakable,
    colbacktitle=baseviolationcolor!25!white,
    coltitle=baseviolationcolor!10!black,
    colframe=baseviolationcolor!35!black,
    boxrule=0.4pt,
    interior style={top color=shadeviolationcolor!10!white,bottom color=baseviolationcolor!10!white},
    title={\footnotesize #2},
    attach boxed title to top*,
    fontupper=\footnotesize,
    left=2pt,
    right=2pt,
    top=2pt,
    bottom=2pt,
    #1
}

\usepackage{bm}

\newcommand{\sect}[1]{\S\ref{#1}}

\usepackage{hyphenat}

\usepackage{url}
\usepackage{pifont}
\usepackage{centernot}
\usepackage[inline,shortlabels]{enumitem}
\usepackage{doi,url}
\usepackage{xcolor}

\newcommand{\tool}[1]{\texttt{#1}}
\newcommand{\nlgcheck}[0]{\tool{nlgcheck}}
\newcommand{\mypar}[1]{\vspace{10pt}\noindent\textbf{#1.}\quad}

\usepackage{tikz-cd}
\tikzstyle{every picture}+=[remember picture]
\tikzstyle{na} = [baseline=-2.5pt]
\usetikzlibrary{shapes.arrows}
\usetikzlibrary{positioning}

\newtheorem{property}{Lemma} 
\newtheorem{mainthm}{Theorem}
\crefname{mainthm}{Theorem}{Theorems}
\Crefname{mainthm}{Theorem}{Theorems}

\newenvironment{propertyproof}[1][Proof]{%
  \begin{proof}[#1]%
}{%
  \end{proof}%
}

\usepackage{soul}

\begin{document}

\title{Statically Verifying Well-Definedness of Separately Compiled Attribute Grammars}

\author{Federico Bruzzone}
\email{federico.bruzzone@unimi.it}
\orcid{0009-0004-6086-8810} 
\affiliation{%
   \institution{Universit\`a degli Studi di Milano}
   \department{Computer Science Department}
   \city{Milan}
   \state{Italy}
   \country{Italy}
}
\author{Walter Cazzola}
\authornote{Corresponding author.}
\authornotemark[0]
\email{cazzola@di.unimi.it}
\orcid{0000-0002-4652-8113} 
\email{cazzola@di.unimi.it}
\affiliation{%
   \institution{Universit\`a degli Studi di Milano}
   \department{Computer Science Department}
   \city{Milan}
   \state{Italy}
   \country{Italy}
}
\author{Luca Favalli}
\email{luca.favalli@unimi.it}
\orcid{0000-0001-7452-2440} 
\affiliation{%
   \institution{Universit\`a degli Studi di Milano}
   \department{Computer Science Department}
   \city{Milan}
   \state{Italy}
   \country{Italy}
}

\begin{abstract}
    The development of programming languages involves complex theoretical and practical challenges, particularly when addressing modularity and reusability through language extensions. While language workbenches aim to enable modular development under the constraints of the language extension problem, one critical constraint—separate compilation—is often relaxed due to its complexity. However, this relaxation undermines artifact reusability and integration with common dependency systems. A key difficulty under separate compilation arises from managing attribute grammars, as extensions may introduce new attributes that invalidate previously generated abstract syntax tree structures. Existing approaches, such as the use of dynamic maps in the \tool{Neverlang} workbench, favor flexibility at the cost of compile-time correctness, leading to potential runtime errors due to undefined attributes.
    This work addresses this issue by introducing \nlgcheck{}, a theoretically sound static analysis tool based on data-flow analysis for the \tool{Neverlang} language workbench, which detects potential runtime errors—such as undefined attribute accesses—when separately compiled language components are composed, before the resulting language is ever run, thus preserving separate compilation while maintaining strong static correctness guarantees. Experimental evaluation using mutation testing on \tool{Neverlang}-based projects demonstrates that \nlgcheck{} effectively enhances robustness without sacrificing modularity or flexibility and with a level of performance that does not impede its adoption in daily development activities.
\end{abstract}




\maketitle

\section{Introduction}\label{sec:intro}
Development of programming languages is a complex activity that requires mastery of both theoretical and practical concepts spanning aspects such as grammars, parsers, optimization techniques, architectures, and software engineering. This is true for both \textit{general-purpose languages} (GPL) and \textit{domain-specific languages} (DSL) although with varying degrees of complexity and different challenges: where GPLs are usually larger and more complex projects overall, DSLs require deep knowledge of the application domain and often involve more interaction with domain experts that may or may not coincide with the end users~\cite{Mernik05}.

Although each programming language is unique and has been designed for a specific purpose, it has been observed that many provide similar language constructs among them, due to recurring modeling patterns and abstractions~\cite{Mendez-Acuna16, Cazzola16c}.
In this context, the challenge for language designers is to take advantage of such commonalities by reusing any existing software artifacts as much as possible. In this regard, the \textit{language extension problem}~\cite{Leduc20} is a popular variant of the even more popular \textit{expression problem}~\cite{Wadler98}, modified to suit the needs and constraints of programming languages and their extension. Among the tools for language implementation, \textit{language workbenches}~\cite{Fowler05}, such as \tool{Neverlang}~\cite{Cazzola12c,Cazzola13e,Cazzola15c}, strive to enable purely modular development and composition techniques by adhering to the constraints of the language extension problem and by providing the means to perform language composition at several levels of granularity~\cite{Erdweg12, Cazzola23b}.

Despite the efforts made to achieve this result, one among the constraints of the language extension problem is often (if not always) disregarded. In the authors' own words~\cite{Leduc20}: «the constraint of \textit{separate compilation} usually impacts other non-functional properties (\dots). Consequently, it can be worthwhile to relax the separate compilation constraint in order to comply with other non-functional properties». This reflects a common tension: joint (\textit{i.e.}, non-separate) compilation can enable whole-program optimizations and simplifies certain correctness guarantees, since the compiler has access to the implementation of every language component at once, at the cost of the modularity and independent shippability that separate compilation provides. While the increase in complexity related to separate compilation is undeniable, we believe that relaxing this constraint can severely hinder the actual reusability of language artifacts. For instance, without separate compilation, language extension is only possible when source code is available, whereas some among the most popular repositories for dependency management, such as \tool{Maven} Central\footnote{\url{https://central.sonatype.com}} and \tool{NuGet},\footnote{\url{https://www.nuget.org}} distribute libraries as compiled artifacts, thus limiting the possibility to reuse language libraries from the web.

\textit{Attribute grammars} (AGs)~\cite{Knuth68} are a well-established formalism for specifying the semantics of programming languages. An AG extends a context-free grammar by attaching \textit{attributes} to its symbols and \textit{semantic equations} to its productions: each equation defines the value of an attribute in terms of other attributes occurring in the same production. Attributes whose values are computed from the children of a node, and thus flow upward in the tree, are called \textit{synthesized}. Attributes whose values are computed from the parent or the siblings of a node, and thus flow downward or sideways, are called \textit{inherited} (\sect{sec:attribute-grammars}). Compilers use AGs to implement the analyses and translations performed on the \textit{abstract syntax tree} (AST) of a program, such as name analysis, type checking, and code generation~\cite{Aho06, Paakki95}.
Among the complexities brought by separate compilation, we believe the most pressing one to be the implementation of AGs. In most language workbenches, the compiler generates data structures for the AST nodes that may appear in the language under development.\footnote{For simplicity, we describe this with respect to a single, unified AST representation; in practice, a compiler is typically organized into several phases, each with its own intermediate representation, but the argument that follows applies equally to each of them.} Then, semantic actions~\cite{Aho06} can be associated with each node to implement the semantics. However, to properly create the data structure for each AST node, it is necessary to foresee all the grammar attributes that may be inherited or synthesized by that node, as each attribute will need a corresponding field within that data structure. It is apparent that when separate compilation is permitted, a language construct may be extended with more semantic actions at a later time, possibly adding more attributes, thus rendering the AST data structures obsolete.

Some language workbenches with separate compilation, such as \tool{Neverlang}, simply ignore this problem. In \tool{Neverlang}, semantic actions are imperative \tool{Java} code written directly by the language developer, and each module is compiled on its own. No compilation step ever sees the full set of attributes written and read by the composed language, so \tool{Neverlang} implements grammar attributes with generic and dynamic maps instead of generating specific data structures for its AST\@. While this solution greatly improves the composability among artifacts thanks to loose coupling, the absence of a complete specification at compile time denies any compile-time composition correctness guarantee. This reduces the robustness of programming languages developed this way and increases the likelihood of defects in the composition to go unnoticed and become evident only at runtime of the compiler.
The well-definedness of AGs, \textit{i.e.}, the guarantee that every attribute of every syntax tree can be computed, has been a concern since Knuth's original paper~\cite{Knuth68}, whose test for circular definitions was later corrected~\cite{Knuth71}. Decades later, \citet{DeMoor00} still pointed out that ``the use of an attribute that has not been defined'' is a common source of runtime errors in AG-based systems---a problem that could be mitigated by an appropriate type system or static analysis (\sect{sec:program-dependence-analysis}). Well-definedness analyses address this problem by verifying that a defining equation exists for every attribute that may be demanded during evaluation. For instance, the modular analysis of Silver~\cite{Kaminski12} performs this check on each language extension in isolation, while JastAdd~\cite{Hedin03, Ekman07b} weaves all the aspects of a specification together before generating its evaluator. As we discuss in \sect{sec:related-work}, these guarantees concern the existence of equations and do not extend to the internal control flow of the (potentially imperative) code implementing each equation, which is exactly the gap this paper addresses.

With this in mind, in the \tool{Neverlang} world, an attempt had already been made by developing a type system and its type inference algorithm~\cite{Cazzola16h}.
This approach suffered from two main limitations: \begin{inparaenum}[(i)]
    \item separate compilation was not supported, and
    \item it was not fully modularizable.
\end{inparaenum}
These issues arose because, in both cases, it needed to see the entire composition of language constructs.
It is therefore of particular interest to determine to which extent the runtime errors which can arise in a \tool{Neverlang}-based project can also be detected statically, without losing the separate compilation constraint and without tightening the coupling among language constructs.
Thus, the development of a tool that performs such static checks and warns developers about potentially erroneous runtime instructions becomes a natural objective.
With such a tool, we believe language workbenches can achieve the best of both worlds by keeping a loose level of coupling while maintaining static correctness without relaxing the separate compilation constraint of the language extension problem.

In this work, we present a generic approach to achieve this goal based on data-flow analysis and an implementation thereof for the \tool{Neverlang} language workbench, named \nlgcheck{}. Concretely, \nlgcheck{} is run once separately compiled language components have already been composed into a single system, at which point it shifts the question from ``is the attribute defined somewhere in the composition?''---the structural guarantee already offered by the compiler---to ``is the attribute initialized before use along every possible execution path of the composed system?'', a property no compiler check can currently offer under separate compilation. We show that \nlgcheck{} can bridge the gap with the static analysis capabilities offered by compilers without separate compilations by navigating a representation of all possible ASTs that can be generated for a given grammar. Moreover, since it follows the control flow of the imperative code implementing each semantic action, our analysis also detects attributes that are written only within some branches of a conditional statement and then read elsewhere, a case that would otherwise surface as a null pointer exception at runtime. We do not claim that this makes it more precise than classical well-definedness analyses on AGs specified with equations: its purpose is to extend their guarantees to imperative semantic actions composed after separate compilation (\sect{sec:related-work}).
Finally, we complement our contribution with an experiment performed on existing \tool{Neverlang}-based projects, showing that \nlgcheck{} could find 73\% of injected naming errors over five different Neverlang languages through mutation testing.

The remainder of this paper is structured as follows.
Section~\ref{sec:background} provides the necessary background, introducing attribute grammars and program dependence analysis~\cite{Towle76, Banerjee79}, while Section~\ref{sec:software-frameworks} introduces the \tool{Neverlang} language workbench and the \tool{Soot} analysis framework that underpin our approach.
Section~\ref{sect:contribution} presents our main contribution: a generic, data-flow-based approach for statically checking attribute usage in the context of language composition with separate compilation, detailing the theoretical underpinnings and the rationale behind our method. Concretely, the artifact this analysis is performed on is not the \tool{Neverlang} source specification but the \textit{compiled bytecode} of each semantic action, from which a \textit{control-flow graph} (CFG) is extracted (\sect{ssect:complete}). The section then explains how the CFGs of individual semantic actions are composed into the CFG of an entire role (\sect{sect:contrib:roles}), how this CFG is visited to detect attribute-access violations while staying sound in the presence of unbounded recursion (\sect{sect:contrib:visit}), the optimizations that make this visit practical on real-world languages (\sect{ContribOptimization}), and, finally, how developers can act on the violations it reports (\sect{ContribActing}).
Section~\ref{sec:evaluation} evaluates the effectiveness of \nlgcheck{} through experiments on real-world \tool{Neverlang}-based projects, including a mutation testing campaign to assess its error-detection capabilities.
Section~\ref{sec:related-work} discusses related work, comparing our approach to existing solutions.
Finally, Section~\ref{sec:conclusion} concludes the paper, summarizing our contributions and outlining directions for future research.
An appendix provides additional technical details and supporting material.

\section{Background}\label{sec:background}

\subsection{Attribute Grammars}\label{sec:attribute-grammars}
A formal extension of context-free grammars is the concept of \textit{attribute grammars} (AG), introduced by \citet{Knuth68} as a way to integrate semantics into the syntactic structure of a language. In his seminal paper ``Semantics of Context-Free Languages'', the productions of a grammar are enriched with attributes and semantic equations that evaluate these attributes based on parse tree context---enabling systematic semantic analysis and evaluation within grammar-based systems. Attributes allow the transmission of information between different parts of the AST\@.
Just as parsing strategies are typically divided into \textit{bottom-up} and \textit{top-down} approaches, attributes are categorized according to the direction in which their values propagate: \textit{synthesized} or \textit{inherited}~\cite{Aho06, Knuth90}, respectively.
To formally define the attributes, let \(G = \langle \Sigma, N, S, \Pi \rangle\) be a context-free grammar, where
\begin{itemize}
   \item \(\Sigma\) is a set of terminal symbols,
   \item \(N\) is a set of non-terminal symbols,
   \item \(S \in N\) is the distinguished start symbol, and
   \item \(\Pi\) is a set of productions of the form \(A \rightarrow \alpha\), where \(A \in N\) and \(\alpha \in {(N \cup \Sigma)}^*\).
\end{itemize}

\smallskip\noindent\textbf{Synthesized Attributes.}\quad A \textit{synthesized attribute} is an attribute whose value is computed from the values of the attributes of its children nodes and propagated upward in the tree. Formally, if \(A \rightarrow \alpha_1 \dots \alpha_n \in \Pi\) is a production, then a synthesized attribute \(s\) of \(A\) can be defined as:
\begin{equation*}
   A.s = f(\alpha_{j_1}.s_1, \alpha_{j_2}.s_2, \dots, \alpha_{j_m}.s_m)
\end{equation*}
where \(f\) is a function that computes the value of \(A.s\) based on the synthesized attributes \(s_1, s_2, \dots, s_m\) of its children nodes \(\{\alpha_{j_1}, \alpha_{j_2}, \dots, \alpha_{j_m}\} \subseteq \{\alpha_1, \alpha_2, \dots, \alpha_n\}\).

\smallskip\noindent\textbf{Inherited Attributes.}\quad As stated by \citet{Knuth90}, P.\ Wegner introduced the \textit{inherited attribute} concept. Its value is computed from the values of the attributes of its parent and/or sibling nodes---either synthesized or themselves inherited---and propagated downward in the tree.
Formally, for a production $A \rightarrow \alpha_1 \dots \alpha_n \in \Pi$, an inherited attribute $i$ of $\alpha_k$ ($1 \le k \le n$) can be defined as:
\begin{equation*}
   \alpha_k.i = g(A.i_1, \dots, A.i_p, \alpha_{l_1}.a_1, \dots, \alpha_{l_q}.a_q)
\end{equation*}
where $g$ is a function that computes $\alpha_k.i$ based on the inherited attributes $i_1,\dots,i_p$ of the parent node $A$ and the attributes $a_1, \dots, a_q$---each either synthesized or inherited---of some sibling nodes $\{\alpha_{l_1}, \dots, \alpha_{l_q}\} \subseteq \{\alpha_1, \dots, \alpha_n\} \setminus \{\alpha_k\}$.

Over the years, several comprehensive works have shaped the study and application of AGs. Early collections such as ``Attribute Grammars: Definitions, Systems and Bibliography'' provided formal foundations, implementation frameworks, and an extensive bibliography~\cite{Deransart88}. A comprehensive survey by \citet{Paakki95} classified modern extensions of attribute grammars---structured, modular, object-oriented, logical, and functional variants---and reviewed techniques for parallel and incremental evaluation.
In parallel, foundational research addressed decidability and efficiency issues.\ \citet{Jazayeri75} established the inherent exponential complexity of circularity tests, while later work by \citet{Deransart84} developed more efficient algorithms to improve circularity checking and ensure consistent attribute evaluation.
More recently, research on AGs has continued along orthogonal directions: \citet{Ringo23} propose \textit{nanopass attribute grammars}, decomposing a single AG into a sequence of small, independently verifiable transformation passes, which shares our own goal of enabling modular, incremental verification of otherwise monolithic language specifications.

\subsubsection{Reference Attribute Grammars}\label{sec:reference-attribute-grammars}
A major limitation of canonical AGs is their handling of \textit{non-local dependencies}, i.e., situations in which a property of one syntax tree node depends on properties of nodes arbitrarily far away in the tree~\cite{Hedin00}. Several extensions to AGs have been proposed to address this limitation from complementary angles.\ \citet{Hedin94}'s \textit{door attribute grammars} let syntax-tree attributes hold references to other nodes, so a node's attributes become directly reachable from arbitrarily far away, without traversing the intervening tree structure; \citet{Boyland98} instead studies which non-local dependencies such a grammar exposes, statically and without requiring a full evaluation.

Once attributes can reference arbitrarily distant nodes, classical evaluation orders no longer directly apply: \citet{Hoover87}, \citet{Johnson82, Johnson85}, and \citet{Kaiser85} address this complementary evaluation problem, in particular for the incremental re-evaluation required by interactive language-based editors.\ \citet{Beshers85} and \citet{Vorthmann93} apply this line of work to two concrete sources of non-local information---state that persists across incremental edits, and name binding/visibility in object-oriented languages, respectively---before \citet{Hedin00} unified and generalized these ideas into \textit{reference attribute grammars}.
The need for such extensions arises from the fact that many semantic analyses, such as \textit{name resolution} and \textit{type checking}, inherently involve non-local dependencies. Reference attribute grammars address this by allowing attributes to hold references to other nodes in the syntax tree, enabling direct access to their attributes regardless of their position in the tree. This capability is particularly useful for implementing complex language features, such as \textit{scoping rules} and \textit{inheritance hierarchies} in object-oriented languages.

\subsection{Program Dependence Analysis}\label{sec:program-dependence-analysis}
Foundational to many static program analyses is the concept of \textit{program dependence analysis}. The main seminal work in this area are \citet{Towle76}'s PhD thesis, \citet{Banerjee79}, and \citet{Allen83b}, and the papers by \citet{Wolfe82} and \citet{Banerjee88}. This analysis identifies dependencies between different parts of a program, that is, a statement \(A\) is said to \textit{depend} on another statement \(B\) if the execution of \(A\) requires information from \(B\).
These dependencies can be classified into two main categories: \textit{control dependencies} (\sect{sec:control-dependencies}) and \textit{data dependencies} (\sect{sec:data-dependencies}).

\subsubsection{Dominance}\label{sec:dominance}
F. Allen~\cite{Allen70} introduced the concept of \textit{control flow graph} (CFG) as a way to represent all possible paths that a program can take during its execution.
Formally, a CFG is a directed graph \(G = (V, E, e_y, e_t)\), where:
\begin{itemize}
   \item \(V\) is a set of vertices (or nodes) representing basic blocks of code,
   \item \(E \subseteq V \times V\) is a set of edges representing a control flow between nodes,
   \item \(e_y \in V\) is the entry node, and
   \item \(e_t \in V\) is the exit node.
\end{itemize}
Since \(E\) is a set of pairs of nodes rather than a multiset, \(G\) admits no parallel edges: at most one edge connects any given ordered pair of nodes, so a path can be identified unambiguously with the sequence of nodes it visits.
In a CFG, a node \(A\) is said to \textit{(pre)dominate} another node \(B\) if every path from \(e_y\) to \(B\) must pass through \(A\). Formally, \(A \ \text{dom} \ B\) \textit{iif}:
\begin{equation*}
   \forall P \in \text{Paths}(e_y, B), \ A \in P
\end{equation*}
where \(\text{Paths}(e_y, B)\) is the set of all paths from \(e_y\) to \(B\), each represented as the sequence of nodes it visits.
The \textit{strictly dominates}  relation (\(A \ \overline{\text{dom}} \ B\)) is defined similarly, but it requires that \(A \neq B\).
The \textit{immediate dominator} of a node \(B\) is the unique node \(A\) such that \(A\) strictly dominates \(B\) and does not strictly dominate any other node that also dominates \(B\). Every node has an immediate dominator, except \(e_y\)~\cite{Lengauer79}. We write \(A \ \text{idom} \ B\) to denote that \(A\) is the immediate dominator of \(B\).
Intuitively, \(A\) is the last required ``checkpoint'' before reaching \(B\) in the control flow.
Formally, \(A \ \text{idom} \ B\) \textit{iif}:
\begin{equation*}
   A \ \overline{\text{dom}} \ B \land \forall C \in V : (C \ \text{dom} \ B \land C \neq A) \Rightarrow \neg(A \ \overline{\text{dom}} \ C).
\end{equation*}
The \textit{dominance frontier} of a node \(A\) is the set of all nodes \(B\) such that \(A\) dominates an immediate predecessor of \(B\), but \(A\) does not strictly dominate \(B\).
Formally, the dominance frontier of \(A\) is defined as:
\begin{equation*}
   \text{DF}(A) = \{ B \in V \mid \exists C \in \text{Pred}(B) : A \ \text{dom} \ C \land \neg(A \ \overline{\text{dom}} \ B) \}
\end{equation*}
\begin{wrapfigure}{r}{0.35\linewidth}
   \centering
   \captionsetup{type=Listing,skip=2pt}
    \showc[0.5\linewidth]{dominance-frontier-example.c}
    \includegraphics[width=1\linewidth]{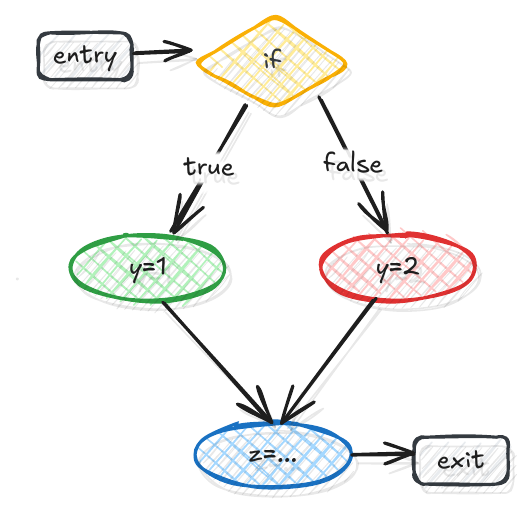}
   \caption{Dominance frontier.}%
   \label{lst:dominance-frontier-example}
\end{wrapfigure}
where \(\text{Pred}(B)\) is the set of immediate predecessors of \(B\) in the CFG\@.
Intuitively, the dominance frontier of \(A\) is the set of nodes where the control exercised by \(A\) ends.
In the Listing~\ref{lst:dominance-frontier-example}, the dominance frontier of the node \colorbox{green!20}{y=1} is the set \(\{\text{\colorbox{blue!20}{z=\ldots}}\}\) because:
\begin{enumerate*}
    \item the node \colorbox{green!20}{y=1} dominates a predecessor of \colorbox{blue!20}{z=\ldots} (itself), and
    \item the node \colorbox{green!20}{y=1} does not strictly dominate \colorbox{blue!20}{z=\ldots} (since \colorbox{red!20}{y=2} is another predecessor of \colorbox{blue!20}{z=\ldots} that is not dominated by \colorbox{green!20}{y=1}).
\end{enumerate*}

\textit{Dominator trees} are a tree-like structure rooted at \(e_y\), where each node \(A\) has a child \(B\) if \(A\) is the immediate dominator of \(B\).
The \textit{post-dominance} relation is the dual of dominance: \(A \ \text{dom}_{post} \ B\) \textit{iif} \(\forall P \in \text{Paths}(B, e_t), A \in P\), and \(A\) \textit{strictly post-dominates} \(B\) (\(A \ \overline{\text{dom}}_{post} \ B\)) if, in addition, \(A \neq B\).
\textit{Dominance} was first introduced by \citet{Prosser59}  in a 1959 study on the analysis of flow diagrams. However, Prosser did not provide an algorithm for computing dominance. A practical method for its computation was proposed ten years later by \citet{Lowry69}.\ \citet{Lengauer79} developed a more efficient algorithm for computing dominator trees in \textit{almost} linear time by using \textit{union-find} with path compression, improving existing solutions~\cite{Aho72, Purdom72b}, and this algorithm was then adopted by the GNU Compiler Collection (GCC)\@. More than twenty years later, \citet{Cooper06} presented a simpler algorithm that is faster in practice, relying on the \textit{reverse post-order} traversal of the CFG\@. The Low Level Virtual Machine (LLVM) compiler infrastructure, on the other hand, uses the algorithm by \citet{Georgiadis06}, which has a worst-case time complexity of \(O(n^2)\) but is more efficient in practice and, additionally, supports incremental updates~\cite{Georgiadis12}.\footnote{
 For a recent empirical characterization of LLVM's broader optimization pipeline, see~\citet{Bruzzone26f-preprint}.
}

\subsubsection{Control Dependencies}\label{sec:control-dependencies}
\textit{Control Dependencies} are post-dominance frontier in the CFG~\cite{Cytron89}. They refer to the dependencies that arise from the control flow of a program, such as conditional statements and loops. A statement \(A\) is said to be control dependent on another statement \(B\) if the execution of \(A\) depends on the outcome of \(B\).
Formally, following \citet{Cytron91}, \(A \overset{c}{\rightarrow} B\) if both of the following conditions hold:
\begin{equation*}
   \exists P \in \text{Paths}(B, A) : \forall C \in P, C \neq B \Rightarrow A \ \text{dom}_{post} \ C,
\end{equation*}
\begin{equation*}
   \neg(A \ \overline{\text{dom}}_{post} \ B).
\end{equation*}
That is, \(A\) post-dominates every node that follows \(B\) along some path from \(B\) to \(A\), but does not strictly post-dominate \(B\) itself. When \(A \neq B\), the second condition means that some path in \(\text{Paths}(B, e_t)\) avoids \(A\), so the outcome of \(B\) decides whether \(A\) is executed.
Note that the arrow in \(A \overset{c}{\rightarrow} B\) points from the dependent statement to the one it depends on, and therefore runs in the opposite direction to the corresponding CFG edge, along which control actually flows from \(B\) to \(A\): in Listing~\ref{lst:control-dependency-example}, the CFG edge goes from the condition \texttt{a < 0} to \texttt{x = -1}, executed only if the condition holds, while the control-dependency arrow is written \(\texttt{x = -1} \overset{c}{\rightarrow} \texttt{a < 0}\), read as ``\texttt{x = -1} is control dependent on \texttt{a < 0}''.
Due to the impact of data dependencies on control dependencies, the two concepts are often intertwined.
J. Allen et al.~\cite{Allen83} show how control dependencies can be transformed into data dependencies using the \textit{if-conversion} technique. This approach eliminates conditional execution by evaluating both control-flow paths unconditionally and then using a predicate to select the appropriate result, which is finally assigned to the target variable.
Listing~\ref{lst:control-dependency-example} illustrates a simple example of control dependency, where the statement \texttt{x = -1} is control dependent on the if-statement.

\subsubsection{Data Dependencies}\label{sec:data-dependencies}
\citet{Banerjee88} formalized the concept of \textit{data dependencies} within the context of program analysis. Data dependencies arise when the execution of one statement depends on the data produced by another statement.
\textit{Data hazards} are a specific type of data dependency that can occur in pipelined processors, where the execution of one instruction depends on the result of a previous instruction that has not yet completed; they can lead to \textit{race conditions}. Three main types of data hazards are identified: \textit{read-after-write} (RAW), \textit{write-after-read} (WAR), and \textit{write-after-write} (WAW).
In accordance with Bernstein~\cite{Bernstein66}, by contrast, data dependencies can be classified into: \begin{enumerate*}
   \item \textit{(true)flow dependencies},
   \item \textit{anti-dependencies}, and
   \item \textit{output dependencies}.
\end{enumerate*}
In the following, we use the notation \(s_i\) to refer to the \(i\)-th statement in the Listing~\ref{lst:data-dependency-example}.
\begin{figure}[t]
   \centering
   \begin{minipage}[b]{0.42\linewidth}
      \centering
      \showc*[0.9\linewidth]{control-dependency-example.c}
      \captionof{Listing}{Control dependency.}%
      \label{lst:control-dependency-example}
   \end{minipage}\hspace{0.06\linewidth}%
   \begin{minipage}[b]{0.42\linewidth}
      \centering
      \showc*[0.9\linewidth]{data-dependency-example.c}
      \captionof{Listing}{Data dependency.}%
      \label{lst:data-dependency-example}
   \end{minipage}
\end{figure}

\smallskip\noindent\textbf{Flow Dependency.}\quad  A statement \(s_3\) is said to be \textit{flow dependent} on another statement \(s_2\) (denoted by \(s_3 \ \delta^f \ s_2\)) if \(s_3\) uses a value produced by \(s_2\). RAW hazards occur when a violation of this dependency happens. In Listing~\ref{lst:data-dependency-example}, \(s_3\) (\texttt{int b = a + c;}) is flow dependent on \(s_2\) (\texttt{int a = 1;}), since it uses the value of \texttt{a} that \(s_2\) produced.

\smallskip\noindent\textbf{Anti-Dependency.}\quad An \textit{anti-dependency} takes place when a statement \(s_4\) writes to a variable that is read by another statement \(s_3\) (denoted by \(s_4 \ \delta^a \ s_3\)). This can lead to WAR hazards if the order of execution is not preserved. In Listing~\ref{lst:data-dependency-example}, \(s_4\) (\texttt{c = 2;}) is anti-dependent on \(s_3\) (\texttt{int b = a + c;}), since \(s_4\) writes to \texttt{c}, which \(s_3\) had read: swapping their order would change \(s_3\)'s result.

\smallskip\noindent\textbf{Output Dependency.}\quad An \textit{output dependency} is symmetric and occurs when two statements \(s_4\) and \(s_5\) write to the same variable (denoted by \(s_4 \ \delta^o \ s_5\)). WAW hazards happens if a violation of this dependency occurs. In Listing~\ref{lst:data-dependency-example}, \(s_4\) (\texttt{c = 2;}) and \(s_5\) (\texttt{c = b + 1;}) both write to \texttt{c}, so their relative order must be preserved to guarantee \texttt{c}'s final value.

\vspace{3pt}\noindent An additional type of data dependency is the \textit{input dependency}, which occurs when two statements read from the same variable. Anti-dependencies and output dependencies are often grouped together as \textit{storage-related dependencies}~\cite{Aho06, Hennessy11}---they may be removed by renaming variables.

\section{Software Frameworks}\label{sec:software-frameworks}

This section introduces the two main software frameworks that underpin the work described in this paper: the \tool{Neverlang} language workbench and the \tool{Soot} static analysis framework. Neverlang provides a modular infrastructure for language engineering, enabling the definition, composition, and extension of programming languages through reusable components. Soot, on the other hand, is a mature and widely used framework for analyzing and transforming Java programs at the bytecode and intermediate representation levels. Together, these tools form the foundation for our approach to integrating advanced program analysis capabilities into language development workflows.

\subsection{Neverlang in a Nutshell}\label{sec:neverlang-in-a-nutshell}
The \texttt{Neverlang}~\cite{Cazzola12c, Cazzola13e, Cazzola15c} language workbench\footnote{
    The term \textit{language workbench} was coined by Martin Fowler in his 2005 article ``Language Workbenches: The Killer-App for Domain Specific Languages''~\cite{Fowler05}. They are software tools that provide an integrated environment for designing, implementing, and maintaining programming languages. Language workbenches typically offer features such as syntax definition, semantic analysis, code generation, and editing support.
} promotes code reusability and separation of concerns leveraging the \textit{language-feature} approach.

\begin{Listing}[t]
   \centering
   \setminted[neverlang]{escapeinside=``}
   \showneverlang*[0.62\linewidth]{Backup.nl}
   \caption{Syntax and semantics for the backup task.}\label{lst:backup}
\end{Listing}
The basic development unit is the \inlineneverlang{module}, as shown in line 1 of Listing~\ref{lst:backup}.
A module may contain a \inlineneverlang{reference syntax} and could have zero or multiple \inlineneverlang{role}s. A role, used to define the semantics, is a composition unit that defines actions that should be executed when some syntax is recognized, as defined by \textit{syntax-directed translation}~\cite{Aho86}. Each role corresponds to one compilation phase, such as type checking or evaluation, and groups the semantic actions executed during one traversal of the AST\@: at runtime, the roles of a language are executed one after the other, each performing its own traversal of the whole tree.
Syntax definitions are defined using \textit{Backus-Naur form} grammars, represented as sets of \textit{productions} and \textit{terminals}.
Syntax definitions and semantic \inlineneverlang{role}s are tied together using \inlineneverlang{slice}s.
Listing~\ref{lst:backup} shows a simple example of a Neverlang module implementing a backup task of the \texttt{LogLang} \textit{language product line}, i.e., a family of related languages sharing a common set of reusable components~\cite{Bruzzone26e-preprint}. Reference syntax is defined in lines 2--6. The \textit{categories} (line 5) are used to generate the syntax highlighting for the \textit{integrated development environments} (IDEs).

A production may be given a \textit{label}, written before it and followed by a colon, such as \inlineneverlang{BackupCmd} in line 3. The label names the production and is distinct from its left-hand-side nonterminal, \inlineneverlang{Backup}.
Within a role, a semantic action is attached to a nonterminal occurrence of the reference syntax, identified either by its position or through the label of its production. Positions count nonterminals only and are numbered from 0 across the whole reference syntax, from the top left to the bottom right. In Listing~\ref{lst:backup}, the nonterminals of the first production (\inlineneverlang{Backup} and the two \inlineneverlang{String}s) have positions 0, 1, and 2, and those of the second production (\inlineneverlang{Cmd} and \inlineneverlang{Backup}) have positions 3 and 4. The action in line 8 is attached to position 0, \textit{i.e.}, to the left-hand side of the first production, and is executed whenever a node built by that production is visited. Inside the action, \inlineneverlang{$0}, \inlineneverlang{$1}, and \inlineneverlang{$2} denote the nonterminals of that same production, so \inlineneverlang{$1.string} reads the \texttt{string} attribute of the first \inlineneverlang{String} child. An action for the second production would be written in the same way, either as \inlineneverlang{3 .{...}.} or, through its label, as \inlineneverlang{BackupAsCmd: .{...}.}. A label followed by an index, such as \inlineneverlang{BackupCmd[1]}, denotes the corresponding nonterminal of the labeled production.
Two productions are given rather than one: the production labeled \inlineneverlang{BackupCmd} defines the concrete syntax of the backup command itself, while the one labeled \inlineneverlang{BackupAsCmd} (line 4) injects \inlineneverlang{Backup} as one alternative of the more general \inlineneverlang{Cmd} nonterminal. Other modules can likewise contribute their own alternatives for \inlineneverlang{Cmd} (\textit{e.g.}, \inlineneverlang{Remove}, \inlineneverlang{Rename}) without altering this module, which is what makes the statement syntax extensible across slices.

Line 10 shows the \inlineneverlang{$$}-prefixed identifier notation (\inlineneverlang{$$FileOp}), used to access a singleton instance of a service object---here, a library providing file system manipulation primitives---from within a semantic action. Such a service is implemented by a separate component called an \textit{endemic slice} (line 22), which declares an interface and a concrete implementation that the \tool{Neverlang} runtime automatically instantiates when the language implementation is initialized. Unlike ordinary \inlineneverlang{slice}s, an endemic slice does not contribute any syntax or AST-visiting role of its own: it only provides a runtime service that semantic actions may invoke directly, independently of the specific production or attributes they are attached to, and is therefore no less compositional than any other module. 

The order in which roles are executed is declared by the \inlineneverlang{roles} clause of the language, as shown in lines 23--25 and detailed below. Besides the ordering, each role is also associated with an \textit{execution mode}, which determines how the role traverses the AST\@. In \textit{postorder} mode, the default, the children of each node are visited before the action attached to the node itself, without requiring the semantic actions to trigger the traversal explicitly. In \textit{manual} mode, selected by prefixing the role name with \inlineneverlang{+} in the \inlineneverlang{roles} clause (\textit{e.g.}, \inlineneverlang{<+ evaluation}), the action attached to a node is executed first and is itself responsible for visiting the children, through the \inlineneverlang{eval} construct (\sect{sect:contribution}), while nodes without an action are still traversed automatically.
Unlike the visit sequences of Kastens' \textit{ordered attribute grammars} (OAGs)~\cite{Kastens80}, which an OAG system \textit{infers} from the attribute dependency graph and are guaranteed by construction to evaluate every attribute after its dependencies, \inlineneverlang{eval} in manual mode is an explicit control-transfer statement written by the language developer: \tool{Neverlang} performs no such inference and does not check that the developer's chosen order actually respects attribute dependencies. This is precisely the gap that \nlgcheck{} is designed to close for manual-mode roles: a static, after-the-fact well-definedness check standing in for the guarantee that an OAG's visit-sequence construction would otherwise have provided at generation time.
In contrast, the \inlineneverlang{BackupSlice} (lines 14--18) reveals how the syntax and semantics are tied together; choosing the \inlineneverlang{concrete syntax} from the Backup module (line 15), and two \inlineneverlang{role}s from two different modules (lines 16--17).
Finally, the \inlineneverlang{language} can be created by composing multiple \inlineneverlang{slices} (line 20).
The composition in Neverlang is twofold: between modules and between slices. Thus, the grammars are merged to generate the complete language parser. On the other hand, the semantic actions are composed in a pipeline, and each \inlineneverlang{role} traverses the syntax tree in the order specified in the \inlineneverlang{roles} clause (lines 23--25). The \inlineneverlang{<} operator states that the role on its left must run to completion before the role on its right starts (\textit{e.g.}, \inlineneverlang{syntax} before \inlineneverlang{terminal-evaluation} before \inlineneverlang{permissions}). The \inlineneverlang{:} operator, instead, combines roles into a single \textit{layered} traversal: with \inlineneverlang{permissions : execution}, each node executes its \inlineneverlang{permissions} action immediately followed by its \inlineneverlang{execution} action, so the two roles are interleaved node by node instead of requiring two separate traversals.
Beyond slices, Neverlang also supports \textit{bundles}: a bundle package is a separately compiled sublanguage (i.e., a self-contained group of modules that does not necessarily define a complete language on its own) that can be shipped and reused as a single unit across different host languages, without requiring access to its source form.

Since 2019~\cite{Cazzola19}, Neverlang has always had an interest in providing the \textit{editing support} for the languages developed with it. Initially, the editing support was provided through the \textit{Eclipse IDE} plugin, which was recently replaced by the \textit{de facto} standards \textit{language server protocol} (LSP) and \textit{debug adapter protocol}~\cite{Cazzola25b}. Both protocols decouple editing/debugging support from any specific editor: any editor that implements the client side of the protocol (e.g., Visual Studio Code, the Eclipse IDE, or Neovim) can provide editing and debugging support for a language, provided the language itself ships a compatible language server. To provide LSP services, Neverlang leverages the \texttt{Typelang} DSL---a language for type systems definition---to generate \textit{language servers} that can be used by any editor that supports the LSP\@.

\subsection{Soot Library}\label{sec:soot-library}
\texttt{Soot}~\cite{Vallee-Rai99, Vallee-Rai00} is a Java analysis and transformation framework, first released in 1999 and still actively maintained today. Initially conceived for bytecode optimization, it has evolved into a comprehensive tool for static analysis~\cite{Vallee-Rai99, Sundaresan00, Pominville01}.
A recent effort (``SootUp'') completely reimplements Soot with a more modular design, modernizing the original tool while preserving its familiar APIs~\cite{Karakaya24}.

Soot's main strength lies in its intermediate representations (IR), in particular \textit{Jimple}, a typed three-address code designed to make control flow and data flow explicit and thus simplify static analyses. Other IRs (\textit{Shimple} in \textit{static single assignment} form, \textit{Baf} for bytecode-level rewriting~\cite{Rudys02}, and \textit{Grimp} for decompilation~\cite{Miecznikowski03}) extend its flexibility~\cite{Lam11}.
On top of these IRs, Soot offers standard analyses such as call-graph construction~\cite{Karakaya24}, pointer analysis (\textit{Spark})~\cite{Lhotak03}, and data-flow analysis frameworks (\textit{Heros})~\cite{Bodden12,Reps95,Sagiv96}. These enable both whole-program analyses and program transformations, with outputs ranging from optimized bytecode to instrumented programs.

Soot has been applied to a wide range of program analysis and transformation tasks. Early on it was used for Java optimizations and class file instrumentation~\cite{Dann19}, but it quickly found broad academic use. It has underpinned research in pointer analysis, concurrency analysis, symbolic execution, and the static portions of hybrid analysis techniques~\cite{Lam11}.

\section{Program Dependence Analysis for Separately Compiled Languages}\label{sect:contribution}
This section presents the concept behind \nlgcheck{}, which represents the core of our contribution.

\subsection{Notation}\label{ssect:notation}
We restate here, in a form specialized to the needs of this section, the grammar and attribute notions already introduced in \sect{sec:attribute-grammars}: a \textbf{grammar} is defined as a quadruple $(\Sigma, N, S, \Pi)$, in which $\Sigma$ is a set of terminal symbols (not interesting in this context, since attributes are only attached to nonterminal symbols), $N$ is a set of nonterminal symbols, $S \in N$ is the start nonterminal symbol (which in \tool{Neverlang} is always named \texttt{Program} by convention), and $\Pi$ is a set of production rules.

A \textbf{production} $p \in \Pi$ is written in the form $X_0 \leftarrow X_1 \cdots X_n$, with $X_i \in N$ all being nonterminal symbols, and has the meaning that during the parsing process, $X_0$ (which is called the \textbf{left-hand side} of the production) can be replaced by the $X_1 \cdots X_n$ sequence (called the \textbf{right-hand side} of the production). The notation $|p|$ is used to indicate the number $n + 1$ of nonterminals forming the production (one on the left-hand side and $n$ on the right-hand side) and $p[i]$ to indicate the $i$-th nonterminal $X_i$.\footnote{Note that this is a simplification of the definition of production, in which the right-hand side is a sequence that can contain both terminal and nonterminal symbols. However, this does not pose a limitation in this context, since our approach focuses on the analysis of attributes attached to nonterminals.}

Given a production $p \in \Pi$, we will refer to the attribute named $a$ on the $i$-th nonterminal of $p$ with the notation $p[i].a$: with $p[0].a$, we indicate an attribute on the left-hand nonterminal of the production, and with $p[i].a$ (with $i \ge 1$), an attribute on one of the nonterminals on the right-hand side.

A \textbf{semantic action}, for the scope of this section, can be defined as a sequence of executable code attached to a production, some of which can set or access attributes associated with the production's nonterminals.\footnote{This is also a simplification, but it suffices to discuss anything relevant to the problem we are dealing with.}

Any attribute set or accessed by an instruction has a \textbf{type}, which is usually known during the compilation of the semantic action code. We call $T$ the set of all available types.
$T$ is partially ordered since, given two types $t_1, t_2 \in T$, exactly one of the following is true:
\begin{compactitem}
    \item $t_1$ and $t_2$ are the same type ($t_1 = t_2$);
    \item $t_1$ is a subtype of $t_2$ ($t_1 <: t_2$);
    \item $t_1$ is a supertype of $t_2$ ($t_1 :> t_2$);
    \item none of the above ($t_1$ and $t_2$ are unrelated).
\end{compactitem}
For brevity, we write $t_1 \le: t_2$ (\textit{subtype-or-equal}) to mean $t_1 <: t_2 \lor t_1 = t_2$. A value of type $t_1$ is \textit{assignable} to a location of type $t_2$ exactly when $t_1 \le: t_2$.

To refer to the type of an attribute $p[i].a$, we use the notation $\tau(p[i].a)$. More precisely, given $\mathbb{N}$ as the set of natural numbers and $I$ as the set of valid attribute names,\footnote{Also called \textit{identifiers}; for all intents and purposes, we can assume that the set of valid identifiers is a subset of all ASCII strings.} we can define $\tau$ as a partial function $\tau: \Pi \times \mathbb{N} \times I \rightarrow T$ that maps the triplet (production, index of a nonterminal in the production, and attribute name) to the type of the corresponding attribute. $\tau(p[i].a) \coloneq \tau(p, i, a)$ is not defined if $p$ has less than $i$ nonterminals or if $p[i]$ has no attribute named $a$. The attribute name $a$ alone does not suffice to identify a type: the same nonterminal can occur at several positions of the same production (\textit{e.g.}, in a right-recursive rule), and different productions can attach an attribute with the same name but a different type to occurrences of the same nonterminal. The triple $(p, i, a)$ disambiguates both cases.

\subsection{Verification Model}\label{ssect:complete}
A simplified model resembling the type-checking algorithms already found in most AG systems (\sect{ssect:postorder}) proved insufficient, as detailed there. We therefore opted for a more complete model that closely represents the executable action. The natural choice in this regard is the CFG~\cite{Allen70}.

The CFG for a semantic action we will be using in this work is an instance of the CFG presented in \sect{sec:dominance}. We define a directed graph $G=(V,E)$ with $E \subseteq V \times V$ in which every vertex (or node) has a unique identifier (for the sake of simplicity, we will assume $V \subset \mathbb{N}$) and is associated with an \textit{action}, each being a sequence of statements that forms an atomic subset of a more complex semantic action.
We also define an \textbf{entry node} and an \textbf{exit node} of the CFG, representing the first and last action of the semantic action execution; the execution can follow different paths but always begins at the entry node and terminates at the exit node. The entry node and exit node are denoted as $v_{entry}(G)$ and $v_{exit}(G)$, respectively.

The semantics of each node are described through their mapping to an action by means of the $a_{node}:V \rightarrow A$ function, where $A$ is the set of all possible instructions that can be executed in a semantic action.
In this work, we consider the following instructions, grouped by category. Most instructions can be parametrized, making $A$ an infinite set.
We adopt the following notation for the parameters and their types: $\mathbb{N}$ is the set of natural numbers, $I$ is the set of valid attribute identifiers, and $T$ is the set of types.\footnote{For the purposes of this work, $T$ coincides with the set of all the types in the \tool{Java} Virtual Machine (JVM).}

\mypar{Attribute actions}
Attribute actions model access and manipulation of AST node attributes. We consider the following attribute actions (and their respective CFG nodes):
\begin{attributebox}{$WriteAttr(i, n, t)$, with $i \in \mathbb{N}, n \in I, t \in T$}
    The action sets an attribute named $n$ of type $t$ on the $i$-th nonterminal of the production.
    If $i = 0$, then the attribute is written on the node currently being visited; otherwise, it is written to its $i$-th child.
\end{attributebox}
Unlike classical AGs, where the single-assignment rule guarantees that each attribute instance is written by exactly one equation, our model treats attributes as ordinary mutable locations: since semantic actions are imperative code rather than declarative equations, nothing in \tool{Neverlang} prevents a semantic action from executing several $WriteAttr$ nodes for the same attribute along the same path (\textit{e.g.}, inside a loop, or across the branches of a conditional that both assign it before continuing). $WriteAttr(i, n, t)$ therefore always overwrites any previous value, and a $ReadAttr$ is checked against the type of the most recent write reaching it along the current path---ordinary imperative dataflow semantics, not the AG single-assignment discipline.
\begin{attributebox}{$ReadAttr(i, n, t)$, with $i \in \mathbb{N}, n \in I, t \in T$}
    The action accesses an attribute named $n$ from the $i$-th nonterminal of the production and casts its value to type $t$.
    Please note that the previous restriction about the ability to read attributes only from child nodes and write only to the current node (see \sect{ssect:postorder}) is lifted, as this model does not assume postorder execution.
\end{attributebox}
\begin{attributebox}{$CopyAttr(i_d, n_d, i_s, n_s)$, with $i_d, i_s \in \mathbb{N}, n_d, n_s \in I$}
    The action accesses an attribute named $n_s$ from the $i_s$-th nonterminal of the production and sets an attribute named $n_d$ on the $i_d$-th nonterminal of the same production to the retrieved value.\footnote{The $s$ and $d$ subscripts stand for \textit{source} and \textit{destination}.} The runtime behavior of $CopyAttr$ is identical to performing a $ReadAttr(i_s, n_s, t)$ followed by a $WriteAttr(i_d, n_d, t)$ with the same $t$ and value, but the distinction matters in this context. Copying an attribute does not involve interacting with its type, so the type is preserved even when not explicitly mentioned. Using the read-write sequence instead, without that type information, would default to $t = Object$, losing relevant type information.
\end{attributebox}
\begin{attributebox}{$CheckAttrType(i, n, t)$, with $i \in \mathbb{N}, n \in I, t \in T$}
    The action accesses an attribute named $n$ from the $i$-th nonterminal of the production and compares its type to $t$.\footnote{in \tool{Java}, this is done with the \inlinejava{instanceof} operator; similar or equivalent type tests are available in other languages.} This action does not modify the state nor can it cause runtime errors, but it is used to raise warnings: using \inlinejava{instanceof} on attribute values is considered an anti-pattern in \tool{Neverlang}, as the same check could be performed more idiomatically with dedicated type guards.
\end{attributebox}
\begin{attributebox}{$PropagateAttrs()$}
    Assuming that the semantic action is attached to a production with a single nonterminal on the right-hand side, thus has the form $X_0 \leftarrow X_1$, each attribute present on $X_1$ is copied to $X_0$ with the same name and type, unless already present. The code that performs such a copy is not written by the developers. Instead, the copy is performed automatically by the \tool{Neverlang} runtime. Thus, we attach additional nodes to the CFG representing this implicit behavior. Like $CopyAttr$, and for the same reason, $PropagateAttrs$ does not perform any type checking: since every propagated attribute retains its original type, no mismatch can occur.
\end{attributebox}

\begin{Listing}[t]
    \centering
    \setminted[neverlang]{escapeinside=``}
    \showneverlang*{LangLibTernaryOp.nl}
    \caption{\tool{Neverlang} code implementing a \tool{C}-like ternary expression language feature.}\label{lst:ternary}
\end{Listing}

\mypar{Evaluation actions} Evaluation actions model the control that the semantic action has over the tree-traversal order.
As exemplified in Listing~\ref{lst:ternary} in the context of the implementation of the ternary operator language feature, a semantic action can contain the \inlineneverlang{eval $i} construct (lines 7, 9, 12) to instruct the \tool{Neverlang} runtime to suspend the execution of the semantic action, instead proceeding by evaluating the $i$-th child of the current AST node according to the same role. This particular example is interpreter-style, in that the semantic action directly computes and returns the expression's value (\inlineneverlang{$0.val}) rather than emitting machine code for it. Our model, however, only concerns itself with which attributes a semantic action reads and writes, and the order in which it does so, regardless of whether the action's ultimate effect is to compute a value or to emit code that will compute it at a later stage---the same reasoning therefore applies unchanged to a semantic action implementing, \textit{e.g.}, \tool{LLVM} IR code generation for the same ternary operator. Once the \inlineneverlang{eval} is complete, the execution of the current semantic action is resumed (as is typically done when invoking functions in most programming languages).
We consider the following attribute actions (and their respective CFG nodes):
\begin{attributebox}{$BeginEvalMetaAction(i)$, with $i \in\mathbb{N}$}
   The action evaluates the $i$-th nonterminal of the production, according to the role currently being executed.
\end{attributebox}
\begin{attributebox}{$EndEvalMetaAction()$}
    The action is resumed after visiting another node.
\end{attributebox}

Concretely, these two nodes always appear jointly in the CFG\@: a node mapped to a \(BeginEvalMeta\-Action\) always has exactly one outgoing edge, connected to an $EndEvalMetaAction$, which is its only incoming edge. The operation is atomic and spans over two nodes for purely practical CFG manipulation reasons (please refer to the next section for further details).
\tool{Neverlang} provides syntactic sugar for common behaviors, such as making a single semantic action behave as if it were used in a postorder role: this is denoted with the \texttt{@\{\ldots\}.} operator.
Such syntactic sugar does not impact the creation of the CFG, as it is equivalent to performing \texttt{eval} statements for each of the right-hand side nonterminals of the production before proceeding further.

\mypar{Other actions} Other actions unrelated to previous categories are the following:
\begin{attributebox}{$Nop()$}
    The action does not represent any real operation; thus, the node is within the CFG for structural reasons (such as denoting branch points).
\end{attributebox}
\begin{attributebox}{$IgnoreBranchAction()$}
    The action throws a runtime exception that would interrupt the execution of the interpreter/compiler. Since the visit cannot proceed further after throwing an exception, the branch is ignored.
\end{attributebox}
Fig.~\ref{fig:ternaryCFG} shows the CFG obtained from the semantic action presented in Listing~\ref{lst:ternary}, according to the concepts presented thus far. The label on top of Fig.~\ref{fig:ternaryCFG} lists the production and the role to which the semantic action is attached. Each node is uniquely identified with a hexadecimal number.
For the implementation, the CFG for a semantic action is constructed by analyzing the compiled bytecode of the \tool{Java} classes generated by the \tool{Neverlang} compiler using \tool{Soot}, as detailed later in this section.

\subsection{Modeling Roles}\label{sect:contrib:roles}
The concept of a CFG is not limited to individual semantic actions, equivalent to \textit{intra-procedural control-flow graphs} in traditional CFGs.

Inspecting the behavior of the whole compiler/interpreter, thus simulating the entire AST visit (role, in \tool{Neverlang}), requires taking inspiration from static analysis techniques based on the \textit{inter-procedural control-flow graph} instead~\cite{Callahan86}.
The natural choice is to combine the CFGs for the individual semantic actions---connecting their entry, exit, and eval nodes---into a unique CFG for the entire role through adequate edges.

\begin{figure}
    \centering
    \includegraphics[width=.8\linewidth]{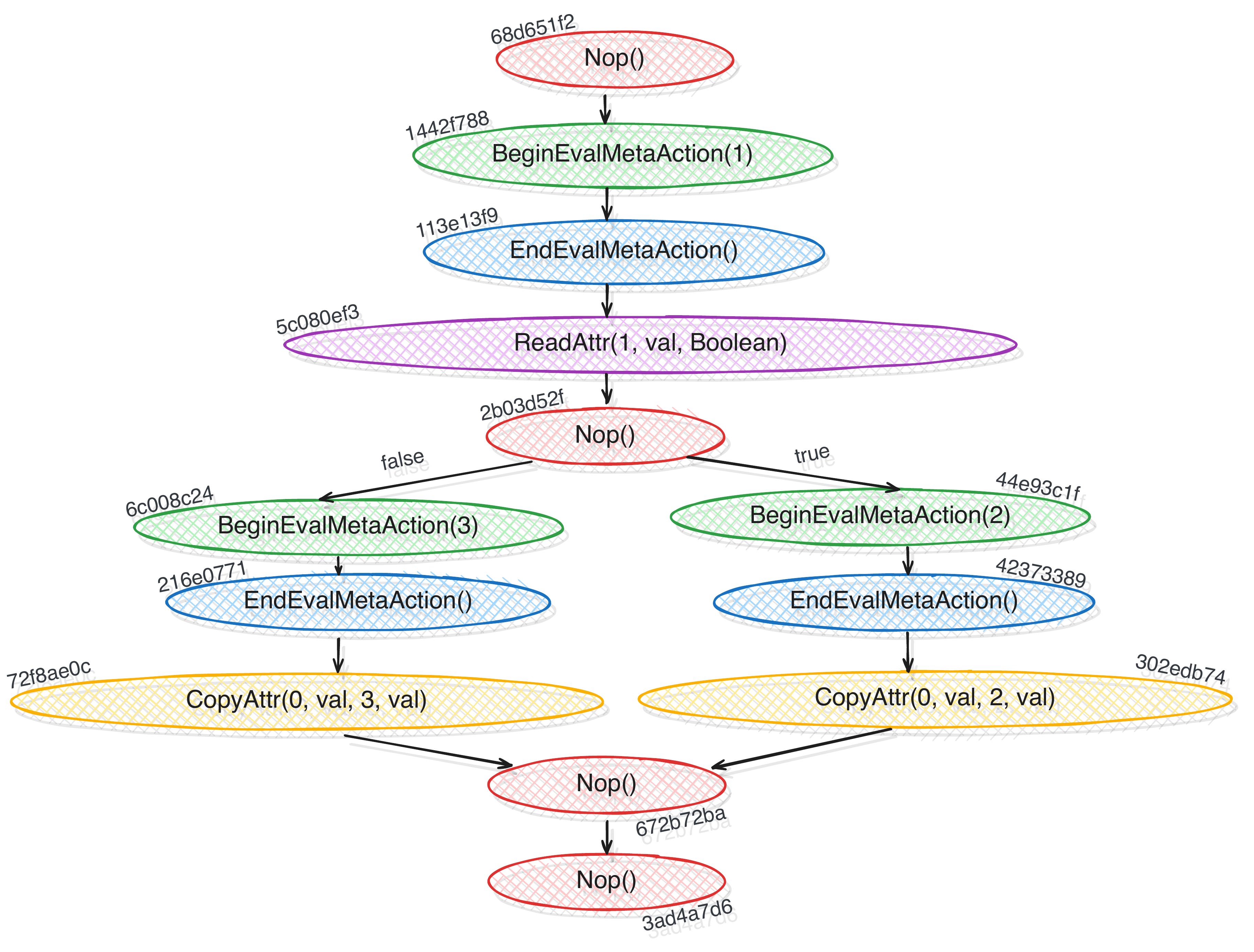}
    \caption{Control flow graph generated by \nlgcheck{} for the semantic action shown in Listing~\ref{lst:ternary}.}\label{fig:ternaryCFG}
\end{figure}

Such information is obtained from the language grammar---\textit{e.g.}, the semantic action for \inlineneverlang{Stmt <-- "if" Expr "then" Stmt}.
To indicate that the execution flow may continue towards a semantic action whose production's head is the \inlineneverlang{Expr} nonterminal a $BeginEvalMetaAction(1)$ node is used---\textit{i.e.}, a production in the \inlineneverlang{Expr <-- `\ldots`} form; then, the semantic actions can be connected accordingly.

There are a few factors to consider when navigating the CFG built in this way.

\mypar{Returning to the correct call site}
In general-purpose program analysis, a given procedure may potentially be called by many others (i.e., it has many \textit{call sites}).
In our model, each semantic action equates to a procedure. Thus, the call site of a semantic action $G$ attached to a production $p: Y \leftarrow \cdots$ may reside in any of the $G_c$ semantic actions (where the $c$ subscript denotes the \textit{call site}) attached to productions $p_c$ having $Y$ on their right-hand side.

Moreover, if $p_c$ has multiple instances of $Y$ on the right-hand side, then the call site of $G$ may correspond to each instance of $Y$ in $p_c$.
Let us consider a ``caller'' semantic action $G_i=(V_i,E_i)$ attached to a production $p_i$ and a ``called'' semantic action $G_j=(V_j,E_j)$ attached to a production $p_j$. We take $v_{beginEval}, v_{endEval} \in V_i$ and $k \in \mathbb{N}$ such that
    \[p_i[k] = p_j[0],\]
    \[a_{node}(v_{beginEval})=BeginEvalMetaAction(k),\]
    \[a_{node}(v_{endEval})=EndEvalMetaAction(),\]
    \[(v_{beginEval}, v_{endEval}) \in E_i\]
Then, given the edge set $E_r$ of the CFG $G_r$, the control flow being transferred from $G_i$ to $G_j$ is performed as follows:
\begin{compactenum}
    \item remove $(v_{beginEval}, v_{endEval})$ from $E_r$ (as the execution of $G_j$ cannot be skipped);
    \item add $(v_{beginEval}, v_{entry}(G_j))$ to $E_r$;
    \item add $(v_{exit}(G_j), v_{endEval})$ to $E_r$.
\end{compactenum}

Since the language grammar potentially comprises many other productions $p_l \neq p_i$ such that $\exists k': p_l[k'] = p_j[0]$, these steps are performed on each $p_i, p_l, \ldots$, leading to several pairs of vertices in the form $v_{beginEval}', v_{endEval}' \in V_l \subseteq V_r$ such that $(v_{beginEval}', v_{entry}(G_j)) \in E_r$ and $(v_{exit}(G_j), v_{endEval}') \in E_r$.
Some of these vertices may never be navigated in a real interpreter; thus, we introduce \textit{tagged} \textit{entry} and \textit{exit edges}.
We define a set $L$ of tags, two sets $E_{entry},E_{exit} \subset E_r$ such that $E_{entry} \cap E_{exit} = \emptyset$ and a function $l_{edge}: (E_{entry} \cup E_{exit}) \rightarrow L$.
We define a \textit{tagged entry edge} as an edge $e \in E_{entry}$ such that $l_{edge}(e)=l$ (with $l \in L$) and a \textit{tagged exit edge} as an edge $e \in E_{exit}$ such that $l_{edge}(e)=l$.
Then, we can imagine paths as well-parenthesized expressions, considering valid only those paths in which every exit edge (closing bracket) has a corresponding entry edge (opening bracket) with the same tag.

This leads to a formal definition for a well-formed CFG path:
\begin{compactitem}
    \item the empty path $\epsilon$ is a well-formed path;
    \item given a pair of vertices $v_a, v_b \in V_r$ connected by an untagged edge $e=(v_a, v_b)$, such that $e \in E_r$ and $e \notin (E_{entry} \cup E_{exit})$, $v_a v_b$ is a well-formed path;
    \item given the vertices $v_a, v_b, v_c, v_d \in V_r$ and a well-formed path $w=v_b \ldots v_c$, let $e_{entry}=(v_a, v_b)$ and $e_{exit}=(v_c, v_d)$; then $v_a w v_d$ is a well-formed path if
        \[e_{entry} \in E_{entry} \land e_{exit} \in E_{exit} \land l_{edge}(e_{entry}) = l_{edge}(e_{exit})\]
    \item given two well-formed paths $w_a$ and $w_b$, their concatenation $w_a w_b$ is a well-formed path.
\end{compactitem}
Finally, any prefix of a well-formed path is a \textit{valid path}.

For instance, reusing the nested-ternary tagged entry/exit edges introduced in \sect{sect:ContribCfgVisitValidPaths}---entry edge $(\texttt{44e93c1f}, \texttt{68d651f2})$ and exit edge $(\texttt{3ad4a7d6}, \texttt{42373389})$, both tagged $l$---a path that traverses the entry edge, visits the whole callee CFG, and returns via that same exit edge is well-formed (the tags match, as required by the third item above). By contrast, a path that traverses the entry edge and then simply stops partway through the callee, without ever traversing a corresponding exit edge tagged $l$, is a \textit{valid} path (it is a prefix of a well-formed one, namely the path that does go on to complete the call) but is not itself well-formed, since it is not balanced.

Two semantic actions $G_i$ and $G_j$ are therefore connected as follows:
\begin{enumerate}
    \item remove $(v_{beginEval}, v_{endEval})$ from $E_r$;
    \item add $e_1=(v_{beginEval}, v_{entry}(G_j))$ to $E_{entry}$ (and thus to $E_r$);
    \item add $e_2=(v_{exit}(G_j), v_{endEval})$ to $E_{exit}$ (and thus to $E_r$);
    \item let $l$ be a new tag ($l \notin L$) and add it to $L$;
    \item let $l_{edge}(e_1)=l_{edge}(e_2)=l$.
\end{enumerate}

When the visit reaches $v_{exit}(G_j)$, only one among many outgoing edges matches the tag followed when entering the semantic action. Following such an edge returns control to the correct call site.

\mypar{Context switching}
Because a semantic action may, through \texttt{eval}, transfer control to another semantic action that (directly or through a chain of further calls) transfers control back to a semantic action attached to the same production---as illustrated by the recursive ternary-expression example of \sect{sect:ContribCfgVisitValidPaths}---simulating execution requires keeping track of a possibly unbounded number of nested, concurrently active invocations. This calls for two stack data structures: one to store the production context of each active invocation (introduced in this section) and one to track the nesting of tagged entry/exit edges as calls are made and returned from (\sect{sect:ContribCfgVisitValidPaths}). This bookkeeping is a direct consequence of \texttt{eval} being an explicit, developer-written control transfer rather than an OAG-style inferred visit sequence (\sect{sec:neverlang-in-a-nutshell}): an OAG evaluator's generated visit sequence needs no such simulation, since it is correct by construction, whereas here nothing guarantees a priori that a chain of \texttt{eval} calls respects attribute dependencies, which is exactly what the stacks introduced below let us check.
Execution simulation involves keeping the state in proper data structures.
The state comprises pieces of information that will be detailed later; for now, consider that the behavior of a semantic action depends on contextual information.
In fact, several of the actions introduced earlier involve $i$ parameters that refer to the AST node being visited.
Thus, we extend the model with two more actions: $EnterCtx(i)$ and $LeaveCtx()$ (with $i \in \mathbb{N}$).
The semantics are akin to procedures in low-level programming, which often involve a \textit{prologue} and an \textit{epilogue}: additional instructions surrounding the actual body, often used to allocate and free stack frames.
Likewise, we can use a stack data structure: each item on the stack, which we call a \textit{production context}, stores information (name and type) of the attributes of the current node's children, with the base of the stack representing the root AST node, which can also store its own attributes.

When the current semantic action sets or accesses an attribute on a right-hand-side nonterminal ($i \ge 1$) of its attached production, the operation is performed directly on the production context on top of the stack.
When the current semantic action sets or accesses the left-hand-side nonterminal ($i=0$), the operation is performed on the second-to-top production context of the stack with position $i_c$.
This model reflects how attributes are passed at runtime between semantic actions depending on their position within the productions.

$EnterCtx(i)$ and $LeaveCtx()$ are therefore equivalent to procedure prologues and epilogues, respectively: the former pushes a new production context with $i_c=i$ to the stack, while the latter pops the production context from the top of the stack.
These two actions are not additional nodes inserted alongside $BeginEvalMetaAction$/$EndEvalMetaAction$ (\sect{sect:contrib:roles}); rather, once the composition step described there connects a caller to a callee through a tagged entry/exit edge pair, the node mapped to $BeginEvalMetaAction(k)$ is remapped to $EnterCtx(k)$, and the node mapped to $EndEvalMetaAction()$ is remapped to $LeaveCtx()$. They are the same two nodes, in the same CFG positions, simply carrying different actions once their role in an actual call is resolved.

This recursive call structure is unrelated to the classical notion of \textit{circular attributes}: since \nlgcheck{} only performs a static analysis over CFG paths, rather than actually evaluating attributes, it is not concerned with whether a fixed-point computation over attribute values would converge, but only with whether every read is preceded by a matching write along a given path. The depth-2 cutoff of \sect{sect:ContribInfinitePathsSolution} accordingly bounds the recursive unfolding of the \textit{analysis}, not of any attribute value computation. Exceptions raised by semantic actions are not prohibited: they are already accounted for by the $IgnoreBranchAction()$ node (\sect{ssect:complete}), which marks the branch as an unreachable dead end for the purposes of the CFG traversal, since no further attribute accesses can occur once control has left the current semantic action via an exception.

\mypar{CFG completion}
Entry and exit nodes must be defined for the entire CFG $G_r$. This is done via two additional nodes, both mapped to the $Nop()$ action ($a_{node}(v_{entry}(G_r)) = a_{node}(v_{exit}(G_r)) = Nop()$), and then adding proper edges to connect them, respectively, with the entry and exit nodes of each production having \texttt{Program} (the start symbol) as its right-hand side.
\Cref{AlgBuildRoleCFG} in \sect{sec:app:algofull} summarizes the contents of this section by showing the complete CFG construction algorithm from a set of productions and the control-flow graphs of the respective semantic actions.

\subsection{Visiting the CFG}\label{sect:contrib:visit}
The control-flow graph is a representation of all possible execution paths for a role: for each parse tree that can be produced according to the grammar, the CFG contains a corresponding valid path, representing the sequence of operations performed during the tree visit. Since most practical grammars are recursive, the CFG generally admits infinitely many such paths (\sect{sect:ContribInfinitePathsSolution} discusses this problem and its solution in detail). The definitions in this section are set up so that they remain meaningful even though the visit algorithm presented later in this section will ultimately only explore a carefully bounded, finite subset of them.

\subsubsection{Execution State Model}
The simulated execution state is a tuple $(v, w, S, C)$ such that
\begin{compactitem}
    \item $v \in V_r$ is the node currently being visited;
    \item $w \in V_r^*$ is the path followed to reach $v$;
    \item $S$ is the stack of edge tags;
    \item $C \in \mathcal{C}$ is the stack of production contexts, where $\mathcal{C}$ is used to refer to the (infinite) set of all possible production context stacks.
\end{compactitem}

To ground the definitions of this section, consider again the semantic action of Listing~\ref{lst:ternary} and its CFG, shown in \Cref{fig:ternaryCFG} with each node labeled by a hexadecimal identifier. When \nlgcheck{} begins visiting this CFG, the initial state is $(\texttt{68d651f2}, \epsilon, \epsilon, \langle c_0 \rangle)$: \texttt{68d651f2} is $v_{entry}(G)$ (mapped to $Nop()$), the path walked so far and the edge tag stack are both empty, and $C=\langle c_0 \rangle$ contains only the root production context, with $m_0$ initially empty.

Items of the edge tag stack $S=s_0 s_1 \ldots s_n$ are edge tags ($s_i \in L$); moreover, if $S$ is empty, we say that $\nexists pop(S)$ and $\nexists top(S)$.
Instead, the production context stack $C=c_0 c_1 \ldots c_n$ always contains at least one item $c_0$ (representing the AST root), which consists solely of an attribute mapping function $m_0: \mathbb{N} \times I \rightarrow T$ (recall that $I$ is the set of valid attribute identifiers, \sect{ssect:complete}), mapping an index and an attribute name to the attribute's type: since the root has no parent to delegate to, $m_0$ stores both its own attributes (at index $0$) and those of its children (at index $\ge 1$).
For $i \ge 1$, the generic $c_i$ item represents a virtual AST node defined as $c_i = (k_i, m_i)$, a tuple where $k_i \in \mathbb{N}_{>0}$ indicates the position of the node in the child list of its parent, and $m_i: \mathbb{N}_{>0} \times I \rightarrow T$ is an attribute mapping function analogous to $m_0$, with the difference that $m_i$ only maps the attributes of the node's children (at index $\ge 1$): the node's own attributes (at index $0$) are instead stored in its parent's mapping function $m_{i-1}$, indexed by the node's own position $k_i$.

The traditional operations offered by the stack data structure are sufficient to express most of the actions associated with a CFG node as context-stack transformations.
One notable exception is the $PropagateAttrs()$ action, which needs to enumerate all the attributes within a node.
We therefore extend the instruction set with a $getAll: \mathcal{C} \times \mathbb{N} \rightarrow \mathcal{P}(I \times T)$ operation defined as
    \[
    getAll(C, i) =
        \begin{cases}
            \{(\eta, t): ((i, \eta) \mapsto t) \in m_0\} & \text{if } C=c_0\\
            \{(\eta, t): ((i, \eta) \mapsto t) \in m_n\} & \text{if } C=c_0 \ldots c_n \text{ and } i \ge 1\\
            \{(\eta, t): ((k_n, \eta) \mapsto t) \in m_{n-1}\} & \text{if } C=c_0 \ldots c_n \text{ and } i = 0
        \end{cases}
    \]
$getAll$ computes a set $X \subset (I \times T)$ of name-type pairs of the attributes at the $i$-th position of the current node.
$PropagateAttrs()$ also sets the attribute name-type pairs into the current node, without replacing pre-existing mappings; this behavior is represented by the $merge: (\mathbb{N} \times I \rightarrow T) \times \mathcal{P}(I \times T) \times \mathbb{N} \rightarrow (\mathbb{N} \times I \rightarrow T)$ and $putAll: \mathcal{C} \times \mathbb{N} \times \mathcal{P}(I \times T) \rightarrow \mathcal{C}$ operations:
    \[
    merge(m, X, i) = m \cup \{(i, \eta) \mapsto t: (\eta, t) \in (X \setminus \{(\eta', t'): ((i, \eta') \mapsto t') \in m\})\}
    \]
    \[
    putAll(C, i, X) =
        \begin{cases}
            merge(m_0, X, i) & \text{if } C=c_0\\
            c_0 \ldots c_{n-1} (k_n, merge(m_n, X, i)) & \text{if } C=c_0 \ldots c_n \text{ and } i \ge 1\\
            c_0 \ldots (k_{n-1}, merge(m_{n-1}, X, k_n)) c_n & \text{if } C=c_0 \ldots c_n \text{ and } i = 0
        \end{cases}
    \]
Since the $put$ only modifies at most the top of the context stack, let us introduce two more functions for brevity:
    \[
    c_{put1}(c_{n-1}, k_n, i, \eta, t)=
        \begin{cases}
            c_{n-1} & \text{if } i \ge 1\\
            (k_{n-1}, m_{n-1} \cup \{(k_n, \eta) \mapsto t\}) & \text{if } i = 0
        \end{cases}
    \]
    \[
    c_{put2}(c_n, i, \eta, t)=
        \begin{cases}
            (k_n, m_n \cup \{(i, \eta) \mapsto t\}) & \text{if } i \ge 1\\
            c_n & \text{if } i = 0
    \end{cases}
    \]
Notice that $put$ (not to be confused with the $push$ operation used by $EnterCtx$ to push a new production context onto the stack) can be defined in terms of these two functions:
    \[
    put(C, i, \eta, t) =
    \begin{cases}
            m_0 \cup \{(i, \eta) \mapsto t\} & \text{if } C=c_0\\
            c_0 \ldots c_{put1}(c_{n-1}, k_n, i, \eta, t) c_{put2}(c_n, i, \eta, t) & \text{if } C=c_0 \ldots c_n
    \end{cases}
    \]
The functions $getAll$, $merge$, and $putAll$ are only relevant to $PropagateAttrs()$, which does not appear in Listing~\ref{lst:ternary}. The function $put$, instead, is exercised by every $WriteAttr$ and $CopyAttr$ node of \Cref{fig:ternaryCFG}. For instance, node \texttt{302edb74} is mapped to $CopyAttr(0, val, 2, val)$: assuming the production context stack $C$ already has $Op[2].val$ set to some value of type $t$, evaluating this node computes $put(C, 0, val, t)$, storing that same value and type under attribute $val$ of $Op$ itself.

\subsubsection{Violations}\label{sect:contrib:violations}
We are now ready to define the kinds of incorrect or undesirable execution that can be detected while visiting the CFG, which we will call the set $\mathcal{V}$ of \textit{violations}.

\begin{violationbox}{$MissingAttr(i, \eta)$, with $i \in \mathbb{N}$ and $\eta \in I$}
 This violation indicates that a semantic action accesses an attribute named $\eta$ on the $i$-th nonterminal of its attached production that was not previously set.
\end{violationbox}
\begin{violationbox}{$BadAttrT(i, \eta, t_r, t_p)$, with $i \in \mathbb{N}$, $\eta \in I$ and $t_r, t_p \in T$}
    This violation indicates that a semantic action accesses an attribute named $\eta$ on the $i$-th nonterminal of its attached production and casts it to $t_r$ (\textit{required} type), which was previously set to a value of type $t_p$ (\textit{provided} type) not assignable to $t_r$ ($t_p \nleq: t_r$). While this is indeed the kind of error a static type system would normally rule out, \tool{Neverlang} has none spanning separately compiled modules: attributes are generated as plain \texttt{Object}-typed \tool{Java} fields, so the \tool{Java} compiler accepts any assignment across module boundaries, and only an explicit runtime cast, inserted by \tool{Neverlang} itself, can fail. $BadAttrT$ is \nlgcheck{}'s static substitute for that missing check.
\end{violationbox}
\begin{violationbox}{$AttrTypeDynamicallyChecked(i, \eta, t)$, with $i \in \mathbb{N}$, $\eta \in I$ and $t \in T$}
   This violation indicates that a semantic action performs an \texttt{instanceof} (or equivalent) check against $t$ of the type of the attribute named $\eta$ on the $i$-th nonterminal of its attached production. This violation does not cause a runtime error but is reported as a warning to suggest that the language developer should use \tool{Neverlang}'s type guards instead.
\end{violationbox}

Consider node \texttt{5c080ef3} of \Cref{fig:ternaryCFG}, which is mapped to $ReadAttr(1, val, Boolean)$: if the sub-expression bound to $Op[1]$ never sets an attribute named $val$, evaluating this node yields $MissingAttr(1, val)$; if $Op[1].val$ is instead set to a value of an incompatible type, say $Integer$, evaluating this node yields $BadAttrT(1, val, Boolean, Integer)$.

\subsubsection{Mapping CFG Actions to State Transformations}
Each of the aforementioned actions can be formally defined in terms of the transformations it performs on the state and the violations it detects: an action $a \in A$ does not manipulate or access the entire visit state, but rather just the production context stack $C$. We can therefore define $f_{action}$ as a function
$(\mathcal{C}, A) \rightarrow (\mathcal{C}, \mathcal{P}(\mathcal{V}))$
mapping a production context stack and an action to a new production context stack and a set of violations.

Here follows a list of the $f_{action}$ with respect to the actions that can appear in the CFG\@:
\begin{itemize}
    \item $f_{action}(C, WriteAttr(i, \eta, t)) = (put(C, i, \eta, t), \emptyset) $
    \item $f_{action}(C, ReadAttr(i, \eta, t)) =
        \begin{cases}
            (C, \{MissingAttr(i, \eta)\}) & \text{if } \nexists get(C, i, \eta) \\
            (C, \{BadAttrT(i, \eta, t, get(C, i, \eta))\}) & \text{if } get(C, i, \eta) \nleq: t\\
            (C, \emptyset) & \text{otherwise}
        \end{cases}$
    \item $f_{action}(C, CopyAttr(i_d, \eta_d, i_s, \eta_s)) =
        \begin{cases}
            (C, \{MissingAttr(i_s, \eta_s)\}) & \text{if } \nexists get(C, i_s, \eta_s) \\
            (put(C, i_d, \eta_d, get(C, i_s, \eta_s)), \emptyset) & \text{otherwise}
        \end{cases}$
    \item $f_{action}(C, CheckAttrType(i, \eta, t)) = (C, \{AttrTypeDynamicallyChecked(i, \eta, t)\})$
    \item $f_{action}(C, PropagateAttrs()) = (putAll(C, 0, getAll(C, 1)), \emptyset)$
    \item $f_{action}(C, EnterCtx(i)) = (push(C, (i, \emptyset)), \emptyset)$
    \item $f_{action}(C, LeaveCtx()) = (pop(C), \emptyset)$
    \item $f_{action}(C, Nop()) = (C, \emptyset)$
\end{itemize}
$BeginEvalMetaAction$ and $EndEvalMetaAction$ nodes, like $Nop()$, always leave $C$ unchanged and never produce violations: they only affect control flow (\textit{i.e.}, which node is visited next), which is instead governed by $f_{edge}$ (\sect{sect:ContribCfgVisitValidPaths}).

For example, consider the path through \Cref{fig:ternaryCFG} corresponding to the \texttt{then} branch, assuming $Op[1].val$ was previously set to \texttt{true}: nodes \texttt{68d651f2}, \texttt{1442f788}, and \texttt{113e13f9} (respectively $Nop()$, $BeginEvalMetaAction(1)$, and $EndEvalMetaAction()$) leave $C$ unchanged; node \texttt{5c080ef3} ($ReadAttr(1, val, Boolean)$) checks that $Op[1].val$ is set and assignable to $Boolean$; nodes \texttt{44e93c1f} and \texttt{42373389} (again $BeginEvalMetaAction(2)$ and $EndEvalMetaAction()$) leave $C$ unchanged; and node \texttt{302edb74} ($CopyAttr(0, val, 2, val)$) finally computes $put(C, 0, val, t)$ as illustrated above.

Note that, unlike $ReadAttr$, $CopyAttr$ can only produce a $MissingAttr$ violation, never a $BadAttrT$ one. This is not an oversight: $ReadAttr(i, \eta, t)$ carries an explicit \textit{required} type $t$ against which the attribute's actual type is checked, whereas $CopyAttr(i_d, \eta_d, i_s, \eta_s)$ has no such parameter. As already noted in \sect{ssect:complete}, it copies the attribute's value \textit{and} its type verbatim, so there is no declared type it could mismatch against.
This mapping from CFG actions to context-stack transformations and detected violations can be seen as a form of \textit{collecting semantics} in the sense of abstract interpretation~\cite{Cousot96}: each $f_{action}$ soundly over-approximates the concrete runtime behavior of its corresponding instruction with respect to the violations of interest.
The result of an $f_{action}$ is made of two elements. We refer to the \textit{next context} element of $f_{action}$ with $f_{aCtx}$: 
\[f_{aCtx}(C, a)=C' \iff \exists U: \: f_{action}(C, a)=(C',U)\]
We refer to the \textit{violations} part of $f_{action}$ with $f_{aErr}$: 
\[f_{aErr}(C, a)=U \iff \exists C': \: f_{action}(C, a)=(C',U)\]
Finally, we define $f_{pCtx}: (V_r^* \times \mathcal{C}) \rightarrow
\mathcal{C}$ as the overall changes made to the production context stack during the traversal of a CFG path:
\[
    f_{pCtx}(w, C)=
        \begin{cases}
            C & \text{if } w = \epsilon \\
            f_{pCtx}(w', f_{aCtx}(C, a_{node}(v))) & \text{if } w = vw'
    \end{cases}
\]

\begin{property}\label{PropfpctxComposition}
$f_{pCtx}$ can be evaluated \textit{piecewise} on a path by composing evaluations of $f_{pCtx}$ on parts of the path: 
    \[
        \forall w_1, w_2 \in V_r^*, C \in \mathcal{C}, \quad
        f_{pCtx}(w_1w_2, C)=f_{pCtx}(w_2, f_{pCtx}(w_1, C))\]
\end{property}
\begin{propertyproof}
    This can be proven by induction on $w_1$.
    The base case $w_1=\epsilon$ is trivial, as by definition $f_{pCtx}(\epsilon, C) = C$.
    For the inductive step, given $w_1$, we assume
    \[
        \forall w_2 \in V_r^*, C' \in \mathcal{C}, \quad f_{pCtx}(w_1w_2, C')=f_{pCtx}(w_2, f_{pCtx}(w_1, C'))
    \]
    and prove that
    \[
        \forall v \in V_r, w_2 \in V_r^*, C \in \mathcal{C}, \quad f_{pCtx}(vw_1w_2, C)=f_{pCtx}(w_2, f_{pCtx}(vw_1, C))
    \]
    using the definition of $f_{pCtx}$ and the inductive hypothesis with $C' = f_{aCtx}(C, a_{node}(v))$.
    Please refer to \sect{sec:app:proofs} for the complete, step-by-step proof.
\end{propertyproof}

\subsubsection{Action Properties}\label{ContribCfgVisitActionProps}
All actions, except for $EnterCtx(i)$ or $LeaveCtx()$, only depend on and affect the contents of the two topmost productions of the stack. More specifically, it can depend on the attributes present at any position in the node at the top of the stack and on those on the node at position 0 of the context right under the top of the stack.
We refer to this information as the $head$ of the production context stack:
\[
    head(C)=
        \begin{cases}
            c_0 & \text{if } C=c_0 \\
            m_n \cup \{(0, \eta) \mapsto t: ((k_n, \eta) \mapsto t) \in m_{n-1}\} & \text{if } C=c_0 \ldots c_{n}
        \end{cases}
\]
Since \Cref{fig:ternaryCFG} shows the CFG of a single semantic action---not yet connected to any other via $EnterCtx$/$LeaveCtx$, which only appear once semantic actions are composed into the full CFG $G_r$ of a role (\sect{sect:contrib:roles})---the production context stack $C$ remains $\langle c_0 \rangle$ throughout its visit. Consequently, $head(C)=c_0=m_0$ at every node of \Cref{fig:ternaryCFG}.
\begin{property}\label{PropHeadDefImmediate}
The head of a production context stack only depends on its two topmost items:
\begin{gather*}
    \forall\:C_1=\langle c_0, \ldots, c_n\rangle, C_2=\langle c_0', \ldots, c_n'\rangle \in \mathcal{C}, \\
    c_{n-1} = c_{n-1}' \:\land\: c_n=c_n' \implies head(C_1) = head(C_2)
\end{gather*}
\end{property}
\begin{propertyproof}
    Trivial by definition of head.
\end{propertyproof}

\begin{property}\label{Propfactionhead}
$f_{aErr}$ only depends on the head of the production context stack:
\[
\forall C_1, C_2 \in \mathcal{C}, a \in A, \quad
head(C_1) = head(C_2) \implies f_{aErr}(C_1, a) = f_{aErr}(C_2, a)
\]
\end{property}
\begin{propertyproof}
    By definition of $f_{action}$ and $a_{aErr}$, if $a$ is one among $WriteAttr$, $PropagateAttrs$, $EnterCtx$, $LeaveCtx$, or $Nop$, then $f_{aErr}(C, a) = \emptyset \:\: \forall C$, trivially proving the property.
    Similarly, if $a=CheckAttrType(i, \eta, t)$, then $f_{err}(C, a)$ is independent of $C$.
    Finally, if $a=ReadAttr$ or $a=CopyAttr$, the contents of $f_{aErr}(C, a)$ depend on the value of a $get$ operation performed on $C$: the value of $get$ applied to a context stack $\langle c_0, \ldots, c_n \rangle$ depends on either $m_n$ or $m_{n-1}(k_n, \eta)$, both of which are equal in $C_1$ and $C_2$ under the hypothesis $head(C_1) = head(C_2)$ by definition of $head$.
\end{propertyproof}

\subsubsection{Balanced Action Sequences}\label{ContribCfgVisitBalActionSeq}
Just like tagged edges (\sect{sect:ContribCfgVisitValidPaths}) model call/return as matching brackets on the edge tag stack, actions performed on the \textit{production context} stack lend themselves to the same analogy, although here only two actions are relevant: $EnterCtx(i)$ (corresponding to opening a parenthesis, since it pushes a new production context) and $LeaveCtx()$ (closing the parenthesis, since it pops one). Just as in a well-parenthesized expression every closing bracket matches the nearest still-open one, every $LeaveCtx()$ in a balanced action sequence pops exactly the production context that the nearest preceding, still-unmatched $EnterCtx(i)$ had pushed.
We define the \textit{set of balanced action sequences}, $\mathcal{A}_b$, with $\mathcal{A}_b \subset A^*$ as follows:
\begin{compactenum}
    \item the empty sequence is balanced ($\epsilon \in \mathcal{A}_b$);
    \item given $b \in \mathcal{A}_b$ and $a \in A \setminus (\{EnterCtx(i): i \in \mathbb{N}\} \cup \{LeaveCtx()\})$, $ab \in \mathcal{A}_b$ and $ba \in \mathcal{A}_b$;
    \item given $a_1=EnterCtx(i)$ for any $i$, $a_2=LeaveCtx()$, and $b \in \mathcal{A}_b$, $a_1ba_2 \in \mathcal{A}_b$.
\end{compactenum}
We also define the set $\mathcal{A}_{pb}$ of \textit{prefixes} of balanced action sequences as follows:
\begin{compactenum}
    \item the empty sequence is a prefix of a balanced action sequence ($\epsilon \in \mathcal{A}_{pb}$);
    \item given $b \in \mathcal{A}_{pb}$ and $a \in A \setminus (\{LeaveCtx()\})$, $ab \in \mathcal{A}_{pb}$ and $ba \in \mathcal{A}_{pb}$;
    \item given $a_1=EnterCtx(i)$ for any $i$, $a_2=LeaveCtx()$, and $b \in \mathcal{A}_{pb}$, $a_1ba_2 \in \mathcal{A}_{pb}$.
\end{compactenum}
\begin{property}
    $\mathcal{A}_b \subset \mathcal{A}_{pb} \subset A^*$
\end{property}
\begin{propertyproof}
    Trivial by definition of $\mathcal{A}_b$ and $\mathcal{A}_{pb}$.
\end{propertyproof}

\subsubsection{Balanced Paths}

For the application of balanced action sequences in CFG paths, let us also introduce the concept of \textit{balanced paths}---\textit{i.e.}, a sequence of nodes whose associated actions form a balanced action sequence.
We call $W_b$ the set of balanced paths, with $W_b \subset V_r^*$, such that:
\[
    W_b = \{(v_1 \ldots v_n): \: (a_1 \ldots a_n) \in \mathcal{A}_b, \text{ with }
    a_i = a_{node}(v_i) \: \forall i\}
\]
We define the set of balanced path prefixes, $W_{pb}$, with $W_b \subset W_{pb} \subset V_r^*$, as
\[
    W_{pb} = \{(v_1 \ldots v_n): \: (a_1 \ldots a_n) \in \mathcal{A}_{pb},
    \text{ with } a_i = a_{node}(v_i) \: \forall i\}
\]

We define a generalization, $top_k$, of the $top$ stack operation, which allows \textit{peeking} at the $k$ topmost items of the stack: 
\[
    top_k(\langle c_0, \ldots, c_n \rangle) = \langle c_{n-k+1}, \ldots, c_n \rangle \quad \text{with } 1 \le k \le n + 1
\]
\begin{property}
    $top$ and $top_1$ are equivalent.
\end{property}
\begin{propertyproof}
    Trivial by definition of $top_k$.
\end{propertyproof}

\begin{property}\label{PropBalancedPath0}
    Given an integer $k>0$ and two production context stacks sharing the same top $k$ items, the evaluation context stacks obtained by evaluating a balanced path prefix that does not contain $CopyAttr$ or $PropagateAttrs$ actions share the top $k$ items:
    \begin{gather*}
        \forall w \in W_{pb}, \forall C_1, C_2 \in \mathcal{C}, k \in \mathbb{N}\\
        top_k(C_1)=top_k(C_2) \implies top_k(f_{pCtx}(w, C_1))=top_k(f_{pCtx}(w, C_2))
    \end{gather*}
\end{property}
\begin{propertyproof}
    This property is proven by induction on each of the following four cases:
    \begin{compactenum}
        \item $w = \epsilon$;
        \item $w=vw'$, with $w' \in W_{pb}$ and $a_{node}(v) = a$ with $a \in A \setminus \{LeaveCtx()\}$;
        \item $w=w'v$, with $w' \in W_{pb}$ and $a_{node}(v) = a$ with $a \in A \setminus \{LeaveCtx()\}$;
        \item $w=v_1w'v_2$ with $w' \in W_{pb}$, $a_{node}(v_1)=a_1=EnterCtx(i)$, and $a_{node}(v_2)=a_2=LeaveCtx()$.
    \end{compactenum}
    Please refer to \sect{sec:app:proofs} for the complete proof.
\end{propertyproof}

\begin{property}\label{PropBalancedPath0b}
Given two production context stacks with the same $head$, evaluating a balanced path prefix without $CopyAttr$ or $PropagateAttrs$ actions produces results with the same $head$.
    \begin{gather*}
        \forall w \in W_{pb}, \forall C_1, C_2 \in \mathcal{C}\\
        head(C_1)=head(C_2) \implies head(f_{pCtx}(w, C_1))=head(f_{pCtx}(w, C_2))
    \end{gather*}
\end{property}
\begin{propertyproof}
    This property is similar to \Cref{PropBalancedPath0} with $k=2$.
    All considerations made on $top_3(C) \: \forall C$ also hold for the conjunction to $top_2(C)$ and $head(pop(C))$.
    When $head(C_1)=head(C_2)$, we have that
    \[head(put(C_1,i,\eta,t)) = head(put(C_2,i,\eta,t))\]
    With $C_1=\langle c_0, \ldots, c_n\rangle$ and $C_2=\langle c_0', \ldots, c_n' \rangle$, that can be expanded as
    \begin{gather*}
        head(\langle c_1,\ldots,c_{put1}(c_{n-1},k_n,i,\eta,t),c_{put2}(c_n,i,\eta,t)\rangle)=\\
        head(\langle c_1',\ldots,c_{put1}(c_{n-1}',k_n',i,\eta,t),c_{put2}(c_n',i,\eta,t)\rangle)
    \end{gather*}
    This can be proven by definition of $head$, since $m_{put2}(c_n,i,\eta,t) = m_{put2}(c_n',i,\eta,t)$, and by considering that the mappings for $k_n$ contained in $m_{put1}(c_{n-1},k_n,i,\eta,t)$ are the same as those for $k_n'$ contained in $m_{put1}(c_{n-1}',k_n',i,\eta,t)$, in which $m_{put1}$ and $m_{put2}$ are used to denote the "$m$-part" of $c_{put1}$ and $c_{put2}$, respectively. 

    Then, the proof of \Cref{PropBalancedPath0} can be adapted with $k=2$, using $head$ in place of $top_2$, proving this property.
\end{propertyproof}

\begin{property}\label{PropBalancedPath1}
Given any two production context stacks, evaluating a balanced path prefix that does not contain $CopyAttr$ or $PropagateAttrs$ actions, wrapped between two $EnterCtx$ nodes, produces context stacks with the same head.
    \[
        \forall C_1, C_2 \in \mathcal{C}, w \in W_{pb}, \quad v_{sep} \in V_r^*:\: \exists i: a_{node}(v_{sep}) = EnterCtx(i),
    \]
    \[
        head(f_{pCtx}(v_{sep}wv_{sep}, C_1))=head(f_{pCtx}(v_{sep}wv_{sep}, C_2))
    \]
\end{property}
\begin{propertyproof}
    The proof follows the same structure as that of \Cref{PropBalancedPath0}, according to the four possible definitions of a balanced path. Please refer to \sect{sec:app:proofs} for the full proof.
\end{propertyproof}

\begin{property}\label{PropBalancedPath2}
    Given an initial production context stack $C \in \mathcal{C}$, two balanced path prefixes $w_1,w_2 \in W_{pb}$, and a \textit{production entry node} $v_{sep}$ mapped to an $EnterCtx$ action, evaluating either $v_{sep}w_1v_{sep}w_2v_{sep}$ or $v_{sep}w_2v_{sep}$ on $C$ produces resulting context stacks with the same head. 
    \begin{gather*}
        \forall C \in \mathcal{C}, \: w_1,w_2 \in W_{pb}, \quad v_{sep} \in V_r^*: \: \exists i: a_{node}(v_{sep}) = EnterCtx(i),\\
        head(f_{pCtx}(v_{sep}w_1v_{sep}w_2v_{sep}, C)) = head(f_{pCtx}(v_{sep}w_2v_{sep}, C))
    \end{gather*}
\end{property}
\begin{propertyproof}
    This property is a trivial consequence of \Cref{PropBalancedPath1}. By definition of $f_{pCtx}$, the goal can be rewritten as
    \[
        head(f_{pCtx}(v_{sep}w_2v_{sep}, f_{pCtx}(v_{sep}w_1, C))) = head(f_{pCtx}(v_{sep}w_2v_{sep}, C))
    \]
    which is proven by applying \Cref{PropBalancedPath1} with $C_1 = f_{pCtx}(v_{sep}w_1)$, $C_2=C$, and $w=w_2$.
\end{propertyproof}

\subsubsection{Following Valid Paths}\label{sect:ContribCfgVisitValidPaths}
In \sect{sect:contrib:roles}, we discussed the problem of correctly modeling control-flow transfers between different semantic actions in a CFG by representing transitions that can happen at runtime, as well as the concepts of \textit{tagged entry and exit edges} and \textit{valid path}.

The \textit{edge tag stack} $S$, previously introduced in this section as part of the state of a CFG visit, is used to keep track, at any given point of the visit, of the tagged entry edges for which a corresponding tagged exit edge has not yet been traversed.

We define the $f_{edge}$ function as a way to drive the algorithm through valid paths according to these concepts.
It determines if an edge can be followed and updates $S$ if needed.
Let $L^*$ be the set of all possible edge tag stacks; $f_{edge}: (L^* \times E_r) \rightarrow L^*$ is defined as
    \[
        f_{edge}(S, e)=
            \begin{cases}
                S & \text{if } e \notin (E_{entry} \cup E_{exit}) \\
                push(S, l_{edge}(e)) & \text{if } e \in E_{entry} \\
                pop(S) & \text{if } e \in E_{exit} \land l_{edge}(e) = top(S)
            \end{cases}
    \]
Therefore, $f_{edge}$ is not defined if following the edge would lead to an invalid path.
Otherwise, $f_{edge}$ returns the updated edge tag stack after the edge is traversed.

To illustrate tagged edges and recursion concretely, reconsider Listing~\ref{lst:ternary} and \Cref{fig:ternaryCFG} ($G$), and recall that \texttt{Expression} occurs recursively on the right-hand side of the \texttt{Op} production, so a ternary expression's branches may themselves be ternary expressions. When building the role's CFG $G_r$ (\sect{sect:contrib:roles}), if $Op[2]$ happens to be bound to another ternary expression, node \texttt{44e93c1f} (mapped to $BeginEvalMetaAction(2)$) becomes $v_{beginEval}$ of a tagged entry edge to $v_{entry}(G)=\texttt{68d651f2}$---connecting $G$ to itself, since the callee is compiled from the same production. Let $l$ be the tag statically assigned to this entry edge and to its corresponding tagged exit edge $(\texttt{3ad4a7d6}, \texttt{42373389})$.

Now consider the doubly nested expression \inlinec{a ? (b ? (f ? g : h) : d) : e}: visiting its outermost \texttt{Op} node traverses the tagged entry edge $(\texttt{44e93c1f}, \texttt{68d651f2})$ once, pushing $l$ onto $S$, to visit the middle \texttt{Op}; it is then traversed a second time to visit the innermost \texttt{Op}, at which point $l$ occurs twice on $S$. A third level of nesting would traverse the same tagged entry edge a third time---which Constraint (a) forbids. This is not because a triply nested ternary cannot be validated, but because, by \Cref{wasatisfies}, any violation only reachable via a third traversal of $(\texttt{44e93c1f}, \texttt{68d651f2})$ is also reachable via an equivalent path in which it is traversed just twice. 

\subsubsection{Solution to the Problem of Infinite Paths}\label{sect:ContribInfinitePathsSolution}
Most practical context-free grammars are recursive.
A parser built from a recursive grammar can parse an infinite set of input strings; therefore, the number of parse trees that can be generated is also infinite.
Even though, according to our model, different parse trees can be mapped to the same CFG visit path, the paths themselves are potentially infinite nonetheless.\footnote{Two trivial examples are when a semantic action may call itself, possibly looping any number of times, or when the semantic action itself contains a loop that may \textit{eval} another action any number of times.}

Thus, it is fundamental to determine a finite subset of all possible paths to render the static analysis feasible.

\mypar{Na\"{\i}ve solutions}
By keeping track of the set of already visited nodes, one can prevent visiting the same node more than once.
It can be trivially shown that this approach can easily cause inconsistency, leaving the violations presented in \sect{sect:contrib:violations} undetected even in small CFGs: take, for instance, an \texttt{if-then} statement in which an attribute is set within the body when the condition is true; then, the same attribute is accessed outside. This should lead to an error, as the attribute is not set unless the condition holds.
It is possible to find similar counterexamples even when keeping track of all navigated edges instead of nodes. Thus, na\"{\i}ve solutions do not suffice, and a more complex one must be adopted.

\mypar{Relation to classical circular-attribute handling} The problem above resembles, but is not the same as, the handling of \textit{circular attributes} in classical AG systems, where the grammar's declared attribute dependencies form a cyclic graph. Two established strategies address that problem. \textit{Ordered attribute grammars}~\cite{Kastens80} forbid such cycles outright: from the grammar's declared dependencies alone, a single visit sequence is computed once, statically, guaranteeing that every attribute is assigned exactly once before it is ever read, so that no repeated computation is needed at evaluation time. \textit{Circular attribute grammars}~\cite{Farrow86}, by contrast, allow genuine cycles and resolve them by repeatedly re-evaluating the cyclic attributes' equations, for each concrete input tree, until their values converge to a fixed point.

Both strategies operate on attribute \textit{values}: an ordered AG's static ordering is a plan for computing them without redundant work, while a circular AG's fixed-point iteration is a mechanism for actually computing them despite the redundancy. The problem tackled in this section is different in kind, not merely in degree: \nlgcheck{} never computes an attribute's value, only whether every read along a path is preceded by a matching write (\sect{ssect:complete}). What may be unboundedly long or numerous here are the CFG paths themselves, purely as a consequence of the grammar---and, transitively, its compiled semantic actions---being recursive, regardless of whether any individual attribute's dependencies are themselves circular. Consequently, neither a static evaluation ordering nor a fixed-point value iteration directly applies: there is no fixed point to converge to, since nothing is being evaluated. Constraint (a), introduced below, instead bounds how many times the \textit{same} recursive call site may be revisited while exploring a single path---the closest analogue, in this setting, to the termination arguments used for ordered and circular AGs, even though it addresses a structurally different problem.

\mypar{Complete solution}
Let $W \subseteq V_r^*$ be the (potentially infinite) set of valid paths. The goal is to find a finite set $W_a \subseteq W$ such that for every node $v_e$, if there exist some paths from the entry node to $v_e$ that may cause violations, then at least one such path is contained in $W_a$.

Let us introduce a function, $f_{pErr}$, that computes the set of violations detected within a path.
Given a path $w \in V_r^*$ and a production context stack $C \in \mathcal{C}$, $f_{pErr}: (V_r^* \times \mathcal{C}) \rightarrow \mathcal{P}(\mathcal{V})$ is defined as follows:
    \[
        f_{pErr}(w, C)=
            \begin{cases}
                \emptyset & \text{if } w = \epsilon \\
                f_{aErr}(f_{pCtx}(w', C), a_{node}(v)) & \text{if } w=w'v
            \end{cases}
    \]
Notice that $f_{pErr}$ \textit{discards} any violations detected prior to reaching the last node. As we will show with the following properties and requirements, this does not affect the validity of the solution. 

\begin{property}\label{PropfperrHead}
    If $w$ is a prefix of a balanced path, then $f_{pErr}(w, C)$ only depends on the $head$ of $C$.
    \[
        \forall w \in W_{pb}, C_1, C_2 \in \mathcal{C}, \quad head(C_1) = head(C_2) \implies f_{pErr}(w, C_1) = f_{pErr}(w, C_2)
    \]
\end{property}
\begin{propertyproof}
    If $w=\epsilon$, the proof is trivial, as by definition, $f_{pErr}(\epsilon, C) = \emptyset \: \forall C$.

    If $w = w'v$, the goal is to prove that
    \[
        f_{aErr}(f_{pCtx}(w', C_1), a_{node}(v)) =f_{aErr}(f_{pCtx}(w', C_2), a_{node}(v))
    \]
    For defining a prefix of a balanced path, and since $w \in W_{pb}$, we can conclude that $w' \in W_{pb}$.
    Therefore, by \Cref{PropBalancedPath0b} and from the assumption $head(C_1) = head(C_2)$ follows that $head(f_{pCtx}(w', C_1)) = head(f_{pCtx}(w', C_2))$.
    Applying \Cref{Propfactionhead} with the context stacks resulting from $f_{pCtx}(w', C_1)$ and $f_{pCtx}(w', C_2)$ as $C_1$ and $C_2$, respectively, and $a_{node}(v)$ as $a$ concludes the proof.
\end{propertyproof}

For conciseness, let us also define $v_0=v_{entry}(G_r)$ and $f_e(w)=f_{pErr}(w, \langle\emptyset\rangle)$.
Then:
    \begin{equation}\label{warequirement}
    \forall v_e \in V_r, \quad \exists w_e=v_0 \ldots v_e: f_e(w_e) \neq \emptyset \implies \exists w_e'=v_0 \ldots v_e \in W_a: f_e(w_e') \neq \emptyset
\end{equation}
Note that a path in $W_a$ does not \textit{necessarily} lead to a violation.\footnote{The implication is not bidirectional.}

There are a few constraints that the paths must adhere to, to satisfy these requirements while keeping $W_a$ finite:

\mypar{Constraint (a): bounded recursion} A tagged entry edge is not followed if its tag is already present twice on the current edge tag stack.\footnote{Notice that this condition is less strict than that of only following any given tagged entry edge at most twice in a path: if $e_1 \in E_{entry}$ is followed, and then in the same path $e_2 \in E_{exit}$ such that $l_{edge}(e_2) = l_{edge}(e_1)$ is also followed, $e_1$ can be followed again twice.} Formally, we define a function $f_{stack}: (V_r^* \times L^*) \rightarrow L^*$ representing how traversing a path mutates the edge tag stack.
    \[
        f_{stack}(w, S)=
        \begin{cases}
            S & \text{if } |w| \le 1 \\
            f_{stack}(v_2w', f_{edge}(S, (v_1, v_2))) & \text{if } w = v_1v_2w'
        \end{cases}
    \]
    Thus, Constraint (a) is expressed as:
    \begin{equation}\label{pathcond1}
        W_a'=\{w \in V_r^*: \quad \nexists l_1,l_2,l_3 \in f_{stack}(w, \epsilon): l_1=l_2=l_3 \}
    \end{equation}
    Intuitively, a tag being present once on the stack tests invocation, a tag being present twice tests recursion, whereas three or more occurrences do not provide additional information: this last point is not merely an empirical assumption but is formally established by \Cref{wasatisfies} below, which proves that any path relying on a third occurrence of the same tag to reach a violation can always be shortened to an equivalent path in which the tag occurs only twice, without losing the violation. This is what justifies cutting off recursion at depth $2$ rather than $1$ (which would miss violations that only manifest across a recursive call) or some higher depth (which would be unnecessary, and would make $W_a'$ needlessly larger).
    Note that Constraint (a) can only ever discard paths from the (infinite) set $V_r^*$ of already-valid CFG paths---it never introduces a path that would not otherwise correspond to a real execution. Consequently, restricting the search to $W_a'$ can only cost completeness (a real violation being missed because every path exhibiting it was excluded), never soundness: whenever a violation is reported along some $w \in W_a'$, $w$ is by construction a genuine valid path of the CFG, so the violation it exposes is a real one, not an artifact of the depth cutoff. \Cref{wasatisfies} then rules out the completeness risk as well, by showing that no violation is exclusively reachable through a third or later occurrence of a tag.

\mypar{Constraint (b): no redundant revisits} A node is not visited again if it was previously visited with the same production context stack and edge tag stack as the current ones.
    Formally:
    \[
        W_a=\{w=v_0 \ldots v_n \in W_a': \quad \forall w'=v_0' \ldots v_n' \in W_a' \text{ such that } w' \neq w \land v_n'=v_n,
    \]
    \begin{equation}\label{pathcond2}
        \quad f_{pCtx}(w, \langle\emptyset\rangle) \neq f_{pCtx}(w', \langle\emptyset\rangle) \lor f_{stack}(w, \epsilon) \neq f_{stack}(w', \epsilon)\}
    \end{equation}

It can be observed that $W_a \subseteq W_a'$.

We now prove two lemmas demonstrating that under the assumption that no $CopyAttr$ or $PropagateAttrs$ nodes are present in the graph (which is a precondition to several properties proven thus far), visiting all the paths contained in the set $W_a$ determined by the constraints above is sufficient to find a violation within a path, if any exists.

A limitation of this approach is that we cannot prove that the same property holds when $CopyAttr$ or $PropagateAttrs$ actions are used. However, we deemed this limitation acceptable for the following reasons.
\begin{compactenum}
    \item While lacking a formal proof, experimentation (see \sect{sec:evaluation}) demonstrated that even when such actions are used, the software is able to find most violations anyway in practice.
    \item All $PropagateAttrs$ actions can always be replaced by copying relevant attributes from child notes to parent nodes explicitly, instead of relying on this quality-of-life feature. Such a feature is not specific to \tool{Neverlang}: other AG systems provide comparable means of automatically propagating attributes without a dedicated equation for each one, such as the \texttt{CHAINS} construct of \tool{Eli}~\cite{Gray92}. Similarly, $CopyAttrs$ actions can also be replaced by other, more explicit, actions: semantic instructions of the form \texttt{\$0.attrA = \$1.attrB} can be replaced with \texttt{\$0.attrA = (T) \$1.attrB}, where $T$ is the expected type of the attribute, thus modifying the CFG accordingly, substituting the single $CopyAttr(0, attrA, 1, attrB)$ node with the two nodes $ReadAttr(1, attrB, T)$ and $WriteAttr(0, attrA, T)$ in succession. In other words, the problem is solved by rendering the type of the attribute explicit. By taking these precautions, the approach is formally sound and can always find any violations.
\end{compactenum}

\begin{mainthm}\label{wasatisfies}
    $W_a'$, restricted by Constraint (a) alone as defined by~\eqref{pathcond1}, satisfies requirement~\eqref{warequirement}. In fact, the proof establishes the stronger property that $w_e'$ can always be chosen such that $f_e(w_e')=f_e(w_e)$, \textit{i.e.}, $w_e'$ detects exactly the same violations as $w_e$, not merely \textit{some} violation.
\end{mainthm}
\begin{proof}
    Let us consider a generic node $v_e \in V_r$; we must prove that if paths from the entry node to $v_e$ with a violation exist, then at least one of them is contained within $W_a'$, so that it satisfies~\eqref{pathcond1}. So, let's assume that
    \begin{equation}\label{wasatisfies_hyp}
        \exists w_e=v_0 \ldots v_e: f_e(w_e) \neq \emptyset
    \end{equation}
    and proceed by proving that
    \[
        \exists w_e'=v_0 \ldots v_e: \quad f_e(w_e') \neq \emptyset \: \land \: w_e' \in W_a'
    \]
    which, due to~\eqref{pathcond1}, can be rewritten as
    \[
        \exists w_e'=v_0 \ldots v_e: \quad f_e(w_e') \neq \emptyset \: \land \: \nexists l_1,l_2,l_3 \in f_{stack}(w_e', \epsilon): l_1=l_2=l_3
    \]
    Proceeding by contradiction, assume the goal is false. Without loss of generality, we consider that evaluating \textit{any} path from the entry node $v_0$ to $v_e$ that leads to the detection of a violation at $v_e$ requires pushing the same edge tag on the stack at least three times. Then:
    \begin{equation}\label{wasatisfies_contra}
        \forall w_e'=v_0 \ldots v_e: \quad f_e(w_e') = \emptyset \: \lor \: \exists l_1,l_2,l_3 \in f_{stack}(w_e', \epsilon): l_1=l_2=l_3
    \end{equation}
    With $w_e'=w_e$ and due to~\eqref{wasatisfies_hyp}
    \[
        \exists l_1,l_2,l_3 \in f_{stack}(w_e, \epsilon): l_1=l_2=l_3
    \]
    or, in words, the edge tag stack necessarily contains the same tag $l$ such that $l=l_1=l_2=l_3$ at least three times after traversing $w_e$.

    We show the contradiction by proving that if $l$ is present on the final edge tag stack exactly three or more times, then part of the path can be removed.

    By construction of the CFG and by definition of $f_{edge}$ and $f_{stack}$, condition~\eqref{wasatisfies_contra} occurs when an entry edge $(v_1, v_2) \in E_{entry}$, such that $l_{edge}((v_1,v_2))=l$, is traversed exactly three times in $w_e$ without the corresponding exit edge $(v_3, v_4) \in E_{exit}$ such that $l_{edge}((v_3, v_4))=l$ ever being traversed.
    Therefore, the (possibly empty) $w_1,w_2,w_3,w_4$ paths exist, such that
    \[
        w_e=v_0w_1v_1v_2w_2v_1v_2w_3v_1v_2w_4v_e
    \]
    and such that none between $w_2, w_3$, and $w_4$ contain the vertex pair $v_3v_4$. Recall that only valid paths (as by definition provided in \sect{sect:contrib:roles}) are considered; thus, $w_e$ is a valid path. It follows that $w_2$, $w_3$, and $w_4$ must also be valid paths, \textit{i.e.}, they cannot contain exit edges without a corresponding entry edge. Upon traversing each of them, the top of the edge tag stack is $l$; therefore, if an exit edge had to be followed, its tag would also need to be $l$, but this would mean that the edge is $(v_3, v_4)$, which is not allowed under the current hypothesis.

    It is true by CFG construction that each entry edge always starts from a node mapped to an $EnterCtx$ action, and also that each exit edge always points to a node mapped to a $LeaveCtx$ action. The inverse is also true, \textit{i.e.}, all outgoing edges from an $EnterCtx$ node are entry edges, and all incoming edges to a $LeaveCtx$ node are exit edges.

    We can therefore conclude that neither $w_2$, $w_3$, nor $w_4$ can contain $LeaveCtx$ nodes unless they also contain a matching $EnterCtx$ node, as by definition of \textit{balanced path prefix} given in \sect{ContribCfgVisitBalActionSeq}; therefore, $w_2, w_3, w_4 \in W_{pb}$.

    We expand $f_e(w_e)$ as follows:
    \begin{align*}
        &f_e(w_e)  = &(\text{by definition of } f_e) \\
        &=f_{pErr}(w_e, \langle\emptyset\rangle)= & (\text{by definition of } f_{pErr})\\
        &=f_{pErr}(v_0w_1v_1v_2w_2v_1v_2w_3v_1v_2w_4v_e, \langle\emptyset\rangle)= & (\text{by \Cref{PropfpctxComposition}})\\
        &=f_{pErr}(v_2w_4v_e, f_{pCtx}(v_0w_1v_1v_2w_2v_1v_2w_3v_1, \langle\emptyset\rangle))= & (\text{by \Cref{PropfpctxComposition}})\\
        &=f_{pErr}(v_2w_4v_e, f_{pCtx}(v_1v_2w_2v_1v_2w_3v_1, f_{pCtx}(v_0w_1, \langle\emptyset\rangle))) &
    \end{align*}
    Recall that $(v_1, v_2)$ is an entry edge and thus:
    \begin{compactenum}
        \item $v_1$ is mapped to an $EnterCtx$ action and is therefore compatible with the requirements $v_{sep}$ in \Cref{PropBalancedPath2};
        \item $v_2$ cannot be a $LeaveCtx$ action---as it is the entry node of a semantic action---and therefore $v_2w_2, v_2w_3 \in W_{pb}$.
    \end{compactenum}
    We can then apply \Cref{PropBalancedPath2} with $C=f_{pCtx}(v_0w_1, \langle\emptyset\rangle)$, $v_{sep}=v_1$, $w_1=v_2w_2$, and $w_2=v_2w_3$ to obtain
    \[
        head(f_{pCtx}(v_1v_2w_2v_1v_2w_3v_1, f_{pCtx}(v_0w_1,\langle\emptyset\rangle)))=head(f_{pCtx}(v_1v_2w_3v_1, f_{pCtx}(v_0w_1,\langle\emptyset\rangle)))
    \]
    Since $v_2$ is not a $LeaveCtx$ action, we know also that $v_2w_4 \in W_{pb}$, whereas from the hypothesis $f_e(w_e)\neq\emptyset$, we know that $a_{node}(v_e) \neq LeaveCtx()$, as $f_{aErr}(C, LeaveCtx())=\emptyset \:\: \forall C$; thus, $v_2w_4v_e \in W_{pb}$.
    We can then apply \Cref{PropfperrHead} to obtain
    \[
        f_e(w_e)=f_{pErr}(v_2w_4v_e, f_{pCtx}(v_1v_2w_3v_1, f_{pCtx}(v_0w_1, \langle\emptyset\rangle)))
    \]
    from which, by \Cref{PropfpctxComposition} and definition of $f_{pErr}$, follows that
    \begin{gather*}
        f_e(w_e)=f_{pErr}(v_2w_4v_e, f_{pCtx}(v_0w_1v_1v_2w_3v_1, \langle\emptyset\rangle))=f_{pErr}(v_0w_1v_1v_2w_3v_1v_2w_4v_e, \langle\emptyset\rangle)
    \end{gather*}
    Let $w_e' \coloneq v_0w_1v_1v_2w_3v_1v_2w_4v_e$ for brevity: by definition of $f_e$, the above chain of equalities states exactly that $f_e(w_e)=f_e(w_e')$---not merely that both are non-empty, but that $w_e'$ detects the same set of violations as $w_e$.
    In particular, since $f_e(w_e)\neq\emptyset$ by hypothesis~\eqref{wasatisfies_hyp}, it follows that $f_e(w_e')\neq\emptyset$ as well, and $w_e'$ only contains the edge $(v_1, v_2)$ twice; therefore, after traversing $w_e'$, the tag $l$ is only present on the edge tag stack twice.
    This conclusion contradicts the hypothesis of the \textit{reductio ad absurdum}~\eqref{wasatisfies_contra}, proving that a generic path leading to the detection of a violation can be reduced to a path, satisfying the requirement for inclusion in $W_a'$ as specified by~\eqref{pathcond1}, that detects the very same violations.
\end{proof}

The reduction argument underlying \Cref{wasatisfies}---that a third traversal of the same recursive call site is always redundant, since it can be collapsed onto an equivalent path traversing it only twice---is conceptually related to recent results on context-sensitive program analysis of service-oriented systems~\cite{Weissbach25}, where compositions of workflow nets are similarly analyzed by bounding the recursive unfolding of net compositions to a small, fixed depth without loss of analysis precision. While the two settings differ (workflow net composition versus attribute grammar composition), both rely on the same underlying principle: once a recursive structure has been unfolded enough times to expose its repeating pattern, further unfolding cannot reveal any additional behavior that a shorter unfolding could not already exhibit.

\begin{mainthm}\label{wasatisfies2}
    With the additional restriction of Constraint (b), specified by~\eqref{pathcond2}, $W_a$ still satisfies requirement~\eqref{warequirement}.
\end{mainthm}
Intuitively: if two distinct paths $w, w'$ both reach the same node $v$ having produced the same production context stack and the same edge tag stack (\textit{i.e.}, $f_{pCtx}(w, \langle\emptyset\rangle)=f_{pCtx}(w', \langle\emptyset\rangle)$ and $f_{stack}(w, \epsilon)=f_{stack}(w', \epsilon)$), then, from $v$ onward, any continuation of the visit behaves identically regardless of which of the two paths was taken to reach it. Constraint (b) exploits exactly this: it is therefore safe to keep only one of $w, w'$ in $W_a$ and discard the other, without missing any violation that the discarded path alone would have revealed.
\begin{proof}
    The goal is to prove that, given a path $w_e \in W_a'$ from the entry node $v_0$ to a node $v_e$ with a detected violation ($f_e(w_e)\neq\emptyset$), then a $w_e' \in W_a$ such that $f_e(w_e')\neq\emptyset$ also exists.

    The proof is trivial when considering that the requirement for inclusion of paths in $W_a$~\eqref{pathcond2} can be formulated as follows: given a set of paths that all share the same results of evaluating $f_{pCtx}$ and $f_{stack}$, \textit{i.e.}, a set $W_x$ such that
    \[
        \exists C \in \mathcal{C}, S \in L^*, \quad \forall w_x \in W_x \:\: f_{pCtx}(w_x, \langle\emptyset\rangle)=C \:\land\: f_{stack}(w_x, \epsilon)=S
    \]
    then exactly one of them is included in $W_a$.

    In fact, a direct consequence of such a definition is that all the paths $w_x \in W_x$ also share the same value of $f_e(w_x)$, by definition of $f_e$. Therefore, if $\exists w_e \in W_x$ such that $f_e(w_e)\neq\emptyset$, then regardless of which of the other items is picked from $W_x$ to fill the role of $w_e'$, it satisfies $f_e(w_e')=f_e(w_x)\neq\emptyset$.
\end{proof}

\subsubsection{Visit Algorithm}
We can finally define a theoretically sound algorithm that traverses the CFG according to the properties defined thus far to detect any of the relevant violations that may occur when composing separately compiled language fragments.

We employ the previously defined functions $f_{action}$ and $f_{edge}$, respectively, to evolve the state and to drive the visit through valid paths only. The implementation keeps a set of states to visit as storage for paths that need to be explored further, and one of all states that have been visited to enforce path condition~\eqref{pathcond2}. Condition~\eqref{pathcond1} is instead
enforced by inspecting the edge tag stack before following an edge. Recall that $f_{edge}$ is a partial function (\sect{sect:ContribCfgVisitValidPaths}): the condition $\exists S'$ in \Cref{AlgVisitRoleCFG} checks that $f_{edge}(S, (v,v'))$ is defined, \textit{i.e.}, that following the edge $(v,v')$ does not lead to an invalid path.
The worklist $\Phi$ is deliberately modeled as a set, and \Cref{AlgVisitRoleCFG} picks an arbitrary element of it at each iteration, rather than fixing a specific traversal order (\textit{e.g.}, breadth-first or depth-first): \Cref{wasatisfies,wasatisfies2} and the soundness of the overall visit hold regardless of which state is picked at each step, thus no particular order is required for correctness, and none is imposed by the formal model. The actual implementation resolves this choice concretely by keeping $\Phi$ as a list and always continuing the visit from the states most recently discovered, \textit{i.e.}, a depth-first traversal. This is simply the simplest instantiation of the formal model's ``pick any'' step, not a requirement.
\begin{algorithm}
\caption{Visit of a role's CFG to find violations}\label{AlgVisitRoleCFG}
\begin{algorithmic}
\State$\Phi := \{(v_{entry}(G_r), v_{entry}(G_r), \epsilon, \langle\emptyset\rangle)\}$
\Comment{Set of states to visit}
\State$\Gamma := \emptyset$\Comment{Set of visited states}
\State$U := \emptyset$\Comment{Found violations}
\While{$\Phi \neq \emptyset$}
\State{} pick any $(v, w, S, C) \in \Phi$
  \State$\Phi := \Phi \setminus \{(v_, w, S, C)\}$
  \If{$(v, C, S) \notin\Gamma$} \Comment{Path condition~\eqref{pathcond2}}
    \State$\Gamma := \Gamma \cup (v, C, S)$
    \State$(C', U') := f_{action}(C, a_{node}(v))$
    \State$U := U \cup U'$
    \ForAll{$(v, v') \in E_r$}
      \State$S' := f_{edge}(S, (v, v'))$
      \If{$\exists S'$ \textbf{and} $\nexists l_1,l_2,l_3 \in S: l_1=l_2=l_3$}
        \Comment{Path validity and~\eqref{pathcond1}}
        \State$\Phi := \Phi \cup \{(v', wv', S', C')\}$
      \EndIf%
    \EndFor%
  \EndIf%
\EndWhile%
\If{$U = \emptyset$}
\State{} verification succeeded
\Else%
\State{} violation(s) have been detected and collected in $U$ 
\EndIf%
\end{algorithmic}
\end{algorithm}

\subsection{Optimization}\label{ContribOptimization}

While sound, \Cref{AlgVisitRoleCFG} needs performance that allows static analysis to be performed as part of the language development process with ease.
Therefore, significant effort went toward reducing the time complexity of the visit algorithm, with the goal of making the execution of \nlgcheck{} adequate to be used with non-trivial languages.
We will not discuss such optimizations in detail, so please refer to \sect{sec:evaluation} for a performance evaluation. This section is limited to providing a brief overview of both optimizations and their rationale. Both optimizations rely, in different ways, on avoiding redundant work: \sect{ContribOpt1} avoids constructing duplicate CFG fragments in the first place, while \sect{ContribOpt2} caches, per nonterminal, which combinations of production-context-stack heads have already been explored, so that a semantic action already analyzed for a given head need not be revisited for the same head again. Please refer to \sect{sec:app:algoopt} for the optimized visit algorithm.

\subsubsection{Graph Deduplication}\label{ContribOpt1}

The execution time of \Cref{AlgVisitRoleCFG} largely depends on the size of the CFG\@.
A reasonable optimization is therefore to reduce the size of the graph whenever possible, which in turn depends on the number of semantic actions in the language.
The optimization leverages a few key empirical observations.
\begin{compactitem}
    \item Based on our own experience developing languages with \tool{Neverlang} (\Cref{MutTestingResults}), it is not uncommon for a language grammar to contain productions that differ from each other only in the terminal keyword used in their concrete syntax to disambiguate which alternative was parsed---regardless of whether that terminal itself survives as an AST node---and are therefore syntactically equivalent in this analysis context. Consider, for example, a grammar with two distinct productions for operators sharing the same precedence, such as \inlineneverlang{AddExpr: Expr <-- Expr "+" Term} and \inlineneverlang{SubExpr: Expr <-- Expr "-" Term}, which share the same nonterminals and are disambiguated in the concrete syntax by the \inlineneverlang{+} and \inlineneverlang{-} keywords, respectively. The same can be said for other binary or unary operators. 
    \item It is also common for semantic actions associated with such productions to have very similar implementations. With respect to the prior example, the implementations of \inlineneverlang{AddExpr} and \inlineneverlang{SubExpr} are extremely similar for all intents and purposes, differing only in the arithmetic operator used to implement the expression. Most notably, such \textit{similar} semantic actions might as well be \textit{identical} from \nlgcheck{}'s standpoint: since the CFG only tracks attribute reads, writes, and control flow, the arithmetic operator itself is irrelevant to the analysis, regardless of whether the AST retains the disambiguating terminal as a node.
    \item Modularity is one of the strongest points of language workbenches, and thus language developers are encouraged to keep semantic actions as small as possible, with the goal of improving their reusability. By adhering to this guideline, semantic actions can be as small as a single or a few lines of code, without any control flow. In this case, the CFGs simply consist of a chain of nodes, each of which has exactly one entry node and one exit node.
\end{compactitem}
In summary, the actual output of converting semantic actions into CFG fragments may contain several duplicates.
Therefore, the CFG building algorithm can be optimized by discarding any CFG of these duplicates.
Detecting all such duplicates would require deciding graph isomorphism, a problem for which no polynomial-time algorithm is known in general. We sidestep it with a conservative check that exploits the observations above. Two productions are considered duplicates only if they have the same left- and right-hand-side nonterminals and, in addition, either their semantic actions are the same compiled action, or the CFGs of their semantic actions are chains of nodes, as in the third observation, mapped to the same sequence of actions. The latter comparison takes time linear in the length of the chains. The check may miss some duplicates, such as two equivalent fragments with internal branching compiled from different semantic actions, but it never merges two fragments that differ for the analysis.

\subsubsection{Path Deduplication}\label{ContribOpt2}

Consider a generic nonterminal symbol $A$ of the grammar. Every time the semantics of a production having $A$ in its right-hand side---\textit{i.e.}, with the form $B \leftarrow \ldots A \ldots$---is analyzed, the semantics of all the productions with $A$ on the left-hand side must also be analyzed.

The language potentially contains several such productions.
We empirically observed that it is frequent for the $head$ of the production context stack to be the same when entering one of the semantic actions among the set of those associated with productions generating the same nonterminal.
There are two possible reasons for this fact.
\begin{compactenum}
    \item We found it to be more common to pass attributes from child to parent in the AST than from parent to child; therefore, when descending into the AST, when encountering $EnterCtx$ nodes in a CFG path, the $head$ of the context stack often contains no attributes.
    \item In the cases in which attributes are passed from parent to child, we can expect the semantic actions associated with the child nodes to actually read these attributes; therefore, for all the paths without a violation\footnote{As is commonly done in error theory, we expect the absence of errors to be more likely than their presence.} the $head$ of the context when entering a semantic action must contain all the attributes that the semantic action will read.
\end{compactenum}
The idea for this optimization leverages these observations to review the structure of the CFG\@: for each nonterminal $A$, we add an \textit{entry node} and an \textit{exit node} that are shared among all the semantic actions attached to productions of the form $A \leftarrow \ldots$.
Such nodes serve as \textit{anchor points} to determine when to access and when to update a cache storing all the production context stacks.
By caching all the possible heads of the production context stack at the corresponding nonterminal exit node for each possible head of the production context stack at each nonterminal entry node, the algorithm can \textit{jump} from the nonterminal entry node to the exit node. The goal is to avoid repeating the whole visit when a particular node is reached more than once with the same context stack.

\subsection{Acting on Reported Violations}\label{ContribActing}
Our analysis is meant to be run by two kinds of users. The authors of a language component have its source code and can run \nlgcheck{} on test compositions before distributing the component, as with the stand-alone test languages that accompany the sublanguages of \textit{numericUnitExpr} (\sect{EvaluationMutationResults}). For them, a reported violation is as actionable as a failing well-definedness check in systems such as Silver (\sect{sec:related-work}). The integrator who composes precompiled components into a language, instead, only has their binaries, and cannot fix a reported violation by editing the components themselves.

For such an integrator, a violation is not a bare yes/no answer. For each violation, \nlgcheck{} reports the source location of the offending attribute access, the missing or ill-typed attribute, the attributes that are actually available at that point together with their types, and a witness: the chain of productions whose evaluation leads to the access. This is usually enough to tell a naming mismatch between two components, such as one component writing \texttt{value} where another reads \texttt{val}, from a contribution that no available component provides.

Both cases can be addressed at composition time, without source code, since \tool{Neverlang} composes languages from separately compiled modules through slices (\sect{sec:neverlang-in-a-nutshell}), which are themselves small, separately compiled units. A naming mismatch can be resolved with an \textit{attribute mapping}: a slice that imports the role of a module can rename, for that module only, the attributes accessed by its semantic actions. For instance, the slice statement
\begin{center}
\texttt{module M with role evaluation mapping attributes \{ \$1.value => val \}}
\end{center}
makes every access to the attribute \texttt{value} of the first right-hand-side nonterminal, within the semantic actions of \texttt{M}, refer to \texttt{val} instead. \nlgcheck{} applies the same renaming when building the CFGs of these semantic actions, so the composition obtained with the new slice can be checked again. A missing contribution can be addressed by replacing, in a new slice, the module that provides the offending role for that syntax with a module written by the integrator, which is feasible since semantic actions tend to be small (\sect{ContribOpt1}). Otherwise, the integrator can select a different version of a component, or report the problem to its authors.

The situation is comparable to a linker reporting an unresolved symbol while combining binary libraries. The report does not fix the problem, but it moves its discovery from the runtime of the compiler, possibly at the site of the end user, to the moment of composition, and it points to where the fix is needed.

\newcommand{\xmark}{\ding{55}}

\section{Evaluation}\label{sec:evaluation}
We performed a quantitative evaluation to test the applicability of \nlgcheck{} beyond its theoretical soundness. This evaluation is guided by the following research questions:
\begin{compactitem}
    \item[\textbf{RQ1}] How effective is \nlgcheck{} at detecting undefined-attribute errors introduced by separately compiling and composing language components?
    \item[\textbf{RQ2}] What is the performance overhead of \nlgcheck{}, and how do the optimizations of \sect{ContribOptimization} affect it?
\end{compactitem}
To answer RQ1, we opted for an evaluation of the results using source code mutations~\cite{DeMillo78, Jia11} on pre-existing \tool{Neverlang} languages, generated by a companion tool, \textit{nlgmutator}, which is not part of \nlgcheck{} itself: it instruments the \tool{Neverlang} compiler to systematically alter attribute names within semantic actions, producing mutated \tool{Neverlang} projects that are then independently checked with \nlgcheck{}.
Traditionally, the goal of mutation testing is to assess the quality of test suites. It therefore fits our requirement of systematically introducing structural errors within a language to verify whether \nlgcheck{} is capable of detecting them. Moreover, mutation testing was part of a prior study with \tool{Neverlang}~\cite{Cazzola22d} and was therefore deemed an adequate choice, especially because the reflection \textit{application programming interface} (API) of \tool{Neverlang} allows developers to easily mimic mutation operators without encompassing source code recompilation: in mutation testing, a mutation operator is a function that takes the original program and produces the mutated one.

Concerning the errors that can happen when composing language fragments, a mutation operator that is extremely easy to implement and that can mimic most real-world programming errors is the random change of one or more characters within the identifier of an attribute. This is a simple yet relevant example because, due to separate compilation, \tool{Neverlang} cannot detect such an error statically---not knowing in advance if a typo in an attribute name was indeed an error or simply an intended design choice foreseeing more language fragments that may be implemented and compiled in the future. We are aware that this operator only introduces naming errors, and not, \textit{e.g.}, errors in the type hierarchy such as the assignment of an incompatible-but-similarly-named type.\ \sect{sect:evaluation:threats} discusses this limitation further. The software and data used to perform this experiment are available at \url{https://zenodo.org/records/21795001}.

\mypar{Experimental setup}
To build the CFG required by \sect{ssect:complete}, \nlgcheck{} uses \tool{Soot} to extract a preliminary CFG (in \tool{Soot}'s own \textit{Jimple} intermediate representation, \sect{sec:soot-library}) from the compiled bytecode of each semantic action. This preliminary CFG is precise---it faithfully represents the compiled method's control flow---but is expressed in terms of \tool{Jimple} instructions rather than the actions of \sect{ssect:complete}, so a translation step then maps it onto a CFG containing only those actions.
We implemented \nlgcheck{} in \tool{Scala}, translating the formal model of \sect{sect:contribution} directly: each action of \sect{ssect:complete} (\textit{e.g.}, $WriteAttr$) is an immutable case class implementing a common \texttt{Action} interface, whose \texttt{visit} method is a pure function computing $f_{action}$ and returning an updated, equally immutable production context, mirroring the $put$/$get$ functions of \sect{sect:contrib:visit}. The visit algorithm of \Cref{AlgVisitRoleCFG} itself is a tail-recursive function over an explicit worklist, threading the visited-states set and the collected violations as immutable function arguments rather than as mutable fields.
Mutable state is used only where the surrounding framework leaves no alternative: \textit{nlgmutator}, for instance, is driven by the \tool{Neverlang} compiler's own translator-plugin architecture, which repeatedly invokes a callback to translate each production's semantic action. Since this callback is invoked by code outside our control, tracking which production is currently being mutated and how far into it the last mutation reached requires state that persists across those calls, and is therefore kept in a small amount of module-level mutable state, isolated from the otherwise immutable analysis core.
\Cref{AlgMutationExperiment} summarizes the experiment performed on each evaluated language.
The unmutated language is first compiled and checked with \nlgcheck{} as a control. If either step reports a problem, the language is excluded, since any subsequent finding would then be unrelated to the mutations being tested.
Mutations are then introduced one at a time by instrumenting the \tool{Neverlang} compiler with the mutation operator. They are therefore introduced while generating \tool{Java} code from the unmodified \texttt{.nl} sources, not by editing the \texttt{.nl} specifications themselves, and each mutant contains exactly one mutation.
Attribute access expressions are selected systematically, in the textual order in which they appear within their production, until each of them has been mutated exactly once. This exhausts the space of single-attribute-name mutations for the language, so the number of generated mutants equals the total number of attribute access expressions across all the semantic actions of the language.
The randomness in each mutation is confined to \textit{how} the selected identifier is corrupted (a uniformly random character position is replaced with a uniformly random different letter), not to \textit{which} attribute access is selected next.
For each mutant, we record whether it already fails to compile to \tool{Java} bytecode, \textit{i.e.}, whether the baseline toolchain catches it without \nlgcheck{}, and whether \nlgcheck{} reports a violation.

\begin{algorithm}[t]
\caption{Mutation-testing experiment performed on each evaluated language $\mathcal{L}$}\label{AlgMutationExperiment}
\begin{algorithmic}
\State{} compile the \texttt{.nl} sources of $\mathcal{L}$ into \tool{Java} code, and the \tool{Java} code into bytecode
\If{a compilation step fails \textbf{or} \nlgcheck{} reports a violation on the compiled classes}
  \State{} exclude $\mathcal{L}$ from the experiment\Comment{Control run}
\Else%
  \State{} instrument the \tool{Neverlang} compiler with the mutation operator
  \ForAll{attribute access expressions $x$ of $\mathcal{L}$, in textual order within each production}
    \State{} $x' :=$ $x$, with a random character of its identifier replaced by a different random letter
    \State{} regenerate the \tool{Java} code from the unmodified \texttt{.nl} sources, emitting $x'$ in place of $x$
    \State{} compile the mutated \tool{Java} code and record whether compilation fails
    \State{} run \nlgcheck{} on the mutated classes and record whether it reports a violation
  \EndFor%
\EndIf%
\end{algorithmic}
\end{algorithm}

\mypar{Data Setup}
At the time of writing, a number of \tool{Neverlang}-based languages exist that can be used to evaluate \nlgcheck{}.
Among these, we chose a subset comprised of simple, reusable sublanguages (\textit{expr} and \textit{numericUnitExpr}), domain-specific languages (\textit{stateMachines} and \textit{logLang}), and a general-purpose programming language (\textit{neverlangJS}). Note that \textit{neverlangJS} is an interpreter, and follows the convention, typical of \tool{JavaScript} interpreters, that every \texttt{Expression} and \texttt{Statement} yields a result of a common \texttt{Value} type and shares a common runtime environment. This convention is representative of interpreters in general, rather than being specific to \textit{neverlangJS}.

Each of the chosen types of languages has a rationale.
\begin{compactitem}
    \item Sublanguages are arguably the most fit to be analyzed with \nlgcheck{}, as they are intended to be reused across several languages and are therefore pre-compiled separately from the host language.
    \item Domain-specific languages are the main application scenario of language workbenches in general. They can also be embedded in other languages, while using syntactic constructs usually different from those of general-purpose languages, which may lead to errors.
    \item A full-fledged general-purpose language represents a large-scale project, often developed by several developers concurrently, each possibly adopting different choices with respect to the attribute grammar. As a result, keeping track of the entire project without relaxing the separate compilation constraint may prove challenging.
\end{compactitem}

\Cref{DatasetChars} characterizes the scale of each language in terms of the number of \tool{Neverlang} modules and lines of \tool{Neverlang} code it is composed of. \Cref{PerfResults} (\sect{ChapEvaluationPerf}) additionally reports, for each language, the number of grammar productions and CFG nodes involved in the analysis.

\begin{table}
    \centering
    \begin{tabular}{lcc}
        \toprule
        \textbf{Language} & \textbf{Modules} & \textbf{LOC} \\\midrule
        \textit{expr}            & 9   & 73   \\
        \textit{numericUnitExpr} & 46  & 474  \\
        \textit{stateMachines}   & 12  & 181  \\
        \textit{logLang}         & 25  & 279  \\
        \textit{neverlangJS}     & 130 & 4616 \\
        \bottomrule
    \end{tabular}
    \vspace{10pt}
    \caption{Characterization of the dataset used for the mutation-testing experiment: number of \tool{Neverlang} modules (\texttt{.nl} files) and lines of \tool{Neverlang} code per language.}\label{DatasetChars}
\end{table}

\mypar{Hardware setup}
All tests were conducted on an Intel Core i7-4720HQ machine with a maximum clock speed of 3.6GHz and 16GB of RAM\@. The software consists of the GNU/Linux 6.1.12 (x86\_64) operating system and OpenJDK 17.0.6. Please note that a different setup may yield different results.

\subsection{Error Detection}\label{EvaluationMutationResults}

\begin{table}
    \centering
        \begin{tabular}{lccc}
            \toprule
            \textbf{Language}        & \textbf{Generated Mutants} & \textbf{Detected w/o \texttt{nlgcheck}} & \textbf{Detected with \texttt{nlgcheck}} \\ \midrule
            \textit{expr}            & 14                         & 0 (0\%)                                 & 14 (100\%)                               \\
            \textit{numericUnitExpr} & 77                         & 14 (18\%)                               & 55 (71\%)                                \\
            \textit{stateMachines}   & 19                         & 0 (0\%)                                 & 5 (26\%)                                 \\
            \textit{logLang}         & 25                         & 0 (0\%)                                 & 22 (88\%)                                \\
            \textit{neverlangJS}     & 596                        & 0 (0\%)                                 & 441 (74\%)                               \\ \midrule
            \textbf{Total}           & 731                        & 14 (2\%)                                & 537 (73\%)                               \\
            \bottomrule
        \end{tabular}

        \vspace{10pt}

        \caption{Results of running \textit{nlgmutator} on various preexisting \tool{Neverlang}-based projects, showing how many mutants are generated, how many are detected in the regular compilation process, and how many are detected by \nlgcheck{}. Percentages are rounded to the nearest integer.}\label{MutTestingResults}
\end{table}

\subsubsection*{Results}
\Cref{MutTestingResults} summarizes the results of the experiment.
The total number and the percentage of mutants that cause the compilation process to fail are then compared to those of mutants detected by \nlgcheck{}.

In the vast majority of cases, mutants survive the compilation process fully undetected, with one notable exception, \textit{i.e.}, when an attribute appears both within the semantic action body and the type guard associated with that action.
In such cases, the \tool{Neverlang} compiler injects type casts within the generated \tool{Java} code, and therefore, if the attribute is mutated while the guard is not, the cast is no longer injected correctly, possibly leading to \tool{Java} compilation errors.\footnote{Recall that in the \tool{Java} code generated by \tool{Neverlang}, attributes are simply assumed to be \texttt{Object}s unless stated otherwise.}
As reported in \Cref{MutTestingResults}, such errors amount to only 18\% of errors in the \textit{numericUnitExpr} code base and only 2\% overall (14 out of 731 across all experiments).
In contrast, the total number of mutants detected by \nlgcheck{} is 537 out of 731, thus about 73\%.

\Cref{MutTestingResults} does not repeat the execution time of \nlgcheck{} on each language, which is reported in \Cref{PerfResults} (\sect{ChapEvaluationPerf}) together with that of the unoptimized and partially optimized variants of the tool. It also omits a compile-only baseline without \nlgcheck{}, since that overhead is orthogonal to the questions this experiment answers. Finally, unlike \Cref{PerfResults}, it does not include \textit{neverlangExpressions}. The control run of \Cref{AlgMutationExperiment} already reports violations on the unmutated \textit{neverlangExpressions}, so the protocol excludes it: any violation reported on its mutants could not be attributed to the mutation. Moreover, a single run of \nlgcheck{} on this language takes minutes (\Cref{PerfResults}), which would make a campaign with one run per mutant considerably more expensive than for the other languages. Extending the campaign to a corrected version of \textit{neverlangExpressions} is left as future work.

\subsubsection*{Discussion}
Because the \tool{Neverlang} compiler does not perform any verification of attribute name matching across multiple modules, it is apparent that it alone does not suffice to effectively capture errors in the definition of grammar attributes. Our tool represents a significant improvement when compared to this baseline. However, one may wonder about the remaining 27\% of mutants that are still undetected even when using \nlgcheck{}, especially given the theoretical soundness of the approach. There are a few considerations to be made in this regard.

\mypar{Read without write}
First, while the mutation operator has the goal of simulating programming errors that would result in a crash at runtime, not all mutations lead to errors. Consider an attribute that is written but never read: this may indicate either a design flaw or the implementation of a sublanguage that exposes attributes to be used by modules from other sublanguages.

\mypar{Unnecessary copy}
Another source of undetected mutants occurs when semantic actions associated with a production of the form \texttt{A $\leftarrow$ B} copy an attribute from their only AST child node: as previously mentioned, attribute access to a node having exactly one child triggers a quality-of-life feature so that the AST chains are navigated from parent to child when an attribute is not found within the parent. In this particular case, copying the attribute with an incorrect name may not result in a runtime error. After performing the experiment, manual inspection of \tool{Neverlang} code revealed such cases of unnecessary copy operations to be not uncommon, especially in the \textit{neverlangJS} code base.

\mypar{Dead code}
Mutants may remain undetected because the mutated code can never be reached within the CFG\@; \nlgcheck{} does not currently detect dead code as a violation in its own right, so this is presently only a source of missed detections rather than a planned feature.
We found such an example in the \textit{numericUnitExpr} code base, which is articulated in several sublanguages, each implementing one concern of the language, such as integer arithmetic, floating-point arithmetic, unit checking, and so on. Each of these libraries also includes portions of so-called \textit{glue code} needed to run them in stand-alone mode---mainly for testing purposes, such as a language only implementing binary operations between integer numbers. Such pieces of code consist of a definition of the \texttt{Program} nonterminal axiom, associated with a semantic action.
Since mutation operators mutate all modules, they eventually also mutate semantic actions that are not actually used in the final language.
Pieces of dead code were discovered in \textit{neverlangJS} as well.

\mypar{Non-idiomatic attribute access}
Under the hood, \tool{Neverlang} can simply be seen as a \tool{Java} API for the development of programming languages.
Therefore, developers can either approach \tool{Neverlang} development in an \textit{idiomatic} way, by writing semantic actions in \tool{Neverlang} source files, or in a \textit{non-idiomatic} way, by invoking \tool{Java} methods and interacting with the AST through the \tool{Neverlang} runtime API, instead of using the \tool{Neverlang} syntax. Currently, there is no support in \nlgcheck{} to translate arbitrary \tool{Java} code into a CFG compatible with the requirements presented in \sect{sect:contribution}. Therefore, any \textit{non-idiomatic} attribute access remains undetected even if the corresponding attribute was mutated upon it being written in a different semantic action. We did not perform a systematic mining of public \tool{Neverlang} repositories to estimate how common non-idiomatic access is in practice; this is left as future work.
A mechanism akin to \tool{JastAdd}'s explicit attribute declarations---requiring non-idiomatic \tool{Java} code to access an attribute only through a declared, module-level interface, rather than directly through the runtime API---could in principle close this gap. Such declarations would give \nlgcheck{} the read/write information it currently lacks for non-idiomatic accesses, without requiring it to parse arbitrary \tool{Java} code.\ \tool{Neverlang} does not presently offer such a mechanism, so this remains a promising direction for future work rather than a solution available today.

\mypar{Relation to dependency taxonomy}
Revisiting the RAW/WAR/WAW hazard taxonomy of \sect{sec:data-dependencies}: the \textit{non-idiomatic attribute access} category corresponds to a genuine, undetectable RAW hazard (a read using a value produced by a write, where the write is outside \nlgcheck{}'s visibility); \textit{read without write} and \textit{dead code} correspond to writes with no corresponding read on any reachable path, i.e., the absence of a hazard rather than an undetected one; \textit{unnecessary copy} corresponds to a (potential) WAW/output dependency between the copy and a same-named write elsewhere, which happens not to manifest as an observable error in this specific case.

\mypar{Why \textit{stateMachines} detects fewer mutants}
\textit{stateMachines} has the lowest detection ratio among the evaluated languages (26\%, \Cref{MutTestingResults}). We do not have a per-mutant breakdown to attribute an exact share of this to each of the four categories above, but re-inspecting \textit{stateMachines}' grammar surfaces a concrete, language-specific contributor to the \textit{read without write} category: the \texttt{StateBlock} nonterminal has an empty alternative (\inlineneverlang{StateBlock <-- ""}, selected whenever a state declares no outgoing events), and the enclosing \texttt{StateDeclaration} action always writes a \texttt{state} attribute onto its \texttt{StateBlock} child regardless of which alternative is parsed, yet only the non-empty alternative's semantic action ever reads it back. Since states with no outgoing events are common in typical state-machine specifications, this single write-without-read pattern alone is likely to account for a disproportionate share of \textit{stateMachines}' undetected mutants compared to the other, less alternative-heavy languages in \Cref{MutTestingResults}.

\mypar{Summary}
We summarized the main reasons for which mutants may be undetected by \nlgcheck{}.
For the \textit{read without write,} \textit{unnecessary copy,} and \textit{dead code} mutants, reporting them as errors should be considered a false positive in the analysis performed by \nlgcheck{}.
Neither case strictly constitutes an error, as the compiler/interpreter would not crash, but we may consider extending \nlgcheck{} to output a warning in these occurrences.
Instead, \textit{non-idiomatic attribute accesses} are actual errors that are beyond the analysis capabilities of \nlgcheck{}. This is the only true limitation of our approach, which should therefore be considered fully sound only when writing \textit{idiomatic} \tool{Neverlang} code, at least until the tool is extended to capture all kinds of attribute accesses.

In mutation-testing terms, the \textit{read without write}, \textit{unnecessary copy}, and \textit{dead code} mutants are \textit{equivalent mutants}~\cite{Jia11}: they do not change the behavior of the language on any program and therefore do not represent real errors. We did not classify each of the 194 undetected mutants into these categories, so we cannot report the exact number of equivalent mutants. Since the mutation score excludes equivalent mutants from the total number of mutants~\cite{Jia11}, the ratio of detected mutants reported in \Cref{MutTestingResults} (537 out of 731, 73\%) is a lower bound on the mutation score of \nlgcheck{}.

\subsection{Performance Evaluation}\label{ChapEvaluationPerf}

\begin{table}
    \resizebox{\textwidth}{!}{
    \begin{tabular}{lccccccc}
        \toprule
                                      & \multicolumn{2}{c}{\textbf{Unoptimized}} & \multicolumn{3}{c}{\textbf{Opt. 1 only}} & \multicolumn{1}{c}{\textbf{Opts. 1\,\&\,2}}                 \\
        \textbf{Language}             & \textbf{Prods (Nodes)}                   & \textbf{Time}                            & \textbf{Prods (Nodes)}              & \textbf{Reduction} &  \textbf{Time} & \textbf{Time} \\ \midrule
        \textit{expr}                 & 10 (61)                                  & \xmark{}                                 & 8 (45)                              & 20\% (26\%)        &  3s            & 3s            \\
        \textit{numericUnitExpr}      & 14 (159)                                 & \xmark{}                                 & 13 (157)                            & 7\% (1\%)          &  \xmark{}      & 7s            \\
        \textit{stateMachines}        & 14 (65)                                  & 3s                                       & 12 (59)                             & 14\% (9\%)         &  3s            & 3s            \\
        \textit{logLang}              & 17 (112)                                 & 4s                                       & 17 (112)                            & 0\% (0\%)          &  4s            & 4s            \\
        \textit{neverlangJS}          & 230 (1728)                               & \xmark{}                                 & 197 (1584)                          & 14\% (8\%)         &  \xmark{}      & 20s           \\
        \textit{neverlangExpressions} & 80 (3537)                                & \xmark{}                                 & 72 (3489)                           & 10\% (1\%)         &  \xmark{}      & 3m 25s        \\ \bottomrule
    \end{tabular}}

    \vspace{10pt}

    \caption{Results of running \nlgcheck{} on various pre-existing \tool{Neverlang}-based projects, with and without optimizations, respectively. The \textbf{Reduction} column reports the \% decrease in productions (nodes) from Optimization 1 alone, relative to the unoptimized CFG\@. An \xmark~indicates that the process either crashed with an out-of-memory error or was manually aborted after 10 minutes. Optimization 1 (\Cref{ContribOpt1}) reduces the number of productions and CFG nodes; Optimization 2 (\Cref{ContribOpt2}) further improves execution time without further modifications to the graph.}\label{PerfResults}
\end{table}

Along with an evaluation of the ability of \nlgcheck{} to detect errors, we are also interested in its performance. Static program analysis is notably intensive, and if the tool does not perform fast enough, it would be infeasible to use as part of the day-to-day development process. In fact, during the development of \textit{nlgcheck}, it quickly became apparent that the base algorithm was only usable with very small languages.
This led to the optimization presented in \sect{ContribOptimization}.
In particular, in this section, we use \textit{Optimization 1} to refer to graph deduplication (\Cref{ContribOpt1}) and \textit{Optimization 2} to refer to path deduplication (\Cref{ContribOpt2}). This performance evaluation only compares \nlgcheck{} against its own unoptimized and partially optimized variants. It does not compare against other static analysis tools for attribute grammars (\textit{e.g.}, Silver, JastAdd, or Eli), which target a different analysis (flow-graph-based well-definedness rather than CFG-based path-sensitive analysis, \sect{sec:related-work}) and are not directly comparable on this benchmark. The same holds on the functional side: we do not provide an empirical comparison of error-detection capability against these systems either, since doing so would require porting our mutation-testing benchmark onto AG systems with entirely different architectures and input languages. This remains future work.

To evaluate the final optimized implementation and the efficacy of the optimizations themselves, \textit{nlgcheck} was tested with and without optimization on each language. The two optimizations can also be disabled independently of each other. Evaluating Optimization 2 alone is left as future work, together with repeated measurements (\sect{sect:evaluation:threats}).

\mypar{Results}
\Cref{PerfResults} summarizes the results of the performance evaluation. For each language, the second column reports the performance with both optimizations disabled.
The third column yields the results with only Optimization 1 enabled, showing how the number of productions and of CFG nodes is reduced. Finally, in the last column, the tool is run with both optimizations enabled; since Optimization 2 only affects the visit algorithm and not the construction of the CFG itself, its structure and size remain the same as in the previous step.

We set a time limit of 10 minutes for each task, a time frame that we deemed reasonable to perform a full analysis of the language
(for instance, before a daily commit).
In most instances, the unoptimized or partially optimized attempts either kept running indefinitely (thus surpassing the time limit) or failed with an out-of-memory error.

In summary, \Cref{PerfResults} shows that while Optimization 1 slightly reduces the size of the CFG, in turn enabling the analysis on the previously failing \textit{expr} case, it is not enough to make the tool suitable for use on real-world languages.
The problem is eventually solved by adding Optimization 2, which makes the analysis run in a reasonable time for all the considered cases.

\mypar{Expected time complexity}
Optimization 1 (graph deduplication, \Cref{ContribOpt1}) replaces each duplicate CFG fragment with a single shared copy. It does not solve graph isomorphism in general. Duplicates are detected by hashing, for each production, its nonterminals together with either the sequence of actions of its chain-shaped CFG or the identity of its semantic action (\Cref{ContribOpt1}). Its expected cost is therefore linear in the total size of the productions and of the CFGs involved. Its benefit is proportional to how much syntactic duplication the grammar contains, \textit{e.g.}, sibling productions differing only in an operator keyword. This explains why \textit{logLang} sees no reduction at all (\Cref{PerfResults}): it happens to contain no such duplicated fragments.
Optimization 2 (path deduplication, \Cref{ContribOpt2}) instead caches, for each nonterminal, the results already computed for a given head of the production context stack (\sect{ContribCfgVisitActionProps}). What would otherwise be a re-exploration of the same paths for every occurrence of a nonterminal thus becomes a single computation reused across all of them. Its benefit therefore scales with how often the same nonterminal recurs across productions. This is why it is responsible for the majority of the speedup on the larger, more repetitive \textit{neverlangJS} and \textit{neverlangExpressions} grammars.

\subsection{Threats to Validity}\label{sect:evaluation:threats}
We discuss the threats to the validity of this evaluation following the classification of \citet{Wohlin03}.

\mypar{Construct validity~\cite{Wohlin03}}
It is debatable whether the experiments performed within this evaluation accurately measure \nlgcheck{}'s capability of detecting real errors (first experiment) and its performance (second experiment), respectively.
For instance, we used only one kind of mutation operator---\textit{i.e.}, changing an attribute name---whereas actual developer errors may vary. To counterbalance this limitation, we performed an extensive mutation testing experiment, thus mutating each attribute in the grammar at least once. We do not claim that mistyping an attribute name is the only mistake a \tool{Neverlang} developer can make. Rather, among the many errors that can occur while developing a language, \nlgcheck{} is concerned with only a very limited subset, and attribute-naming errors are the ones relevant to that subset. For instance, type errors are mostly irrelevant because they can be captured by the \tool{Java} compiler without needing additional tooling. Similarly, errors such as forgetting an attribute or swapping two attributes can eventually be remapped to a combination of attribute renames, as with separate compilation, there is no way to determine if two similar identifiers were meant to represent two different attributes or a typing error made by the developer while writing the same attribute twice.

We did not manually annotate each generated mutant with the specific error it was expected to trigger, which would have allowed us to compute \nlgcheck{}'s precision and recall directly. We leave this as future work. Similarly, we did not separately quantify how many of the undetected mutants specifically fall under the \textit{unnecessary copy} category discussed in \sect{EvaluationMutationResults}; both would require re-instrumenting the mutation-testing campaign to retain a per-mutant classification, which we leave as future work alongside the precision/recall computation above.
We also did not write test programs to confirm, by execution, that each detected mutant causes a runtime error. Each mutation replaces a character of an attribute name with a different random letter, which, with high probability, yields a name that no semantic action uses. A mutated read therefore has no matching write anywhere, and a mutated write may leave the reads of the original name without a matching write. Since a violation reported by \nlgcheck{} lies on a valid path of the composed CFG, which corresponds to the traversal of an AST of the language, any execution following that path fails when it reaches the mutated access. The only exception is a path that requires a combination of branch outcomes that cannot occur at runtime, since \nlgcheck{} does not track attribute values. Confirming each detected mutant with a concrete program is left as future work.
The first experiment was likewise run once per language rather than repeated with different random seeds: since \textit{nlgmutator}'s mutation operator exhaustively covers every attribute access expression in the grammar at least once by construction, a different seed would only change which character position and replacement letter are used to corrupt each already-determined site (\sect{sec:evaluation}), not which sites are exercised, so repetition would not meaningfully change detection coverage.
Additionally, our experiments only introduce naming errors, not errors in the type hierarchy (\textit{e.g.}, assigning an incompatible-but-similarly-named type), nor do they account for the fact that interpreters typically manage their runtime environment (the current state of execution) as an attribute. Our approach implicitly assumes that such environments are compatible across composed components, not just in terms of type but also of runtime behavior, a stronger assumption than the one we verify.

For the second experiment, each configuration (language, optimization level) was likewise run once, without a warm-up phase or repeated runs to report a median. We acknowledge that a single run per configuration cannot rule out incidental system-level noise (\textit{e.g.}, just-in-time (JIT) compilation warm-up, garbage collection pauses) that repeated runs and a reported median would have controlled for.
The problem of construct validity is less objective, as we arbitrarily determined a threshold of 10 minutes that some may consider either too large or too strict depending on their personal preference and development schedule. Eventually, we deemed this choice acceptable for two reasons:
\begin{compactitem}
    \item Our tool is not intended to be used on each compilation, but rather once or a few times a day, to ensure the correctness of the attribute grammar before performing a commit.
    \item The time limit was only relevant when the optimizations were not in place. All experiments made with both optimizations completed well within it. This is true even for a rather complex language, such as \textit{neverlangJS}. The only outlier is \textit{neverlangExpressions}, which was originally developed in a deliberately complex manner---\textit{e.g.}, by using several abstractions and attribute mappings---to stress the modularization capabilities offered by \tool{Neverlang}. We believe that for most real development scenarios, such as domain-specific languages, the 10-minute target should never be reached, and possibly \nlgcheck{} could even be included as a standard compilation phase---at least for languages in the 3 to 4-second range.
\end{compactitem}

\mypar{Internal validity~\cite{Wohlin03}}
Some of the experiments may be affected by our specific expertise with \tool{Neverlang} and with languages developed by its means. To avoid this threat as much as possible, we included several different languages in the experimentation, making sure to include domain-specific languages, general-purpose languages, and sublanguages, each preceding the existence of \nlgcheck{}, which should mitigate the risk of having designed languages that specifically fit \nlgcheck{} and its requirements. All five languages were nonetheless developed within our own research group, with no external contributors. For instance, \textit{stateMachines} was originally authored as teaching material for a hands-on session at the BISS 2023 summer school, independently of and prior to this evaluation, but by a member of our group rather than an external developer. We also included results both with and without optimization to avoid any concerns about bias in the results.

\mypar{External validity~\cite{Wohlin03}}
Even though we tried to design the static analysis approach and its proof in a general manner that should be suitable for any language workbench, it is undeniable that the idea for this work originated from a \tool{Neverlang}-specific need, and the solution was then designed according to these needs. For instance, while we believe we adequately justified the advantages of separate compilation in \sect{sec:intro}, it must be said that separate compilation is far from common in language workbenches: \tool{Neverlang} is, to the best of our knowledge, the only one that fully supports the compilation of each individual artifact separately. This should not be intended as a limitation of other language workbenches, but rather a design choice. Most language workbenches reasonably opt for relaxing the requirement of separate compilation in favor of the guarantees provided by joint compilation, such as checking attribute types and dependencies. We can only hope that this contribution encourages other language workbenches trying to support separate compilation to achieve the best of both worlds.

\section{Related Work}\label{sec:related-work}
The challenge of ensuring static correctness in AGs has been a cornerstone of language engineering since Knuth’s seminal work~\cite{Knuth68}. However, as language workbenches have evolved toward modularity and separate compilation, a significant gap has emerged between structural completeness and operational safety. This section situates \nlgcheck{} within the landscape of modular AG verification, advanced AG formalisms, and static analysis techniques.

\mypar{Structural Verification vs. Operational Path-Sensitivity}
The primary focus of traditional AG verification is \textit{well-definedness}—ensuring that for any valid syntax tree, every attribute instance has a unique definition and no circular dependencies exist~\cite{Deransart88}.

Modern systems like Silver~\cite{VanWyk10, Kaminski12} have pioneered \textit{modular well-definedness analysis}. This allows independent language extensions to be verified in isolation, relying only on information from their dependencies~\cite{Kaminski12}. JastAdd~\cite{Hedin03, Ekman07b}, instead, weaves all the aspects of a specification together before generating the evaluator. Its optional check for missing inherited equations therefore applies to the complete specification, and, as observed by \citet{Kaminski12}, missing equations are otherwise reported when attributes are evaluated. It is important to note, however, that in both cases the modularity operates at the specification level: all contributing modules must still be present, in source form, at the single joint compilation step that generates the resulting language processor---as with the other systems discussed throughout this section (\sect{sec:evaluation} elaborates on why \tool{Neverlang}'s genuine separate compilation of already-compiled binary components is, to our knowledge, uncommon among language workbenches). While these tools are highly effective at ensuring structural integrity (i.e., that a definition exists in the codebase) within such a joint compilation, they treat the internal logic of semantic actions as \textit{black boxes.}

In a framework like \tool{Neverlang}, where semantic actions are implemented via imperative \tool{Java} code and attribute access is mediated by dynamic maps, the existence of a definition does not guarantee runtime safety. Our tool addresses this by shifting the focus from ``is the attribute defined in the module?'' to ``is the attribute initialized before use along every possible execution path?'' In an AG specified with equations, the first question suffices: the evaluator schedules each equation so that the attributes it uses are computed first, either statically, as with the visit sequences of ordered AGs~\cite{Kastens80}, or on demand, as in JastAdd~\cite{Hedin03}. In \tool{Neverlang}, instead, the order in which attributes are written and read is fixed by the imperative code of the semantic actions and by the \texttt{eval} calls written by the developer (\sect{sec:neverlang-in-a-nutshell}). A read may therefore precede the write that would define the attribute, even if such a write exists somewhere in the composition. The path-sensitive data-flow analysis of \sect{sect:contribution} is what prevents the resulting $MissingAttr$ runtime error in this setting.

\mypar{Flow Graphs versus Control/Data-Dependency Graphs}
Classical AG well-definedness analyses, from Knuth's original proposal~\cite{Knuth68} to ordered attribute grammars~\cite{Kastens80} and their implementation in systems such as Eli~\cite{Gray92} and Silver~\cite{VanWyk10, Kaminski12}, are built around a \textit{flow graph} per production: a graph whose vertices are the attributes occurring on the nonterminals of that production, with an edge between two attributes whenever one is declared, by the production's equations, to depend on the other. Our tool, in contrast, builds a \textit{control/data-dependency graph} (the CFG of \sect{sect:contribution}) from the compiled bytecode of each semantic action, whose vertices are individual read, write, and branch instructions instead of whole attributes.

We do not claim that this finer granularity makes our analysis more precise than the classical ones on AGs specified with equations. \nlgcheck{} does not track attribute values. In the semantic action of Listing~\ref{lst:ternary}, it therefore requires both \texttt{Op[2].val} and \texttt{Op[3].val} to be defined, exactly as a flow-graph-based analysis of the equation \texttt{Op.val = Op[1].val ?\ Op[2].val :\ Op[3].val} would. We are not aware of a realistic AG that can be specified both with equations and with \tool{Neverlang} semantic actions, is well defined for all syntax trees, fails the classical analysis, and passes ours. The control/data-dependency graph is needed, instead, because of how \tool{Neverlang} semantic actions are written and compiled:
\begin{compactenum}
    \item the evaluation order is written by the developer through \texttt{eval}, so well-definedness depends on the order in which reads and writes occur, which a flow graph abstracts away;
    \item attributes are mutable locations that may be written several times (\sect{ssect:complete});
    \item the analysis must run on compiled binaries, after composition, where no equations are available.
\end{compactenum}

Two examples illustrate the first two points. In the first, a semantic action for \texttt{Op} reads \texttt{\$1.val} before evaluating its first child with \texttt{eval \$1}. Read as an equation for \texttt{Op.val}, the action is complete and non-circular, so an analysis of its declared dependencies accepts it, although every execution fails. \nlgcheck{}, instead, reports a $MissingAttr$ violation, since no path leading to the read has written \texttt{Op[1].val}. In the second example, an environment is threaded through a sequence of statements by overwriting a single attribute \texttt{env}. The action of a production \texttt{Stmts <-- Stmt Stmts} writes \texttt{\$1.env} from \texttt{\$0.env} and evaluates \texttt{\$1}. It then copies \texttt{\$1.env}, possibly updated by the action of \texttt{Stmt}, into \texttt{\$2.env}, evaluates \texttt{\$2}, and finally copies \texttt{\$2.env} back into \texttt{\$0.env}. Read as equations, \texttt{env} would be both inherited and synthesized by the same nonterminals, so the specification is not a well-formed AG and a classical system rejects it. \nlgcheck{} accepts it, since on every path each read of \texttt{env} is preceded by a write. The second example shows a difference in representation, not an increase in precision: a specification with equations would use two distinct attributes, such as \texttt{env\_in} and \texttt{env\_out}, and would pass the classical analysis.

We do not provide an empirical performance comparison against flow-graph-based tools such as Eli or Silver. As \sect{ChapEvaluationPerf} already notes, our performance evaluation only compares \nlgcheck{} against its own unoptimized and partially optimized variants, since these systems target a structurally different analysis on different inputs (attribute-level flow graphs derived from AG equations, as opposed to a CFG derived from compiled bytecode) and are not directly comparable on a shared benchmark.

\mypar{The Modularity Paradox: Flexibility vs. Static Safety}
A recurring theme in the literature is the trade-off between the expressiveness of AGs and the ease of their verification.

Extensions such as reference attribute grammars~\cite{Hedin00} and higher-order attribute grammars~\cite{Vogt89} significantly increase the power of AGs by allowing the tree structure to be computed dynamically. However, systems implementing these (e.g., JastAdd) typically enforce strict static typing and require global knowledge of the tree structure to ensure safety~\cite{Ekman07b}.

Neverlang adopts a \textit{feature-oriented} architecture designed for true separate compilation and binary distribution. This requires a loose coupling where attributes are not hard-coded offsets but dynamic entries in a map. While this maximizes modularity—allowing users to swap components without recompilation—it circumvents the safety nets of traditional compilers.

\citet{Mernik05} suggest that the power of domain-specific tools often necessitates a reliance on strict static discipline. Our tool serves as the essential compensatory mechanism for Neverlang’s dynamic nature. By applying static analysis to the generated bytecode, it restores the safety guarantees of a strongly typed system to a highly decoupled, dynamic architecture.

\mypar{Evolution of Modularity: From Composable Attribute Grammars to Forwarding}
The quest for reusability has led to several modular composition techniques:

Early work on composable attribute grammars~\cite{Farrow92} enabled the combination of smaller, independent AGs to solve complex subproblems (e.g., name analysis). Similarly, Generic Attribute Grammars~\cite{Kastens94} focused on decomposing AGs into reusable components. Ganzinger~\cite{Ganzinger83} pursued analogous goals from the compiler-construction side, increasing the modularity and language-independence of automatically generated compilers, while Polyglot~\cite{Nystrom03} offers an extensible compiler framework for Java built around a similar notion of composable language extensions. \textit{Forwarding}~\cite{VanWyk02} provides a mechanism for a node to forward attribute queries to a replacement subtree, aiding in modular extensibility.

Beyond AG-specific systems, general-purpose language workbenches address modular language composition at a coarser granularity. MontiCore~\cite{Krahn10, Butting18, Butting19} composes languages from independent syntactic and semantic features. JetBrains MPS (Meta Programming System)~\cite{Volter11} supports modularizing both a language and its IDE tooling. Melange~\cite{Degueule15} treats language components as reusable, composable meta-language artifacts.\ \citet{Leduc18} pursue a related goal with \textit{Alex}, a non-intrusive approach to modularizing language concerns behind well-defined interfaces, built atop the Eclipse Modeling Framework. All four, like the systems discussed above, compose language definitions before generating the resulting compiler or interpreter, rather than composing already-compiled binary artifacts as \nlgcheck{} does.

While these approaches facilitate design-time modularity, they often struggle with deployment-time modularity (binary distribution). They typically require the entire grammar to be re-analyzed or re-compiled when modules are composed. Our tool is unique in that it validates the interaction of separately compiled binaries by analyzing the underlying data-flow of their semantic actions, ensuring that modular composition does not lead to runtime failures.

\mypar{Data-Flow Analysis and Fixed-Point Iteration}
The methodology of \nlgcheck{} draws from classic compiler optimization and verification techniques~\cite{Aho06}.

Tools like the Extended Static Checker for Java (ESC/Java)~\cite{Flanagan02} demonstrated that data-flow analysis could effectively catch null pointer exceptions and array-out-of-bounds errors in general-purpose code. Our approach applies these principles specifically to the domain of AGs.

Historically, Farrow used fixed-point evaluators to ensure value convergence in circular AGs~\cite{Farrow86}. Our tool repurposes fixed-point principles for state-space bounding. By utilizing the Soot framework~\cite{Vallee-Rai99} to analyze Java bytecode, \nlgcheck{} performs an iterative analysis over the CFG\@. To ensure termination in the presence of recursive semantic actions or complex loops, it implements a recursion depth bound, making the analysis of infinite execution paths computationally tractable.

Loops within the imperative body of a single semantic action (as opposed to recursion across semantic actions, which the depth bound of \sect{sect:contrib:visit} addresses) are handled implicitly by the underlying CFG traversal: since \nlgcheck{} visits the graph rather than executing it, a loop is simply a cycle that the visit algorithm's already-visited-state tracking (Constraint (b), \sect{sect:contrib:visit}) prevents from being traversed indefinitely, in the same way classical data-flow frameworks iterate over loops until a fixed point is reached~\cite{Aho06}, rather than unrolling them.

\section{Conclusions and Future Work}\label{sec:conclusion}
In this work, we presented a novel approach to overcome the limitations of separate compilation in language workbenches. Separate compilation offers several advantages regarding modularity and the distribution of pre-compiled artifacts through public repositories. However, this requirement is often relaxed to comply with other correctness guarantees and non-functional properties~\cite{Leduc20}. With this contribution, we aim to bridge this gap, so that future language workbenches can adopt separate compilation without giving up the detection of attribute errors before runtime. Under separate compilation, undefined and ill-typed attribute accesses are then detected when the components are composed, instead of at the runtime of the compiler, while repairing them relies on the composition mechanisms discussed in \sect{ContribActing}.
Concretely, this paper contributes: (i) a formal model of semantic actions as control/data-dependency graphs, built directly from separately compiled bytecode rather than from source-level specifications; (ii) a static analysis, together with a soundness proof, that detects undefined and ill-typed attribute accesses along every possible execution path, including across recursive compositions of separately compiled language components; (iii) two optimizations (graph and path deduplication, \sect{ContribOptimization}) that make this analysis practical on non-trivial languages; and (iv) \nlgcheck{}, an implementation of this approach for \tool{Neverlang}, together with an empirical evaluation of its error-detection capability and performance on several pre-existing \tool{Neverlang}-based languages. The static analysis approach presented in this paper is sound and effective in practice, and the optimizations mentioned above make it usable on a daily basis during the development process.

Several improvements could expand this research. First, the current analysis is limited to a single AST visit at a time: \nlgcheck{} checks each role's CFG $G_r$ independently, resetting the production context stack between roles (roles combined with the \inlineneverlang{:} operator of \sect{sec:neverlang-in-a-nutshell} are executed in a single traversal and are therefore analyzed together)---in other words, \nlgcheck{}'s current model assumes the checked artifact behaves as a single-pass compiler. This is a hard limitation whenever a later role in a language's \inlineneverlang{roles} clause (\sect{sec:neverlang-in-a-nutshell}) depends on attributes set by an earlier one, as when, for instance, a \inlineneverlang{permissions} role executes after, and relies on, a \inlineneverlang{terminal-evaluation} role: \nlgcheck{} currently cannot verify that an attribute set by an earlier role is available to a later one, as the state from the first traversal is reset before the second begins.

In practice, we expect this limitation to matter less than dependencies within a role. Roles correspond to compiler phases and are coarse grained, and their order is declared explicitly by the \inlineneverlang{roles} clause, so a dependency between two roles follows the order of the phases and is easy to manage without an analysis. The languages evaluated in \sect{sec:evaluation} confirm this. In \textit{logLang}, the \inlineneverlang{terminal-evaluation} role writes the values of terminals (the \texttt{string} and \texttt{identifier} attributes), which the later \inlineneverlang{logging}, \inlineneverlang{permissions}, and \inlineneverlang{execution} roles read. In \textit{stateMachines}, the \texttt{id} attribute written by the \inlineneverlang{evaluation} role is read by the later \inlineneverlang{register} role. In \textit{numericUnitExpr} and \textit{neverlangJS}, no attribute written by one role is read by another: in \textit{numericUnitExpr}, for instance, the \inlineneverlang{unit-checking} role only uses the \texttt{label} attribute, while the \inlineneverlang{evaluation} role only uses \texttt{value}. None of these languages exhibited a runtime error at a role boundary, so the limitation remains, for now, a theoretical rather than an observed one.

Other minor improvements include better error reporting, such as warnings for issues currently undetected by \nlgcheck{}. In this regard, so-called ``non-idiomatic'' attribute accesses should either trigger a warning or, even better, be included as part of the analysis. Finally, this approach should be tested with other language workbenches to verify the degree to which the solutions proposed in this work are limited to \tool{Neverlang} and its ecosystem.

\begin{acks}
   A special thanks to Gianluca Nitti for his support in the development of \nlgcheck{}.
\end{acks}

\bibliographystyle{ACM-Reference-Format}
\bibliography{local,strings,metrics,compilers,programming,foundations,software_engineering,tools,software_architecture,dsl,pl,splc,oolanguages,my_work,grammars,security,roles,learning,cop,testing,dsu,distributed_systems,reflection,aosd,petri-nets,miscellaneous,pattern}

\appendix
\section{Simplified Verification Model: Postorder Semantic Actions}\label{ssect:postorder}

\begin{algorithm}
\caption{Verification of attribute accesses in postorder semantic actions
execution}\label{AlgVerifyPostorder}
\begin{algorithmic}
  \State{} $E := \emptyset$
\ForAll{$p \in P$}
  \For{$i := 1$ \textbf{to} $|p|$}
    \ForAll{$p' \in P$}
      \If{$p'[0] = p[i]$}
   \ForAll{$a \in p[i]$}
    \If{$\tau(p'[0].a)$ is not defined}
    \State{} $E := E \cup$ \{``$p'$ does not provide $a$ required by $p[i]$''\}
    \ElsIf{\textbf{not} $\tau(p'[0].a) \le: \tau(p[i].a)$}
            \State{} $E := E \cup$ \{``type of provided $p'[0].a$ is incompatible
      with $p[i].a$''\}
          \EndIf%
        \EndFor%
      \EndIf%
    \EndFor%
  \EndFor%
\EndFor%
\If{$E = \emptyset$}
  \State{} verification succeeded
\Else%
  \State{} incorrect attribute accesses have been detected and collected in $E$
\EndIf%
\end{algorithmic}
\end{algorithm}

Before facing the complete verification model presented in this paper, we attempted solving the problem with a simplified model. We hereby show the limitations that ultimately lead us opting for a more complex solution.
Let us consider the problem of verifying the correctness of attributes within a single semantic role in which the tree visit is performed in postorder: for each node, each of its child nodes is visited first by executing the semantic action associated with the production which generated the child node, then the same operation is repeated on the parent node itself. With this visit strategy, data can only be passed bottom-up in the tree (from child to parent node).
Given this premise a semantic action attached to a production $p: X_0 \leftarrow X_1 \cdots X_n $ will read attributes from $X_1 \cdots X_n$, and then write attributes to $X_0$.
Although more complex behaviors are possible, we will consider semantic actions in which there is no reason to:
\begin{compactitem}
    \item read attributes on $X_0$, as no semantic action which could have written them has been executed yet;
    \item write attributes to $X_1 \cdots X_n$ as no further semantic action will read them.
\end{compactitem}

Given these restrictions, we can define that $\tau(p[0].a) = t$ if the semantic action attached to $p$ \textit{provides} an attribute $a$ of type $t$, and $\tau(p[i].a) = t$ with $1 \le i \le |p|$ if the semantic action attached to $p$ \textit{requires} an attribute $a$ of a type $t'$ assignable from $t$ ($t' \le: t$) on $X_i$.

Now, we can define the concept of \textit{safe attribute access} in the scope of postorder visits.
Let us consider a production $p: X_0 \leftarrow X_1 \cdots X_n$ and a semantic action attached to it which accesses the attribute $p[i].a$ (with $1 \le i \le n$), and with type requirement $t = \tau(p[i].a)$.
We say this access is type safe iff.
\[\forall p' \in P:\;X_i \leftarrow Y_1 \cdots Y_n,\quad\tau(p'[0].a) \le: t\]
or, in other words, if for each production $p'$ with $X_i$ as its left-hand side, an attribute with the same name is provided by the semantic action attached to $p'$ on its left-hand side nonterminal with the same type or a subtype of $t$.

This lead to a first, albeit limited, algorithm to verify the safety of attribute accesses, that works on semantic actions within a visit executed in postorder: the full pseudocode is shown in \Cref{AlgVerifyPostorder}.

\subsection{Implementation Aspects and Limitations}\label{FirstApproachLimitations}
An implementation of~\Cref{AlgVerifyPostorder} needs a way to gather information about attributes required and provided by the semantic action attached to each production for the role being analyzed.
For \tool{Neverlang}, there are two stages of the build process in which it is practical to access this information: during the DSL compilation, while \texttt{Neverlang} is performing a translation from the semantic action implementation DSL to the backend code, or after the entire compilation process by inspecting the bytecode through a bytecode analysis tool of choice. Both have been implemented in \tool{Neverlang}. The former adds helper methods to the semantic action classes, returning information about required and provided attributes, in the form of index-name-type triplets; \nlgcheck{} queries these methods to implement \Cref{AlgVerifyPostorder}.

It is important to note that this text-based approach has limitations due to it lacking any type information which is later generated by the \tool{Java} compiler, and thus it requires additional syntax to properly generate the correct methods---\textit{e.g.} \texttt{\$0.value: Integer = \ldots} denotes the semantic action generating an attribute of \inlinejava{Integer} type.
However, usage of this syntax extensions when writing semantic actions is optional, and when the developer does not provide any type hints, the attribute type is assumed to be \texttt{Object}. As we will discuss later in this section, a more ``hands on'' approach, inspecting the generated semantic actions bytecode directly, is preferable in general.

This approach represents a solution that was used in the early stages of this work and was successfully tested against simple languages such as the \tool{Java} expression language~\cite{Cazzola21b}, modified by adding the necessary type annotations. This allows the language to be verified successfully, thus introducing errors in the form of incorrect attribute names and type casts or annotations.

However, such an approach has several limitations that ultimately caused us to deem it inadequate to be used in more concrete scenarios.
First, the approach relies on the presence of type hints that the language developers need to manually write in the semantic actions: this requires modification and recompilation of existing modules to make them properly verifiable, and makes the code less readable due to redundant information (the \tool{Java} compiler is aware of the type of each expression without needing type hints).
This specific limitation does not regard the algorithm itself but rather its naive original implementation.
A much more pressing matter regards the way information about required and provided attributes is extracted from semantic actions: the mere textual presence of an expression accessing an attribute is an imprecise indicator of these events actually happening at runtime. If the statement is, for example, inside a conditional block, its actual execution depends on runtime value of the condition, an aspect that cannot be captured by a simplified textual analysis.
Extracting the information regarding required and provided attributes from the bytecode enables the possibility of considering all possible execution paths. More specifically, if different paths provide the same attribute with different types, the resulting type of the attribute can be inferred as the closest common supertype. Similarly, the presence of one branch in which the attribute is not may cause indicate that the attribute is not provided in specific execution scenarios, even though the textual representation indicates otherwise. Finally, if different paths require the same attribute, but with different types, the closest common subtype must be considered.
It is arguable that some of these scenarios are unlikely or may suggest improper language implementation, however please recall that this contribution aims at verifying that the composition among \textbf{precompiled} language features is feasible, even when they where not originally intended to be used jointly; therefore we believe unlikely scenarios should be taken into consideration.

Last but not least, limiting the solution to postorder visits is not a fair assumption: in any non-trivial language, semantic actions need full control with regards of how the visit must be performed. As a notable example, please consider a \texttt{while} statement in an interpreter for a GPL\@: a strictly postorder (or preorder) visit is not viable, because:
\begin{compactitem}
    \item the node representing the condition must checked (and thus visited) on each iteration;
    \item the node representing the body may or may not be visited depending on the runtime value of the condition.
\end{compactitem}

\section{Algorithms}
This appendix contains any algorithms that were omitted from the main body of this manuscript for brevity.

\subsection{Full CFG Construction Algorithm}\label{sec:app:algofull}
\Cref{AlgBuildRoleCFG} shows the entire CFG construction algorithm for the complete model not limited to postorder visit.

For the sake of conciseness, this algorithm omits an implementation aspect involving semantic action guards: if multiple semantic actions with different guards are attached to the same production, their respective CFGs are merged in a single CFG representing all the possible branches which can be taken before running \Cref{AlgBuildRoleCFG}.
In case of type guards, edges connecting the entry node of the production CFGs to the entry nodes of the individual semantic actions are tagged with the required attribute types: this allows the visit algorithm to exclude paths which would never be followed at runtime because of the guards.
Please refer to~\cite{Cazzola16e} for a full examination of guards and type guards in \tool{Neverlang}.

\begin{algorithm}
\caption{Construction of a role's CFG\@. $G_i=(V_i, E_i)$ denotes the CFG of the
semantic action attached to the production $p_i$.}\label{AlgBuildRoleCFG}
\begin{algorithmic}
\State{} create two nodes $v_{begin}, v_{end}$
\State$a_{node}(v_{begin}) := a_{node}(v_{end}) := Nop()$
\State$G_r=(V_r,E_r) := (\{v_{begin}, v_{end}\}, \emptyset)$
\State$v_{entry}(G_r) := v_{begin}$
\State$v_{exit}(G_r) := v_{end}$
\ForAll{$p_i \in P$}
  \State$V_r := V_r \cup V_i$
  \State$E_r := E_r \cup E_i$
  \If{$p_i[0] = Program$}
    \State$E_r := E_r \cup \{(v_{begin}, v_{entry}(G_i)), (v_{exit}(G_i), v_{end})\}$
  \EndIf%
\EndFor%
\State$E_{entry} := \emptyset$
\State$E_{exit} := \emptyset$
\State$L := \emptyset$
\ForAll{$p_i \in P$}
  \ForAll{$v_{beginEval} \in V_i: \exists k: a_{node}(v_{beginEval}) = BeginEvalMetaAction(k)$}
  \State{} find $v_{endEval} \in V_i$ such that $(v_{beginEval}, v_{endEval}) \in E_i$
    \State$E_r := E_r \setminus \{(v_{beginEval}, v_{endEval})\}$
    \ForAll{$p_j \in P: p_j[0] = p_i[k]$}
    \State{} generate an unique tag $l \notin L$
      \State$L := L \cup \{l\}$
      \State$e_{entry} := (v_{beginEval}, v_{entry}(G_j))$
      \State$e_{exit} := (v_{exit}(G_j), v_{endEval})$
      \State$E_{entry} := E_{entry} \cup \{e_{entry}\}$
      \State$E_{exit} := E_{exit} \cup \{e_{exit}\}$
      \State$l_{edge}(e_{entry}) := l_{edge}(e_{exit}) := l$
    \EndFor%
    \State$a_{node}(v_{beginEval}) := EnterCtx(k)$
    \State$a_{node}(v_{endEval}) := LeaveCtx()$
  \EndFor%
\EndFor%
\State$E_r := E_r \cup E_{entry} \cup E_{exit}$
\end{algorithmic}
\end{algorithm}

\subsection{Optimized CFG Visit Algorithm}\label{sec:app:algoopt}
\Cref{AlgVisitRoleCFG} shows the CFG visit algorithm, including the optimizations introduced in \sect{ContribOptimization}.

\begin{algorithm}
\caption{Optimized visit of a role's CFG, with cache and jump logic (replaces \Cref{AlgVisitRoleCFG})}\label{AlgVisitRoleCFGOpt}
\begin{algorithmic}
\State$\Phi := \{(v_{entry}(G_r), v_{entry}(G_r), \epsilon, \langle\emptyset\rangle, \epsilon)\}$
\Comment{Set of states to visit}
\State$\Gamma := \emptyset$\Comment{Set of visited states}
\State$U := \emptyset$\Comment{Found violations}
\State$f_c := \emptyset$\Comment{Cache for jumps}
\While{$\Phi \neq \emptyset$}
\State{} pick any $(v, w, S, C, S_c) \in \Phi$ and remove it from $\Phi$
  \If{$(v, C, S) \notin\Gamma$} \Comment{Path condition~\eqref{pathcond2}}
    \State$\Gamma := \Gamma \cup (v, C, S)$
    \State$(C', U') := f_{action}(C, a_{node}(v))$
    \State$U := U \cup U'$
    \ForAll{$(v, v') \in E_r$}
      \State$S' := f_{edge}(S, (v, v'))$
      \If{$\exists S'$ \textbf{and} $\nexists l_1,l_2,l_3 \in S: l_1=l_2=l_3$}
        \Comment{Path validity and~\eqref{pathcond1}}
        \If{$(v, v') \in E_{entry}$}\Comment{Then $v'$ is a nonterminal entry}
          \If{$\exists f_c(v', head(C))$}
            \State$(v_{out}, C_{outs}, C_{taken}) := f_c(v', head(C))$
            \ForAll{$C_{out} \in C_{outs}$}
              \State$\Phi := \Phi \cup \{(v_{out}, wv', S', merge_h(C, C_{out}), push(S_c, head(C)))\}$
            \EndFor%
            \State$f_c(v', head(C)) := (v_{out}, C_{outs}, C_{taken} \cup C)$
          \Else%
            \State$\Phi := \Phi \cup \{(v', wv', S', C', push(S_c, head(C)))\}$
          \EndIf%
        \ElsIf{$(v, v') \in E_{exit}$}\Comment{Then $v$ is a nonterminal exit}
        \State{} find $v_{in}$ such that $\exists v_x:\:l_{edge}((v_x, v_{in}))=l_{edge}(v, v')$
          \If{$\exists f_c(v_{in}, top(S_c))$}
            \State$(v_{out}, C_{outs}, C_{taken}) := f_c(v', head(C))$
            \If{$head(C) \notin C_{outs}$}
              \ForAll{$C_t \in C_{taken}$}
                \State$\Phi := \Phi \cup \{(v', wv', S', merge_h(C_t, head(C)))\}$
              \EndFor%
            \EndIf%
            \State$f_c(v', head(C)) := (v_{out}, C_{outs} \cup \{head(C)\}, C_{taken})$
          \Else%
            \State$f_c(v_{in}, top(S_c)) := (v, \{head(C)\}, \emptyset)$
          \EndIf%
          \State$\Phi := \Phi \cup \{(v', wv', S', C', pop(S_c))\}$
        \Else%
          \State$\Phi := \Phi \cup \{(v', wv', S', C', S_c)\}$
        \EndIf%
      \EndIf%
    \EndFor%
  \EndIf%
\EndWhile%
\end{algorithmic}
\end{algorithm}

\section{Proofs}\label{sec:app:proofs}
This appendix contains any proofs that were omitted from the main body of this manuscript for brevity.

\setcounter{property}{0}

\begin{property}\label{PropertyPieceWise}
$f_{pCtx}$ can be evaluated ``piecewise'' on a path by composing evaluations of $f_{pCtx}$ on parts of the path:
    \[
        \forall w_1, w_2 \in V_r^*, C \in \mathcal{C}, \quad
        f_{pCtx}(w_1w_2, C)=f_{pCtx}(w_2, f_{pCtx}(w_1, C))\]
\end{property}
\begin{propertyproof}
This can be proven by induction on $w_1$.
The base case $w_1=\epsilon$ is trivial, as by definition $f_{pCtx}(\epsilon, C) = C$ and the goal is reduced to an identity.
For the inductive step, given $w_1$, we assume
\begin{equation} \label{PropfpctxCompositionIndHyp}
    \forall w_2 \in V_r^*, C' \in \mathcal{C}, \quad f_{pCtx}(w_1w_2, C')=f_{pCtx}(w_2, f_{pCtx}(w_1, C'))
\end{equation}
and we must prove that
    \[
        \forall v \in V_r, w_2 \in V_r^*, C \in \mathcal{C}, \quad f_{pCtx}(vw_1w_2, C)=f_{pCtx}(w_2, f_{pCtx}(vw_1, C))
    \]
By definition, \[f_{pCtx}(vw_1w_2, C)=f_{pCtx}(w_1w_2, f_{aCtx}(C, a_{node}(v)))=\]
applying the inductive hypothesis~\eqref{PropfpctxCompositionIndHyp} with $C' = f_{aCtx}(C, a_{node}(v))$,
\[=f_{pCtx}(w_2, f_{pCtx}(w_1, f_{aCtx}(C, a_{node}(v))))=\]
and finally, applying once more the definition of $f_{pCtx}$, this time in
``reverse'' for the $vw_1$ path,
\[=f_{pCtx}(w_2, f_{pCtx}(vw_1, C))\]
thus proving the property.
\end{propertyproof}

\setcounter{property}{5}

\begin{property}
    Given an integer $k>0$ and two production context stacks sharing the same top $k$ items, the evaluation context stacks obtained by evaluating a balanced path prefix which does not contain $CopyAttr$ nor $PropagateAttrs$ actions share the top $k$ items:
    \begin{gather*}
        \forall w \in W_{pb}, \forall C_1, C_2 \in \mathcal{C}, k \in \mathbb{N}\\
        top_k(C_1)=top_k(C_2) \implies top_k(f_{pCtx}(w, C_1))=top_k(f_{pCtx}(w, C_2))
    \end{gather*}
\end{property}
\begin{propertyproof}
Due to the hypothesis $w \in W_{pb}$, the action sequence associated with $w$ is either one of the following cases.

\begin{compactitem}
    \item If $w=\epsilon$, the implication is trivial by definition of $f_{pCtx}$, as $\forall C, \: f_{pCtx}(\epsilon, C) = C$.
    \item If $w=vw'$, with $w' \in W_{pb}$ and $a_{node}(v) = a$ with $a \in A \setminus \{LeaveCtx()\}$ then the implication is proved by induction, assuming
    \begin{equation}\label{PropBalancedPath0IndHyp}
        \forall C_1, C_2, k, \quad top_k(C_1)=top_k(C_2) \implies top_k(f_{pCtx}(w', C_1))=top_k(f_{pCtx}(w', C_2))
    \end{equation}
    for a fixed $w'$ and proving that
    \[
        \forall C_1, C_2, k, \quad top_k(C_1)=top_k(C_2) \implies top_k(f_{pCtx}(vw',C_1))=top_k(f_{pCtx}(vw',C_2))
    \]
    Assuming $top(C_1)=top(C_2)$ and by definition of $f_{pCtx}$:
    \begin{equation}\label{PropBalancedPath0IndGoal}
        top_k(f_{pCtx}(w', f_{aCtx}(C_1, a))) = top_k(f_{pCtx}(w', f_{aCtx}(C_2, a)))
    \end{equation}
    For $a \in \left\{ ReadAttr, CheckAttrType, Nop \right\}$, it is true that $f_{aCtx}(C', a) = C'\:\forall C'$, thus by substituting  $f_{aCtx}(C_1, a) = C_1$ and $f_{aCtx}(C_2, a) = C_2$ we obtain the inductive hypothesis~\eqref{PropBalancedPath0IndHyp}.

    Otherwise either $a=WriteAttr(i, \eta, t)$ or $a=EnterCtx(i)$.

    In the former case, applying the definition of $f_{action}$ for $WriteAttr$,~\eqref{PropBalancedPath0IndGoal} becomes
    \[
        top_k(f_{pCtx}(w',put(C_1,i,\eta,t)))=top_k(f_{pCtx}(w',put(C_2,i,\eta,t)))
    \]
    Then, applying the definition of $put$ on $C_1=\langle c_1, \ldots, c_n \rangle$ and $C_2=\langle c_1', \ldots, c_n' \rangle$:
    \begin{gather*}
        top_k(f_{pCtx}(w',\langle c_1,\ldots,c_{put1}(c_{n-1},k_n,i,\eta,t),c_{put2}(c_n,i,\eta,t)\rangle))=\\
        top_k(f_{pCtx}(w',\langle c_1',\ldots,c_{put1}(c_{n-1}',k_n',i,\eta,t),c_{put2}(c_n',i,\eta,t)\rangle))
    \end{gather*}
    Since $top_k(C_1)=top_k(C_2)$, $\forall i \in \mathbb{N}: \: k \le i \le n, \quad c_i=c_i'$, therefore, the arguments of the $f_{pCtx}(w', \ldots)$ in the two members of~\eqref{PropBalancedPath0IndGoal} share the same $top_k$, proving the inductive step.

    Finally, when $a=EnterCtx(i)$,~\eqref{PropBalancedPath0IndGoal} can be rewritten as
    \[
        top_k(f_{pCtx}(w', push(C_1, (i, \emptyset)))) = top_k(f_{pCtx}(w', push(C_2, (i, \emptyset))))
    \]
    By the definitions of $push$ and $top_k$, and the hypothesis $top_k(C_1) = top_k(C_2)$, it follows that
    \[
        top_{k+1}(push(C_1, (i, \emptyset))) = top_{k+1}(push(C_2, (i, \emptyset)))
    \]
    and, by applying the inductive hypothesis~\eqref{PropBalancedPath0IndHyp} to $push(C_1, (i, \emptyset))$ and $push(C_2, (i, \emptyset))$:
    \[
        top_{k+1}(f_{pCtx}(w', push(C_1, (i, \emptyset)))) = top_{k+1}(f_{pCtx}(w', push(C_2, (i, \emptyset))))
    \]
    The equality of $top_{k+1}$ trivially implies the equality of $top_k$, thus proving~\eqref{PropBalancedPath0IndGoal}.

    \item If $w=w'v$, the proof is similar. The same inductive hypothesis holds.
        According to the definition of $f_{pCtx}$ and \Cref{PropfpctxComposition} the goal $top_k(C_1) = top_k(C_2)$ can be rewritten as:
    \begin{gather*}
        top_k(f_{pCtx}(w'v, C_1))=top_k(f_{pCtx}(w'v, C_2))\\
        top_k(f_{aCtx}(f_{pCtx}(w', C_1), a))=top_k(f_{aCtx}(f_{pCtx}(w', C_2), a))
    \end{gather*}
    Let $C_3=f_{pCtx}(w', C_1)$ and $C_4=f_{pCtx}(w', C_2)$; from the inductive hypothesis~\eqref{PropBalancedPath0IndHyp} follows that $top_k(C_3)=top_k(C_4)$. As seen in the previous case, only $WriteAttr$ and $EnterCtx$ can modify the context stack and both actions preserve the equality of the $top_k$, thus from the definition of $f_{action}$ follows that
    \[top_k(f_{aCtx}(C_3, a)) = top_k(f_{aCtx}(C_4, a))\]
    \item If $w=v_1w'v_2$ with $w' \in W_{pb}$, $a_{node}(v_1)=a_1=EnterCtx(i)$, and $a_{node}(v_2)=a_2=LeaveCtx()$. The inductive hypothesis~\eqref{PropBalancedPath0IndHyp} still holds and the goal is to prove that
    \[top_k(f_{pCtx}(v_1w'v_2, C_1))=top_k(f_{pCtx}(v_1w'v_2, C_2))\]
    for any two $C_1=\langle c_0, \ldots, c_n \rangle$ and $C_2=\langle c_0', \ldots, c_n'\rangle$ such that $top_k(C_1)=top_k(C_2)$. This goal can be rewritten as:
    \begin{gather*}
        top_k(f_{aCtx}(f_{pCtx}(w', f_{aCtx}(C_1, a_1)), a_2))= top_k(f_{aCtx}(f_{pCtx}(w', f_{aCtx}(C_2, a_1)), a_2))\\
        top_k(pop(f_{pCtx}(w', push(C_1, (i, \emptyset)))))= top_k(pop(f_{pCtx}(w', push(C_2, (i, \emptyset)))))\\
        top_k(pop(f_{pCtx}(w', \langle c_0, \ldots, c_n, (i, \emptyset)\rangle)))= top_k(pop(f_{pCtx}(w', \langle c_0', \ldots, c_n', (i, \emptyset)\rangle)))
    \end{gather*}
    Since $top_k(C_1) = top_k(C_2)$, $c_i=c_i'$ $\forall i \in \left\{k, \dots, n \right\}$, and thus
    \[
        top_{k+1}(\langle c_0, \ldots, c_n, (i, \emptyset)\rangle) = top_{k+1}(\langle c_0', \ldots, c_n', (i, \emptyset)\rangle)
    \]
    This satisfies the premise of the inductive hypothesis~\eqref{PropBalancedPath0IndHyp} (notice that it holds for any $k$, and $k+1$ in particular), therefore
    \[
        top_{k+1}(f_{pCtx}(\langle c_0, \ldots, c_n, (i, \emptyset)\rangle))= top_{k+1}(f_{pCtx}(\langle c_0', \ldots, c_n', (i, \emptyset)\rangle))
    \]
    Applying $pop$ to each of the two stacks removes the topmost item preserving the rest; therefore even though the equality no longer holds for $top_{k+1}$, it still holds for $top_k$, thus proving the goal.
\end{compactitem}
\end{propertyproof}

\setcounter{property}{7}

\begin{property}
Given any two production context stacks, evaluating a balanced path prefix which does not contain $CopyAttr$ or $PropagateAttrs$ actions, wrapped between two $EnterCtx$ nodes produces context stacks with the same head.
    \[
        \forall C_1, C_2 \in \mathcal{C}, w \in W_{pb}, \quad v_{sep} \in V_r^*:\: \exists i: a_{node}(v_{sep}) = EnterCtx(i),
    \]
    \begin{equation}\label{PropBalancedPath1Goal}
        head(f_{pCtx}(v_{sep}wv_{sep}, C_1))=head(f_{pCtx}(v_{sep}wv_{sep}, C_2))
    \end{equation}
\end{property}
\begin{propertyproof}
    The lemma 
    is proven by proving each of the following cases.

    \begin{compactitem}
        \item If $w = \epsilon$, then the goal~\eqref{PropBalancedPath1Goal} becomes
            \[head(f_{pCtx}(v_{sep}v_{sep}, C_1))=head(f_{pCtx}(v_{sep}v_{sep}, C_2))\]
            By definition of $f_{pCtx}$:
            \[
                head(f_{aCtx}(f_{aCtx}(C_1, a_s), a_s)) = head(f_{aCtx}(f_{aCtx}(C_2, a_s), a_s)) \quad\text{ with } a_s=EnterCtx(i)
            \]
            By definition of $f_{aCtx}$ and $f_{action}$ for $EnterCtx$,\footnote{Notice that for brevity, we will keep using the equivalence $a_s=EnterCtx(i)$ throughout this proof.}
            \[
                head(push(push(C_1, (i, \emptyset)), (i, \emptyset))) =
                head(push(push(C_2, (i, \emptyset)), (i, \emptyset)))
            \]
            Let $C_1=\langle c_0, \ldots, c_n\rangle$ and $C_1=\langle c_0', \ldots, c_n' \rangle$; by definition of $push$,
            \[
                head(\langle c_0,\ldots c_n,(i, \emptyset),(i, \emptyset)\rangle) =
                head(\langle c_0',\ldots c_n',(i, \emptyset),(i, \emptyset)\rangle) =
            \]
            which is proven by \Cref{PropHeadDefImmediate}.

        \item If $w =vw'$, with $w' \in W_{pb}$ and $a_{node}(v) = a$ with $a \in A \setminus \{LeaveCtx()\}$ then we must prove that, $\forall C_1, C_2 \in \mathcal{C}$,
            \begin{equation}\label{PropBalancedPath1IndGoal}
                head(f_{pCtx}(v_{sep}vw'v_{sep}, C_1))=head(f_{pCtx}(v_{sep}vw'v_{sep}, C_2))
            \end{equation}
            Let's expand both sides of the equality using a generic $C$ instead of $C_1$ and $C_2$ by applying \Cref{PropfpctxComposition}:
            \begin{gather*}
                head(f_{pCtx}(v_{sep}vw'v_{sep}, C))=\\
                head(f_{pCtx}(vw'v_{sep}, f_{pCtx}(v_{sep}, C)))=\\
                head(f_{pCtx}(v_{sep}, f_{pCtx}(vw', f_{pCtx}(v_{sep}, C))))=
            \end{gather*}
            by definition of $f_{pCtx}$ and then of $f_{aCtx}$ for $EnterCtx(i)$,
            \begin{gather*}
                =head(f_{aCtx}(f_{pCtx}(vw', f_{aCtx}(C, a_s)), a_s))=\\
                head(push(f_{pCtx}(vw', push(C, (i, \emptyset))), (i, \emptyset)))
            \end{gather*}
            The goal~\eqref{PropBalancedPath1IndGoal} can therefore be rewritten as
            \begin{gather*}
                head(push(f_{pCtx}(vw', push(C_1, (i, \emptyset))), (i, \emptyset)))=\\
                head(push(f_{pCtx}(vw', push(C_2, (i, \emptyset))), (i, \emptyset)))
            \end{gather*}
            This can be proven through \Cref{PropHeadDefImmediate}, showing that the last two items of the context stacks passed to $head$ in the two members of the equality are equal. By definition of $push$, the last item is $(i, \emptyset)$ at both members; also by definition of $push$ we know the second-last item is the $top$ of the argument of $push$, therefore we must prove that
            \[
                top(f_{pCtx}(vw', push(C_1, (i, \emptyset))))= top(f_{pCtx}(vw', push(C_2, (i, \emptyset))))
            \]
            By definition, $vw'$ is a balanced path prefix because it is formed by the concatenation of a balanced path prefix $w'$ and a node $v$ such that $a_{node}(v) \neq LeaveCtx()$. Moreover, $top(push(C_1, (i, \emptyset))) = top(push(C_2, (i, \emptyset)))$ by definition of $top$ and $push$. Therefore, the goal is proven by applying \Cref{PropBalancedPath0} with $k=1$, path $vw'$ and context stacks $push(C_1, (i, \emptyset))$ and $push(C_2, (i, \emptyset))$.

        \item If $w=w'v$, the proof is identical to the previous case, as the considerations made on $vw'$ can be made on $w'v$ as well.

        \item If $w=v_1w'v_2$ with $w' \in W_{pb}$, $a_{node}(v_1)=a_1=EnterCtx(j)$, and $a_{node}(v_2)=a_2=LeaveCtx()$:
            \begin{equation}\label{PropBalancedPath1IndGoal4}
                head(f_{pCtx}(v_{sep}v_1w'v_2v_{sep}, C_1)) =
                head(f_{pCtx}(v_{sep}v_1w'v_2v_{sep}, C_2))
            \end{equation}
            Let us expand both members of the equality, using $C$ in place of $C_1$ and $C_2$ respectively:
            \begin{gather*}
                head(f_{pCtx}(v_{sep}v_1w'v_2v_{sep}, C))=\\
                head(f_{aCtx}(f_{aCtx}(f_{pCtx}(w', f_{aCtx}(f_{aCtx}(C, a_s), a_1)), a_2), a_s))=
            \end{gather*}
            recalling that $a_s=EnterCtx(i)$ and thus $f_{aCtx}(C',a_s)=push(C',(i,\emptyset))$,
            \[
                =head(push(f_{aCtx}(f_{pCtx}(w', f_{aCtx}(push(C, (i, \emptyset)), a_1)), a_2),(i,\emptyset)))=
            \]
            recalling $a_1=EnterCtx(j)$ and $a_2=LeaveCtx()$, thus $f_{aCtx}(C',a_1)=push(C',(j,\emptyset))$ and $f_{aCtx}(C',a_2)=pop(C')$,
            \[
                =head(push(pop(f_{pCtx}(w', push(push(C, (i, \emptyset)), (j, \emptyset)))),(i,\emptyset)))
            \]
            With $C_1=\langle c_0, \ldots, c_n \rangle$ and $C_2=\langle c_0', \ldots, c_n' \rangle$, the goal~\eqref{PropBalancedPath1IndGoal4} can be rewritten as
            \begin{gather*}
                head(push(pop(f_{pCtx}(w', \langle c_0, \ldots, c_n, (i, \emptyset), (j, \emptyset)\rangle)),(i,\emptyset)))=\\
                head(push(pop(f_{pCtx}(w', \langle c_0', \ldots, c_n', (i, \emptyset), (j, \emptyset)\rangle)),(i,\emptyset)))
            \end{gather*}
            Let $C_3=\langle c_0, \ldots, c_n, (i, \emptyset), (j, \emptyset)\rangle$ and $C_4=\langle c_0', \ldots, c_n', (i, \emptyset), (j, \emptyset)\rangle$. Then, by definition of $top_k$, $top_2(C_3)=top_2(C_4)$. Applying \Cref{PropBalancedPath0}:
            \[
                top_2(f_{pCtx}(w', C_3))=top_2(f_{pCtx}(w', C_4))
            \]
            Thus, we must prove that
            \[
                head(push(pop(f_{pCtx}(w', C_3)), (i, \emptyset))) =
                head(push(pop(f_{pCtx}(w', C_4)), (i, \emptyset)))
            \]
            Applying $pop$ to $f_{pCtx}(w', C_3)$ and to $f_{pCtx}(w', C_4)$ produces, respectively, two stacks $C_5$ and $C_6$ such that $top(C_5)=top(C_6)$, reducing the goal to
            \[
                head(push(C_5, (i, \emptyset))) = head(push(C_6, (i, \emptyset)))
            \]
            which is trivially proven by definition of $push$ and $head$ and by the equality $top(C_5)=top(C_6)$.
    \end{compactitem}
\end{propertyproof}

\end{document}